\documentclass[sigconf]{acmart}

\usepackage{graphicx} 
\usepackage{balance}

\usepackage[mathscr]{eucal}

\usepackage{booktabs}
\usepackage{amsmath, amsfonts} 
\usepackage[linesnumbered,ruled,algonl,vlined,noend]{algorithm2e}
\usepackage{epsfig}
\usepackage{enumitem}
\usepackage{tabularx}
\usepackage{subfig}
\usepackage{multirow}
\usepackage{verbatim}
\usepackage{float}
\usepackage{framed}
\usepackage{xcolor, soul}
\usepackage[export]{adjustbox}

\colorlet{shadecolor}{gray!20}

\newcommand{\todo}[1]{{\footnotesize \textcolor{blue}{$\ll$\textsf{TODO #1}$\gg$}}}

\newcommand{\lgl}[1]{{\textcolor{red}{?????LGL:#1}?????}}

\newtheorem{defn}{Definition}[section]
\newtheorem{thm}{Theorem}[section]
\newtheorem{lem}[thm]{Lemma}

\newcommand{\var}[1]{\mbox{\emph{#1}}}
\newcommand{\svar}[1]{\mbox{\scriptsize\emph{#1}}}
\newcommand{\bao}[1]{\textrm{\textcolor{orange}{!!!Bao says: #1!!!}}}

\newcommand{\ping}[1]{\textrm{\textcolor{black}{#1}}}
\newcommand{\yuchen}[1]{\textrm{\textcolor{red}{Yuchen says: #1}}}

\newcommand{\budget}{\ensuremath{\var{L}}}                   
\newcommand{\td}{\ensuremath{\mathcal{T}}}   		 	    	
\newcommand{\std}{\ensuremath{\svar{T}}}	                    
\newcommand{\tr}{\ensuremath{t}}   		 	     		
\newcommand{\str}{\svar{t}}   		 	     					
\newcommand{\p}{\ensuremath{\var{p}}}           		 		

\newcommand{\bb}{\ensuremath{\var{b}}}              		
\newcommand{\sbb}{\svar{b}}               				
\newcommand{\bbl}{\ensuremath{\var{loc}}}             		
\newcommand{\cost}{\ensuremath{\var{w}}}                  

\newcommand{\size}{\ensuremath{\var{size}}}                  
\newcommand{\pr}{\ensuremath{\var{pr}}}   					

\newcommand{\ur}{\ensuremath{U}}  

\newcommand{\opt}{\ensuremath{OPT}}                   			

\newcommand{\ifl}{I}                                       

\newcommand{\region}{\ensuremath{\var{R}}}                   

\newcommand{\ratio}{\vartheta}                   

\newcommand{\overlap}{\varPsi}                 

\newcommand{\marginifl}{\Delta}     

\newcommand{\gimatrix}{\mathbb{I}}   

\newcommand{\gsmatrix}{\gimatrix_s}   

\newcommand{\limatrix}{\xi}   

\newcommand{\lsmatrix}{\xi_s}   

\newcommand{\upperlimatrix}{\xi^\uparrow} 

\newcommand{\indx}{\mathcal{H}}

\newcommand{\problem}{TIP\xspace}

\newcommand{\ngre}{NaiveGreedy\xspace}
\newcommand{\gre}{GreedySel\xspace}
\newcommand{\enumgreedy}{EnumSel\xspace}
\newcommand{\psel}{PartSel\xspace}
\newcommand{\bbsel}{LazyProbe\xspace}
\newcommand{\topk}{TrafficVol\xspace}
\newcommand{\ann}{Annealing\xspace}
\newcommand{\estbod}{EstimateBound\xspace}

\begin{document}

\title{Trajectory-driven Influential Billboard Placement}
\sloppy

\author{Ping Zhang}\affiliation{\institution{Wuhan Univeristy}}
\email{pingzhang@whu.edu.cn}

\author{Zhifeng Bao}\affiliation{\institution{RMIT University}}
\email{zhifeng.bao@rmit.edu.au}

\author{Yuchen Li}\affiliation{\institution{Singapore Management University}}
\email{yuchenli@smu.edu.sg}

\author{Guoliang Li}\affiliation{\institution{Tsinghua University}}
\email{liguoliang@tsinghua.edu.cn}

\author{Yipeng Zhang}\affiliation{\institution{RMIT University}}
\email{s3582779@student.rmit.edu.au}

\author{Zhiyong Peng}\affiliation{\institution{Wuhan Univeristy}}
\email{peng@whu.edu.cn}

\setcopyright{rightsretained}
\acmDOI{10.1145/3219819.3219946}
\acmISBN{978-1-4503-5552-0/18/08}
\acmConference[KDD '18]{The 24th ACM SIGKDD International Conference on Knowledge Discovery \& Data Mining}{August 19--23, 2018}{London, United Kingdom}
\acmYear{2018}
\copyrightyear{2018}
\acmPrice{15.00}
\settopmatter{printacmref=false, printccs=true, printfolios=true}
\pagestyle{plain}



\begin{abstract}
In this paper we propose and study the problem of 
trajectory-driven influential billboard placement: given a set of billboards $\ur$ (each with a location and a cost), a database of trajectories $\td$ and a budget $\budget$, find a set of billboards within the budget to influence the largest number of trajectories.  One core challenge is to identify and reduce the overlap of the influence from different billboards to the same trajectories,  while keeping the budget constraint into consideration. We show that this problem is NP-hard and present an enumeration based algorithm with $(1-1/e)$ approximation ratio. \ping{However, the enumeration should be very costly when $|\ur|$ is large.} 
By exploiting the locality property of billboards' influence, we propose a partition-based framework \psel. \psel partitions $\ur$ into a set of small clusters, computes the locally influential billboards for each cluster, and merges them to generate the global solution. \ping{Since the local solutions can be obtained much more efficient than the global one, \psel should reduce the computation cost greatly; meanwhile it achieves a non-trivial approximation ratio guarantee. }
Then we propose a \bbsel method to further prune billboards with low marginal influence, while achieving the same approximation ratio as \psel. Experiments on real datasets verify the efficiency and effectiveness of our methods.
\end{abstract}

\maketitle

\section{Introduction}\label{sec:intro}
Outdoor advertising (ad) has a \$500 billion global market;
its revenue has grown by over 23\% in the past decade to over \$6.4 billion in the US alone~\cite{admarketsize2}.
As compared to social, TV, and mobile advertising, outdoor advertising delivers a high return on investment, and according to ~\cite{adbuyerincome} an average of \$5.97 is generated in product sales for each dollar spent. Moreover, it literally drives consumers `from the big screen to the small screen' to search, interact, and transact~\cite{advantagesofdigitalbillboard}. 
Billboards are the highest used medium for outdoor advertising (about 65\%), and 80\% people notice them when driving~\cite{relocatablead}. 

Nevertheless, existing market research only leverages traffic volume to assess the performance of billboards \cite{liu2017smartadp}. Such a straight-forward approach often leads to coarse-grained performance estimations and undesirable ad placement plans. To enable more effective placement strategies, we propose a fine-grained approach by leveraging the user/vehicle trajectory data.
Enabled by the prevalence of positioning devices, tremendous amounts of trajectories are being generated from vehicle GPS devices, smart phones and wearable devices.
The massive trajectory data provides new perspective to assess the performance of ad placement strategies.

In this paper, we propose a quantitative model to capture the billboard influence over a database of trajectories. Intuitively, if a billboard is close to a trajectory along which a user or vehicle travels, the billboard influences the user to a certain degree. When multiple billboards are close to a trajectory, the marginal influence is reduced to capture the property of diminishing returns.
Based on this influence model, we propose and study the the \underline{T}rajectory-driven \underline{I}nfluential Billboard \underline{P}lacement (\problem) problem:
given a set of billboards, a database of trajectories and a budget constraint $\budget$, it finds a set of billboards within budget $\budget$ such that the placed ads on the selected billboards influence the largest number of trajectories. To the best of our knowledge, this is the first work to address the \problem problem. The primary goal of this paper is to maximize the influence within a budget, which is critical to advertisers because the average unit cost per billboard is not cheap. For example, the average cost of a unit is \$14000 for four weeks in New York~\cite{adunitcost}; the total cost of renting 500 billboards is \$7,000,000 per month. Since the cost of a billboard is usually proportional to its influence, if we can improve the influence by 5\%, we can save about \$10,000 per week for one advertiser. The secondary goal is how to avoid expensive computation while achieving the same competitive influence value, 
so that prompt analytic on deployment plans can be conducted with different budget allocations.

In particular, there are two fundamental challenges to achieve the above goals. {First}, a user's trajectory can be influenced by multiple billboards,  which incurs the influence overlap among billboards. 
Figure~\ref{fig:example} shows an example for 6 billboards ($\bb_1, \dots, \bb_6$) and 6 trajectories ($\tr_1, \dots, \tr_6$). Each billboard is associated with a $\lambda$-radius circle, which represents its influence range. If any point $\p$ in a trajectory $\tr$ lays in the circle of $\bb$, $\tr$ is influenced by $\bb$ with a certain probability. Thereby, trajectory $\tr_1$ is first influenced by billboard $\bb_1$ and then influenced by $\bb_3$.  If the selected billboards have a large overlap in their influenced trajectories, advertisers may waste the money for repeatedly influencing the audiences who have already seen their ads.  {Second}, the budget constraint $\budget$ and various costs of different billboards make the optimization problem intricate. To our best knowledge, this is the first work that simultaneously takes three critical real-world features into consideration, i.e., budget constraint, non-uniform costs of billboards, and influence overlap of the selected billboards to a certain trajectory (Section~\ref{sec:problem}).

To address these challenges, we first propose a  greedy framework \enumgreedy by employing the enumeration technique~\cite{khuller1999budgeted}, 
which can provide an $(1 - 1/e)$-approximation for \problem. 
%
%
However the algorithm runs in a prohibitively large complexity of $O(|\mathcal{\td}|\cdot|\ur|^5)$, where $|\td|$ and $|\ur|$ are the number of trajectories and billboards respectively.
To avoid such high computational cost, we exploit the locality feature of the billboard influence and propose a partition-based framework. The core idea works as follows: first, it partitions the billboards into a set of clusters with low influence overlap; second, it executes the enumeration algorithm to find local solutions; third, it uses the dynamic programming approach to construct the global solution based on the location solutions maintained by different clusters.
We prove that the partition based method provides a theoretical approximation ratio.
To further improve the efficiency, we devise a lazy probe approach by pro-actively estimating the upper bound of each cluster and 
combining the results from a cluster only when its upper bound is significant enough to contribute to the global solution. 

\begin{figure}[!t]
\vspace{-.5em}
	\centering
	\includegraphics[width=0.75\linewidth]{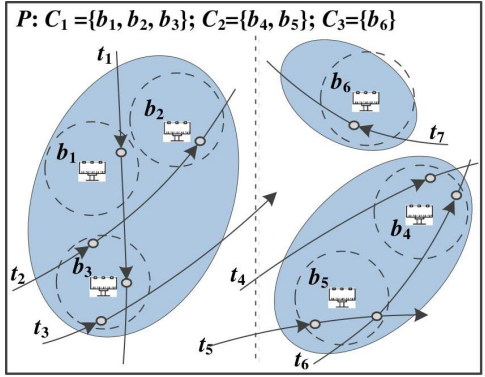}
\vspace{-1.25em}
	\caption{A Motivating Example ($\cost(b_i)=i$)
		}
	\label{fig:example}
\vspace{-.25em}
\end{figure}

Beyond billboard selection, our solution is useful in any store site selection problem that needs to consider the influence gain w.r.t. the cost of the store under a budget constraint. 
The only change is a customization of the influence model catered for specific scenarios, while the influence overlap is always incurred whenever the audiences are moving. 
For example in the electric vehicle charging station deployment, each station has an installment fee and a service range, which is similar to the billboard in \problem. Given a budget limit, its goal is to maximize the deployment benefit, which can be measured by the trajectories that can be serviced by the stations deployed.
In summary, we make the following contributions.

\begin{itemize}\vspace{-.25em}
\item 
We formulate the problem of trajectory-driven influential billboard placement (\problem). To our best knowledge, this is the first work that simultaneously takes three critical real-world features into consideration, i.e., budget constraint, unequal costs of billboards, and influence overlap of the selected billboards to a certain trajectory  (see Section~\ref{sec:problem}).
\item We present a greedy algorithm with the enumeration technique (\enumgreedy) as the baseline solution, 
which provides an approximation ratio of $(1-1/e)$ (see Section~\ref{sec:pre}). 
\item We propose a partition-based framework (\psel) by exploiting the locality property of the influence of billboards. \psel significantly reduces the computation cost while achieving a theoretical approximation ratio (see Section~\ref{sec:partition_method}).
\item We propose a \bbsel method to further prune billboards with low benefit/cost ratio, which significantly reduces the practical cost of \psel while achieving the same approximation ratio (see Section~\ref{sec:bound-selection}). 
\item We conduct extensive experiments on real-world trajectory and billboard datasets. 
Our best method \bbsel significantly outperforms the traditional greedy approach in terms of quality improvement over the naive traffic volume approach by about 99\%, and provide competitive quality against the \enumgreedy baseline while achieving 30$\times$-90$\times$ speedup in efficiency (see Section~\ref{sec:exp}).
\end{itemize}
\vspace{-.5em}

\section{Preliminary}\label{sec:problem}

In this section, we first formulate our problem,  and then review the relevant studies and justify their differences to our work. 

\subsection{Problem Formulation}\label{sec:pf}

In a trajectory database $\td$, each (human or vehicle) trajectory $t$ is in the form of a sequence of locations $\tr=\{\p_1,\p_2,...,\p_{|t|}\}$; a trajectory location $\p_i$ is represented by $\{\var{lat}, \var{lng}\}$, where $\var{lat}$ and $\var{lng}$ represent the latitude and longitude respectively.
A billboard $\bb$ is in the form of a tuple $\{\bbl, \cost \}$, where $\bbl$ and $\cost$ denote $\bb$'s location and leasing cost respectively. 
Without loss of generality, we assume that a billboard carries either zero or one advertisement at any time.

\begin{defn}\label{def:meet}
\label{meet}
 We define that $\emph\bb$ can influence $\emph\tr$, if $\exists \emph\p_i \in \emph\tr$, such that $Distance(\emph\p_i,\emph\bb.\emph\bbl) \leq \lambda$, where $Distance(\emph\p_i,\emph\bb.\emph\bbl)$ computes a certain distance between $\emph\p_i$ and $\emph\bb.\emph\bbl$, and $\lambda$ is a given threshold. 
\end{defn}

The choice of distance functions is orthogonal to our solution, and we choose Euclidean distance for illustration purpose. 

\noindent {\bf Influence of a billboard $\bb_i$ to a trajectory} $\tr_j$, $\pr(\bb_i, \tr_j)$. Given a trajectory $\tr_j$ and a billboard $\bb_i$ that can influence $\tr_j$, $\pr(\bb_i, \tr_j)$ denotes the influence of $\bb_i$ to $\tr_j$. The influence can be measured in various ways depending on application needs, such as the panel size, the exposure frequency, the travel speed and the travel direction. Note that our solutions of finding the optimal placement is orthogonal to the choice of influence measurements, so long as it can be computed deterministically given a $\bb_i$ and $\tr_j$. By looking into the influence measurement of one of the largest outdoor advertising companies LAMAR \cite{adunitcost}, we observe that panel size and exposure frequency are used. Moreover, these two can be obtained from the real data, hence we adopt them in our influence model and experiment. (1) For all $\bb_i \in \ur$ and $\tr_j \in \td$, we set $\pr(\bb_i, \tr_j)$ as a uniform value (between 0 and 1) if $\bb_i$ can influence $\tr_j$.  (2) Let $\size(\bb_i)$ be the panel size of $\bb_i$. We set $\pr(\bb_i, \tr_j)= \size(\bb_i)/A$ for $\tr_j$ influenced by $b_i$, where $A$ is a given value that is larger than ${\max }_{\bb_i \in \ur} \size(\bb_i)$.

\noindent {\bf Influence of a billboard set $S$ to a trajectory $\tr_j$, $\pr(S, \tr_j)$.}
It is worth noting that $\pr(S, \tr_j)$ cannot be simply computed as $\sum\nolimits_{\sbb_i \in S} {\pr(\bb_i,\tr_j)}$, because different billboards in $S$ may have overlaps when they influence $\tr_j$. Obviously $\pr(S,\tr_j) $ should be the probability that at least one billboard in $S$ can influence $\tr_j$. Thus, we use the following equation to compute the influence of $S$ to $\tr_j$.
	\begin{equation}
	\label{inftoset}
		\pr(S,\tr_j) = 1 - \prod\nolimits_{\sbb_i \in S} {({1 - \pr(\bb_i, \tr_j)})} 	     
	\end{equation}
where $(1 - \pr(\bb_i, \tr_j))$ is the probability that $\bb_i$ cannot influence $\tr_j$.

\noindent {\bf Influence of a billboard set $S$ to a trajectory set \td, $\ifl(S)$.} Let  $\td_S$ denote the set of trajectories in $\td$ that are influenced by at least one billboard in $S$.  The influence of a billboard set $S$ to a trajectory set $\td$ is computed by summing up $\pr(S, \tr_j)$ for $\tr_j\in\td_S$:
\begin{equation}
\label{influenceofset}
\ifl(S) = \sum\nolimits_{\str_j \in \td_S} {\pr(S, \tr_j)} 
\end{equation}

\begin{example} 
	Let $S=\{\bb_1, \bb_2, \bb_3\}$ be a set of billboards chosen from all billboards in Figure~\ref{fig:example}, and trajectories $\tr_1$, $\tr_2$ and $\tr_3$ that are influenced by at least one billboard in $S$. Let $\pr(\bb_1,\tr_1)$ = 0.1,  $\pr(\bb_3,\tr_1)$ = 0.3 and $\pr(\bb_2,\tr_1)$ = 0 ($\bb_2$ does not influence $\tr_1$). By Equation~\ref{inftoset}, we have $\pr(S,\tr_1)=1-(1-\pr(\bb_1, \tr_1))\times(1-\pr(\bb_3, \tr_1))= 1-(1-0.1)\times(1-0.3)=0.37$. Similarly, we have $\pr(S,\tr_2)=0.44$ and $\pr(S, \tr_3)=0.3$. Finally, the total influence of $S$ is equal to $\pr( S,\tr_1)+\pr(S,\tr_2)+\pr(S,\tr_3)=1.11$. 
\end{example}

\vspace{-.5em}

\begin{defn} \rm{\textbf{(\underline{T}rajectory-driven \underline{I}nfluential Billboard \underline{P}lacement ~(\problem))}} Given a trajectory database $\td$, a set of billboards $U$ to place ads and a cost budget $\budget$ from a client, our goal is to select a subset of billboards $S\subset\ur$, which maximizes the expected number of influenced trajectories such that the total cost of  billboards in $S$ does not exceed budget $\budget$. 
\end{defn}

\vspace{-.75em}
\begin{thm}
The \problem problem is NP-hard.
\vspace{-.25em}
  \begin{proof}
\vspace{-.5em}
   	We prove it by reducing the Set Cover problem to the \problem problem. In the Set Cover problem, given a collection of subsets $S_1,\ldots,S_m$ of a universe of elements $U'$, we wish to know whether there exist $k$ of the subsets whose union is equal to $U'$. We map each element in $U'$ in the Set Cover problem to each trajectory in $\td$. We also map each subset $S_i$ to the set of trajectories influenced by a billboard $\bb_i$. Consequently, if all the trajectories in $U'$ are influenced by $S$, the influence of $S$ is $|U'|$. 
Subsequently, the cost of each billboard is set to $1$ and budget $\budget$ in \problem is set to $k$ (selecting only $k$ billboards). The Set Cover problem is equivalent to deciding if there is a $k$-billboard set with the maximum influence $U'$ in the \problem problem. As the  set cover problem is NP-complete, the decision problem of \problem is NP-complete, and the optimization problem is NP-hard.
   \end{proof}    
\end{thm}

\vspace{-.5em}


\subsection{Related work}\label{sec:relatedwork} 
\textbf{Maximized Bichromatic Reverse k Nearest Neighbor (MaxBR$k$NN).}
The MaxBR$k$NN queries \cite{wong2009efficient,maxbrknn_1,maxbrknn_2,maxbrknn_3} aim to find the optimal location to establish a new store such that it is a $k$NN of the maximum number of users based on the spatial distance between the store and users' locations. Different spatial properties are exploited to develop efficient algorithms, such as space partitioning~\cite{maxbrknn_2}, intersecting geometric shapes~\cite{wong2009efficient}, and sweep-line techniques~\cite{maxbrknn_1}. Recently, the MaxRKNN query \cite{maxrknnt} is proposed to find the optimal bus route in term of maximum bus capacity by considering the audiences' source-destination trajectory data. Regarding the usage of trajectory data, most recent work only focus on top-k search over trajectory data~\cite{traj_1,traj_2}. 

Our \problem problem is different from MaxBR$k$NN in two aspects. (1) MaxBR$k$NN assumes that each user is associated with a fixed (check-in) location. In reality, the audience can meet more than one billboard while moving along a trajectory, which is captured by the \problem model. Thus it is challenging to identify such influence overlap when those billboards belong to the same placement strategy. (2) Billboards at different locations may have different costs, making this budget-constrained optimization problem more intricate. However, MaxBR$k$NN assumes that the costs of candidate store locations are uniform.

\vspace{1mm}
\textbf{Influence Maximization and its variations.} 
The original Influence Maximization (IM) problem aims to {find a size-$k$ subset of all nodes} in a social network that could maximize the spread of influence~\cite{kempe2003maximizing}. Independent Cascade (IC) model and Linear Threshold (LT) model are two common models to capture the influence spread. Under both models, this problem has been proven to be NP-hard, and a simple greedy algorithm guarantees the best possible approximation ratio of $(1-1/e)$ in polynomial time. Then the key challenge lies in how to calculate the influence of sets efficiently, and a plethora of algorithms \cite{chen2010scalable, leskovec2007cost, borgs2014maximizing, tang2014influence,DBLP:journals/pvldb/ChenFLFTT15} have been proposed to achieve speedups.
Some new models are also introduced to solve IM under complex scenarios. 
IM problems for propagating different viral products are studied in \cite{li2015real,li2017discovering}.
%
Recently, the IM problem is extended to location-aware IM (LIM) problems by considering different spatial contexts~\cite{li2014efficient, guo2017influence, liu2017smartadp}. \citet{li2014efficient} find the seed users in a location-aware social network such that the seeds have the highest influence upon a group of audiences in a specified region.  
\citet{guo2017influence} select top-k influential trajectories based on users' checkin locations.
See a recent survey \cite{li2018influence} for more details.

Our \problem differs from the IM problems as follows. (1) The cardinality of the optimal set in IM problems is often pre-determined because the cost of each candidate is equal to each other (when the cost is 1, the cardinality is $k$), thus a theoretically guaranteed solution can be directly obtained by a naive greedy algorithm. However, in our problem, the costs of billboards at different locations differ from one to another, so the theoretical guarantee of the naive greedy algorithm is poor~\cite{khuller1999budgeted}. 
(2) Since IM problems adopt a different influence model to ours, they mainly focus on how to efficiently and effectively estimate the influence propagation, while \problem focuses on how to optimize the profit of $k$-combination by leveraging the geographical properties of billboards and trajectories.

\vspace{1mm}
\textbf{Maximum $k$-coverage problem.}
Given a universe of elements $U$ and a collection  $S$ of subsets from $U$, the Maximum $k$-coverage problem (MC) aims to select at most $k$ sets from $S$ to maximize the number of elements covered. This problem has been shown to be NP-hard, and \citet{feige1998threshold} has proven that the greedy heuristic is the most effective polynomial solution and can provide $(1-1/e)$ approximation to the optimal solution. The budgeted maximum coverage (BMC) problem \cite{khuller1999budgeted} further considers a cost for each subset and tries to maximize the coverage with a budget constraint.
\citet{khuller1999budgeted} show that the naive greedy algorithm no longer produces solutions with an approximation guarantee for BMC. To overcome this issue, they devise a variant of the greedy-based algorithm for BMC, which provides solutions with a $(1-1/e)$-approximation.  However, by a rigorous complexity analysis in Section \ref{sec:marginal-index}, we find that this algorithm needs to take $O(|\mathcal{\td}|\cdot|\ur|^5)$ time to solve our \problem problem, which does not scale well in practice (see Section~\ref{sec:exp}).

\section{Our Framework}\label{sec:pre}

We first discuss two baselines that are extended from the algorithms for the general Budgeted Maximum Coverage (BMC) problem. In particular, we first present a basic greedy method (Algorithm~\ref{greedy}). It is worth noting that, the basic greedy method is proved by \citet{khuller1999budgeted} to achieve $(1-1/\sqrt{e})$-approximation; however, we find it is not correct and we prove it to be $\frac{1}{2}(1-1/e)$. As the approximation ratio of this algorithm is low, we then propose an enumeration algorithm with $(1-1/e)$-approximation (Algorithm~\ref{greedy_enumerate}). However, the enumeration algorithm incurs a high computation cost as it has to enumerate a large number of feasible candidate combinations, which is impractical when $|\ur|$ and $|\td|$ are large. This motivates us to exploit the spatial property between billboards and trajectories to propose our own framework to dramatically reduce the computation cost, where an overview is shown in Section~\ref{sec:overview}.
Important notations used in our framework are presented in Table~\ref{table:notation_a}.

\subsection{Baselines} \label{sec:baseline}

\subsubsection{A Basic Greedy Method}\label{sec:greedy}

A straightforward approach is to select the billboard $\bb$ which maximizes the unit marginal influence, i.e., $\frac{\marginifl(b|S)}{\cost(\{b\})}$, to a candidate solution set $S$, until the budget is exhausted, where $\marginifl(\bb|S)$ denotes the marginal influence of $\bb$ to $S$, i.e., $\ifl(S\cup \{\bb\})-\ifl(S)$. Lines 1.3-1.8 of Algorithm~\ref{greedy} present how it works. However, such a greedy heuristic cannot achieve a guaranteed approximation ratio. For example, given two billboards $\bb_1$ with influence 1 and $\bb_2$ with influence $x$. Let $\cost(\bb_1)=1$, $\cost(\bb_2)=x+1$ and $\budget=x+1$. The optimal solution is $\bb_2$ which has influence $x$, while the solution picked by the greedy heuristic contains the set $\bb_1$ and the influence is 1. The approximation factor for this instance is $x$. As $x$ can be arbitrarily large, this greedy method is unbounded.

To overcome this issue, 
we modify the above method by considering the best single billboard solution as an alternative to the output of the naive greedy heuristic.  
In particular, we add lines 1.9-1.13 in Algorithm~\ref{greedy} to consider such best single billboard solution. As a result, a complete Algorithm~\ref{greedy} forms our basic greedy method (\gre) to solve the \problem problem.

\noindent
\textbf{Time Complexity of \gre.} 
In each iteration, Algorithm~\ref{greedy} needs to scan all the billboards in $(\ur\setminus S)$ and compute their (unit) marginal influence to the chosen set. 
Each marginal influence computation needs to traverse $\td$ once in the worst case. Thus, adding one billboard into $S$ takes $O(|\td| \cdot |\ur|)$ time. Moreover, when $\budget$ is sufficiently large, this process would repeat $|\ur|$ times at the worst case.  Therefore, the time complexity of Algorithm \ref{greedy} is $O(|\td| \cdot |\ur|^2)$. 

It is worth noting that the authors in \cite{khuller1999budgeted} claim that \gre achieves an approximation factor of $(1-1/\sqrt e)$ for the budgeted maximum coverage problem. However, we find that this claim is problematic and the bound of \gre should be $\frac{1}{2}(1-1/e)$, as presented in Theorem~\ref{correctness_gre}. 
 
\begin{theorem}\label{correctness_gre}
\gre achieves an approximation factor of $\frac{1}{2}(1-1/e)$ for the \problem problem.
\end{theorem}

\noindent\textbf{Discussion on the problematic approximation ratio of $(1 - \frac{1}{\sqrt{e}})$ originally presented in \cite{khuller1999budgeted}).} Note that Theorem~\ref{correctness_gre} is essentially the Theorem 3 introduced in \cite{khuller1999budgeted} because both try to find the approximation ratio of the same cost-effective greedy method for a budgeted maximum coverage (BMC) problem. We first present a proof of Theorem~\ref{correctness_gre} which shows that the \gre achieves $\frac{1}{2}(1-1/e)$-approximation, then we justify why the  approximation ratio of $(1 - \frac{1}{\sqrt{e}})$ originally presented in \cite{khuller1999budgeted}) is problematic. 

\begin{proof}
	(Theorem~\ref{correctness_gre})
	Let $\opt$ denote the optimal solution and $\mathbb{M}_{k^*+1}$ be the marginal influence of adding $\bb_{k^*+1}$ (be consistent to the definition in Lemma \ref{hypo_bound}).
	When applying Lemma \ref{khuller} to the $(k^*+1)$-th iteration, we get:
	\begin{align*}
	\small		
	I(S_{k^*+1}) & =\ifl(S_{k^*}\cup b_{k^*+1})=\ifl(S_{k^*})+\mathbb{M}_{k^*+1}\\
	& \ge \left[ 1 -  \prod\nolimits_{j = 1}^{k^*+1} \left( {1 - \frac{w(\bb_j)}{L}} \right) \right] \cdot \ifl(\opt) \\
	& \ge \left( 1 - {(1 - \frac{1}{k^* + 1})^{k^* + 1}}\right)  \cdot \ifl(\opt)\\
	&  \ge (1 - \frac{1}{e}) \cdot \ifl(\opt)
	\end{align*} 
	
	Note that the second inequality follows from the fact that adding $\bb_{k^*+1}$ to $S$ violates the budget constraint $\budget$, i.e., $\cost(S_{k^*+1})=\cost(S_{k^*})+\cost(\bb_{k^*+1}) \geq \budget$. 
	
	Intuitively, $\mathbb{M}_{k^*+1}$ is at most the maximum influence of the elements covered by a single billboard, i.e., $H$ is found by \gre in the first step (line 1.3). Moreover, as $S_{k^*}\subseteq S$ ($S$: the solution of \gre ), we have:
	\begin{equation}\label{appr_ce}
	\ifl(S)+\ifl(H) \geq \ifl(S_{k^*+1}) \geq (1-1/e)\ifl(\opt)
	\end{equation}
	From the above inequality we have that, among $\ifl(S)$ and $\ifl(H)$, at least one of them is no less than 	$\frac{1}{2}(1-1/e)\ifl(\opt)$. Thus it shows that \gre achieves an approximation ratio of at least $\frac{1}{2}(1-1/e)$.
\end{proof}


In the original proof of Theorem 3 in \cite{khuller1999budgeted}, the authors have tried to prove that \gre is $(1-1/\sqrt{e})$-approximate for the following three cases respectively.

\noindent\textbf{Case 1:} the influence of the most influential billboard in $\ur$ is greater than $\frac{1}{2}\ifl(\opt)$.

\noindent\textbf{Case 2:} no billboard in $\ur$ has an influence greater than $\frac{1}{2}\ifl(\opt)$ and $\cost(S)\leq \frac{1}{2}\budget$.

\noindent\textbf{Case 3:} no billboard in $\ur$ has an influence greater than $\frac{1}{2}\ifl(\opt)$ and $\cost(S)\geq \frac{1}{2}\budget$.

The authors also proved that the bound in Theorem~\ref{correctness_gre} can be further tightened to $\frac{1}{2}$ for case 1 and case 2, which are  right.
However, there is a problem in the proof for case 3.
Intuitively, if we can prove that \gre is $(1-1/\sqrt{e})$-approximate in Case 3, then by the union bound \gre can achieve an approximation factor of $(1-1/\sqrt{e})$. 

Let $\cost(S_{k^*})$ be equal to $\gamma \budget$ and $\gamma \in (0, 1)$. By applying Lemma~\ref{khuller} to the $k^*$-th iteration, we get:
\begin{align*}
\small		
I(S) \ge I(S_{k^*}) & \ge \left[ 1 -  \prod\nolimits_{j = 1}^{k^*} \left( {1 - \frac{w(\{\bb_j\})}{\budget}} \right) \right] \cdot \ifl(\opt) \\
& \ge \left( 1 - {(1 - \frac{1}{\frac{1}{\gamma}{k^*}})^{k^*}}\right)  \cdot \ifl(\opt)\\
&  \ge (1 - \frac{1}{e^\gamma}) \cdot \ifl(\opt)
\end{align*} 
Note that $\cost(S)\geq \frac{1}{2}\budget$ cannot guarantee $\gamma \geq 1/2$ because $S_{k^*} \subseteq S$. Consequently, the inequality cannot guarantee $ I(S) \ge (1 - \frac{1}{\sqrt{e}}) \cdot \ifl(\opt)$. However, it is concluded in\cite{khuller1999budgeted} that \gre achieves an approximation factor of $(1 - \frac{1}{\sqrt{e}})$ under the assumption of $\gamma \geq 1/2$. Therefore, the proof in \cite{khuller1999budgeted} is problematic.


\setlength{\algomargin}{1.2em} 
\begin{algorithm}[!t]
\caption{\gre$(\ur, \budget, S)$} 
\label{greedy}
\begin{small}
{\bf Input:} A billboard set $\ur$, a budget $\budget$ and a set $S$ ($S=\phi$ by default)

{\bf Output:} A billboard set $S \subseteq \ur$ such that $\cost(S)\leq \budget$


\Repeat{$U = \phi$}
{
	Select $\bb \in \ur \setminus S$ that maximizes $\frac{\marginifl(\sbb|\svar{S})}{\cost(\{\sbb\})}$ 
	
	\uIf {$\emph\cost(S)+ \emph\cost(\emph\bb) \leq \emph\budget$}
		{
			$S \gets{S} \cup \{\bb\}$
		} 
	
	$\ur \gets \ur \setminus  \{\bb\}$
	
}

{$H \gets \var{\rm{argmax}}\{\ifl(\{\bb\})|$  $b \in \ur$, and ${\cost(\{\bb\})\leq \budget}\}$}

\uIf {$\ifl(H) > \ifl(S)$}
{
	\Return $H$ 
}
\Else
{
	\Return $S$ 	
}
	\end{small}

\end{algorithm}

\setlength{\algomargin}{1.2em} 
\begin{algorithm}[!t]
		\caption{\enumgreedy$(\ur, \budget)$}
		\label{greedy_enumerate}
		\begin{small}
			{\bf Input:} A billboard set $\ur$, budget $\budget$
			
			{\bf Output:} A billboard set $S \subseteq \ur$ with the cost constraint $\cost(S)\leq \budget$
			
			{Let $\tau$ be a constant /* $\tau$=2 to achieve the lowest time complexity */}
			
			 $H_1 \gets \var{\rm{argmax}}\{\ifl({S'})|$  $S' \subseteq \ur, |S'| \leq \tau$, and ${\cost(S')\leq \budget}\}$
		
			 $H_2 \gets \phi$
		
		\For{all $S \subseteq \ur$, such that $|S|=\tau+1$ and $\emph\cost(S)\leq \emph\budget$}
			{
		
				 $S \gets \textbf{\gre}(\ur \setminus  S, \budget-\cost(S), S)$
				
				\uIf {$\ifl(S)>\ifl(H_2)$}
					{
					 	$H_2 \gets S$
					}			
			}		
		
		\uIf {$\ifl(H_1) > \ifl(H_2)$}
			{
				\Return $H_1$ 
			}
		\Else
			{
				\Return $H_2$ 	
			}
		\end{small}
	\end{algorithm}

\begin{figure}[!t]	\vspace{-.5em}
	\centering
	\includegraphics[width=0.95\linewidth]{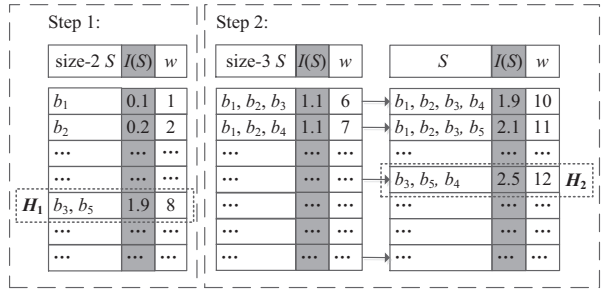}
	\vspace{-1em}
	\caption{A running example of Algorithm~\ref{greedy_enumerate}}
	\label{fig:greedyinlf}
\end{figure}
\subsubsection{Enumeration Greedy Algorithm}\label{sec:marginal-index}

Since \gre is only $\frac{1}{2}(1-1/e)$-approximation, we would like to further boost the influence value, even at the expense of longer processing time as compared to \gre. Note that it is critical to maximize the influence as it can save real money, while keeping acceptable efficiency.  
\ping{Thus we utilize the enumeration-based solution proposed in \cite{khuller1999budgeted} to obtain $(1-1/e)$-approximation.}

\enumgreedy runs in two phases. In the first phase (line 2.4), 
it enumerates all feasible billboard sets whose cardinality is no larger than a constant $\tau$, and adds the one with the largest influence to $H_1$. In the second phase (lines 2.5-2.9), 
it enumerates each feasible set of size-$(\tau+1)$ whose total cost does not exceed budget $L$. Then for each set $S$, it invokes \ngre to greedily select new billboards (if any) that can bring marginal influence, and chooses the one that maximizes the influence under the remaining budget $\budget-\cost(S)$ and assigns it to $H_2$. 
Last, if the best influence of all size-$(\tau+1)$ billboard sets is still smaller than that of its size-$\tau$ counterpart (i.e., $\ifl(H_1)>\ifl(H_2)$), $H_1$ is returned; otherwise, $H_2$ is returned. 


\begin{example}
	Figure~\ref{fig:greedyinlf} illustrates an instance of Algorithm \ref{greedy_enumerate} on Figure \ref{fig:example}'s scenario. We assume $\tau= 2$ and $\budget=12$, and the cost of a billboard is its id number (e.g. $\cost(\bb_1)$=1). For $\pr(\bb_i, \tr_j)$, we use the value set in Figure~\ref{fig:invlist} to compute $\ifl(S)$.
	In the first step, Algorithm~\ref{greedy_enumerate} enumerates all feasible sets of size less than 3, among which the billboard set \{$\bb_3$, $\bb_5$\} has the largest influence ($\ifl(H_1)$ = $\pr(\bb_3, \tr_1)$+$\pr(\bb_3, \tr_2)$+$\pr(\bb_3, \tr_3)$+$\pr(\bb_5, \tr_5)$+$\pr(\bb_5, \tr_6)$=1.9). 
	In the second step, it starts from the feasible size-3 sets and expands greedily until the budget constraint is violated. The right part of Figure~\ref{fig:greedyinlf} shows the eventual billboard set $S$ whose total cost does not violate the budget constraint $\budget$ (line 2.7 of Algorithm~\ref{greedy_enumerate}). Here $\cost(S)$=$\budget$=12, and assigns it to $H_2$ (line 2.9), so $H_2$= $\{\bb_3, \bb_4, \bb_5\}$ and its influence value $\ifl(H_2) = 2.5$ which is the largest influence. 
	Since $\ifl(H_1)<\ifl(H_2)$, Algorithm~\ref{greedy_enumerate} returns $\{\bb_3, \bb_4, \bb_5\}$ as the final result.
\end{example}

\vspace{-.5em}
\noindent
\textbf{Time Complexity of \enumgreedy.} 
 At the first phase, Algorithm \ref{greedy_enumerate} needs to scan all feasible sets with cardinality $\tau$ and the number of such sets is $O(|\ur|^\tau)$. For each such candidate set, we need to scan $\td$ to compute its influence, thus the first phase takes $O(|\td| \cdot {|\ur|}^\tau)$ time. At the second phase, there are $O(|\ur|^{\tau+1})$ sets of cardinality $\tau+1$, and Algorithm \ref{greedy_enumerate} invokes Algorithm~\ref{greedy} for each set. In the worse case, the cost of any size-$(\tau+1)$ sets should be much smaller than $\budget$ and thus these sets would not affect the complexity of \gre in line 2.6. Therefore, the second phase takes $O(|\td| \cdot |\ur|^2 \cdot |\ur|^{\tau+1})$ time. In total, Algorithm \ref{greedy_enumerate} takes $O(|\td| \cdot|\ur|^\tau + |\td|\cdot|\ur|^{\tau+3})=O(|\td|\cdot|\ur|^{\tau+3})$.

\vspace{.25em}
\noindent \textbf{Selection of $\tau$.} It has been proved in \cite{khuller1999budgeted} that Algorithm \ref{greedy_enumerate} can achieve an approximation factor of $(1 - 1/e)$ when $\tau \geq 2$. Note that (1) the approximation ratio $(1 - 1/e)$ cannot be improved by a polynomial algorithm~\cite{khuller1999budgeted} and (2) a larger $\tau$ leads to larger overhead,  thus we set $\tau=2$. So Algorithm~\ref{greedy_enumerate} can achieve the $(1-1/e)$-approximation ratio with a complexity of  $O(|\td|\cdot|\ur|^5)$.


\subsection{A Partition-based Framework}\label{sec:overview}

Although \enumgreedy provides a solution with an approximation ratio of $(1-1/e)$, it involves high computation cost, because it needs to enumerate all size-$\tau$ and size-$(\tau+1)$ billboard sets and compute their influence to the trajectories, which is impractical when $|\ur|$ and $|\td|$ are large. To address this problem, we propose a partition-based framework.

\vspace{.25em}
\noindent {\bf Partition-based Framework.} Our problem has a distance requirement that if a billboard influences a trajectory, the trajectory must have a point close to the billboard (distance within $\lambda$). All of existing techniques neglect this important feature, which can be utilized to enhance the performance.  After deeply investigating the problem, we observe that most trajectories span over a small area in the real world. 
For instance, around 85\% taxi trajectories in New York do not exceed five kilometers (see Section~\ref{sec:exp}). 
It implies that billboards in different areas should have small overlaps in their influenced trajectories, 
e.g., the number of trajectories simultaneously influenced by two billboards located in Manhattan and Queens is small. 
Thereby, we exploit such locality features to propose a partition based method called \psel. Intuitively, we partition $\ur$ into a set of small clusters, compute the locally influential billboards for each cluster, and merge the local billboards to generate the globally influential billboards of $\ur$. Since the local cluster has much smaller number of billboards, this method reduces the computation greatly while keeping competitive influence quality.

\begin{table}[!t]
	\caption{{Notations for problem formulation and solutions}}
	\vspace{-1em}
	\centering
	\begin{tabular}{lp{6.8cm}}
		\toprule
		{\bfseries Symbol} & {\bf Description}\\
		\midrule
	 $\tr$ ($\td$) & A trajectory (database) \\
	 $\ur$ & A set of billboards that a user wants to advertise\\
	 $\budget$ & the total budget of a user\\
	 $\ifl(S)$ & The influence of a selected billboard set $S$\\
	 $P$ & A billboard partition \\
	 $\ratio_{\svar\emph{ij}}$ & The overlap ratio between clusters\\
	 $\marginifl(\bb|S)$ & The marginal influence of $\bb$ to $S$  \\
	 $\theta$  & The threshold for a $\theta$-partition  \\
		$\gimatrix$ & The DP influence matrix: $\gimatrix[i][l]$ is the maximum influence of the billboards selected from the first $i$ clusters within budget $l$ ($i \leq m$ and $l \leq \budget$)
 \\
		$\limatrix$ & The local influence matrix: $\limatrix[i][l]$ is the influence returned by $\enumgreedy(C_i, l)$, i.e., the maximum influence of billboards selected from cluster $C_i$ within budget $l$ \\
		\bottomrule
	\end{tabular}
	\label{table:notation_a}
\end{table}

\vspace{.25em}

\noindent {\bf Partition.}  We first partition the billboards to $m$ clusters $C_1,C_2,\ldots,C_m$, where different clusters have no (or little) influence overlap to the same trajectories. Given a budget $l_i$ for cluster $l_i$, by calling $\enumgreedy(C_i,l_i)$, we select the locally influential billboard set $S[i][l_i]$ from cluster $C_i$ within budget $l_i$, where $S[i][l_i]$ has the maximum influence $\limatrix[i][l_i]$. Next we want to assign a budget  to each cluster $C_i$ and take the union of $S[i][l_i]$ as the globally influential billboard set, where $l_1+l_2+\ldots+l_m\leq \budget$. Obviously, we want to allocate the budgets to different clusters to maximize 
\begin{equation}
\label{PBIM}
\mathop   \sum\nolimits_{i = 1}^m {\limatrix[i][l_i]}  
\end{equation}
$$s.t.~  l_1+l_2+\ldots+l_m\leq \budget$$

There are two main challenges in this partition-based method. 
(1) How to allocate the budgets to each cluster to maximize the overall influence? We propose a dynamic programming algorithm to address this challenge (see Section~\ref{sec:partition_method}). (2) How to partition the billboards to reduce the influence overlap among clusters? We propose a partition strategy to reduce the influence overlap and devise an effective algorithm to generate the clusters (see Section~\ref{sec:partition_method}).

\vspace{.25em}

\noindent {\bf Lazy Probe.}  Although the partition-based method significantly reduces the complexities over the enumeration approach, its dynamic programming process has to repeatably invoke \enumgreedy to probe the partial solution for every cluster in the partition. It is still expensive to compute the local influence by calling $\enumgreedy(C_i,l_i)$ many times. We find that it is not necessary to compute the real influence value for those clusters which have low influence to affect the final result, thus reducing the number of calls to $\enumgreedy$. The basic idea is that we estimate an upper bound ${\upperlimatrix[i][l_i]}$ of the local solution for a given cluster $C_i$ and a budget $l_i$; and we do not need to compute the real influence ${\limatrix[i][l_i]}$, if we find that using this cluster cannot improve the influence value. This method significantly reduces the practical cost of \psel while achieving the same approximation ratio. There are two challenges in the lazy probe method. (1) How to utilize the bounds to reduce the computational cost (i.e., avoid calling $\enumgreedy(C_i,l_i)$)? We propose a lazy probe technique (see Section~\ref{sec:bound-selection}). (2) How to estimate the upper bounds while keeping the same approximation ratio as \psel? We devise an incremental algorithm to estimate the bounds (see Section~\ref{sec:bound-selection}). 


\begin{figure}[!t]
	\centering
	\includegraphics[width=1.0\linewidth]{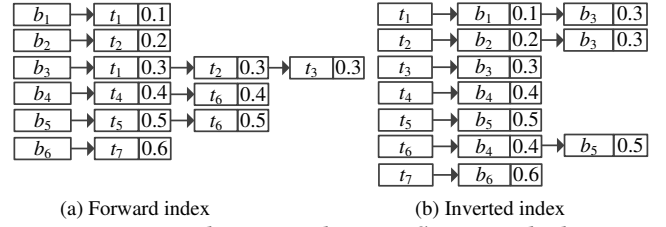}
	\vspace{-2.5em}
	\caption{An Index to Accelerate Influence Calculation} 
	\label{fig:invlist}
\end{figure}

\vspace{.25em}
\noindent 
\textbf{Index for efficient Influence Calculation.}  The most expensive part of the algorithm is to compute $\ifl(S)$, which in turn transforms to the computation of $\pr(b_i, \tr_j)$ and $\pr(S, \tr_j)$ for $1\leq i \leq |\td|$ and $1\leq j \leq |\ur|$. To improve the performance, we propose two effective indexes. (1) A forward index for billboards (Figure~\ref{fig:invlist}a). For each billboard $b_i$, we keep a forward list of trajectories that are influenced by this billboard, associated with the weight $\pr(b_i, \tr_j)$. Then we can easily compute $\pr(\{b_i\})$ by summing up all the weights in the forward list.  To build the forward list, we need to find the trajectories that are influenced by $b_i$. To achieve this goal, we build an R-tree for the points in trajectories. Then a range query on $b_i$ can build the forward list efficiently. (2) An inverted index for trajectories (Figure~\ref{fig:invlist}b). For each trajectory $t_j$, we keep an inverted list of billboards that influence $t_j$, associated with the weight $\pr(b_i, \tr_j)$. To compute $\pr(S, \tr_j)$, we can use the inverted list to find all the billboards that influence $\tr_j$ and use Equation~\ref{inftoset} to compute $\pr(S, \tr_j)$. Then we use Equation~\ref{influenceofset} to compute $\pr(S)$. 
\begin{example}
	In Figure~\ref{fig:invlist}, let $S=\{\bb_1\}$, and we want to compute the marginal influence of $\bb_3$ w.r.t. the current candidate set $S$.  First, we traverse the forward index to get the trajectory set influenced by $\bb_3$, and find that $\emph\tr_1$ is co-influenced by $\bb_3$ and $S$. As $\bb_3$ also can influence $\tr_2$ and $\tr_3$, the marginal influence of $\bb_3$ is computed by $\pr(S \cup \{\bb_3\}, \tr_1)-\pr(\bb_3, \tr_1)+\pr(\bb_3, \tr_2)+\pr(\bb_3, \tr_3)$.  According to Equation \ref{inftoset}, this computation depends on $\pr(\bb_i, \tr_1)$, $\pr(\bb_i, \tr_2)$ and $\pr(\bb_i, \tr_3)$, for all $\bb_i \in S \cup \{\bb_3\}$, and these values can be obtained from traversing the inverted list directly. 
\end{example}

\section{Partition based Method}\label{sec:partition_method}

This section proposes a partition-based method which contains three steps to reduce the computation cost:
\begin{enumerate}[noitemsep]
\item Partition $\ur$ into a set of clusters according to their influence overlap;
\item Find local influential billboards with regard to each cluster by calling \enumgreedy;
\item Aggregate these local influential billboards from clusters to obtain the global solution for \problem.
\end{enumerate}

For convenience sake, this section first presents how to select the billboards based on a given partition scheme, and then discuss how to find a good partition that can provide a high performance and a theoretical approximation ratio for our partition-based method.

\begin{table*}[!t]
	\vspace{-1em}
	\centering
	\captionsetup{justification=centering,margin=2cm}
	\caption{An example of Algorithm \ref{PBA} (each cell in $\lsmatrix$ and $\gsmatrix$ record the selected billboards corresponding to the maximum influence value recorded in each cell of $\limatrix$ and $\gimatrix$)}
	\vspace{-8pt}
	\label{example-PBA}
	\begin{small} 
		\begin{minipage}[c]{0.27\linewidth}
			\subfloat[$\lsmatrix$]{
				\label{S_PBA}
				\begin{tabular}{llll}
					\hline
					\multicolumn{1}{c}{} & \multicolumn{1}{c}{1} & \multicolumn{1}{c}{2} & \multicolumn{1}{c}{3} \\ \hline
					\multicolumn{1}{c}{-}   & \multicolumn{1}{c}{-}      & \multicolumn{1}{c}{-}           & \multicolumn{1}{c}{-}   \\
					$C_1$                 & $\{1\}$                    & $\{1, 3\}$                  & $\{1, 3, 2\}$                             \\
					$C_2$                 & $\{6\}$                     & $\{4, 5\}$                  & $\{4, 5, 6\}$                           \\ 
					$C_3$                 & $\{7\}$                     & $\{7, 8\}$                   & $\{7, 8, 9\}$                         \\ \hline
				\end{tabular}
			}
		\end{minipage}
		\centering 
		\begin{minipage}[c]{0.22\linewidth}
			\subfloat[$\limatrix$]{
				\label{I_PBA}	
				\begin{tabular}{llll}
					\hline
					      & 1  & 2  & 3      \\ \hline
					  -   & 0  & 0  & 0   \\
					$C_1$ & 10 & 18 & 25  \\
					$C_2$ & 8  & 16 & 21  \\ 
					$C_3$ & 5  & 9  & 13  \\ \hline
				\end{tabular}
			}
		\end{minipage}
		\centering 
		\begin{minipage}[c]{0.27\linewidth}
			\subfloat[$\gsmatrix$]{
				\label{SH_PBA}	
				\begin{tabular}{llll}
					\hline
					\multicolumn{1}{c}{} & \multicolumn{1}{c}{1} & \multicolumn{1}{c}{2} & \multicolumn{1}{c}{3}  \\ \hline
					$i = 0$                      &\multicolumn{1}{c}{-}        &\multicolumn{1}{c}{-}           &\multicolumn{1}{c}{-} \\
					$i = 1$                      & $\{1\}$                     & $\{1, 3\}$                     & $\{1, 3, 2\}$    \\ 
					$i \leq 2$                   & $\{1\}$                     & $\{1, 3\}$                     & $\{1, 3, 6\}$      \\ 	
					$i \leq 3$                   & $\{1\}$                     & $\{1, 3\}$                   	& $\{\textbf{1}, \textbf{3}, \textbf{6}\}$               \\ \hline
					
				\end{tabular}
			}
		\end{minipage}
		\centering 
		\begin{minipage}[c]{0.22\linewidth}
			\subfloat[$\gimatrix$]{
				\label{W_PBA}	
				\begin{tabular}{llll}
					\hline
					& 1   & 2   & 3   \\ \hline
					$i = 0$ & 0 & 0 & 0  \\ 
					$i = 1$ & 10 & 18 & 25  \\ 
					$i \leq 2$ & 10 & 18 & 26  \\ 
					$i \leq 3$ & 10 & 18 & \textbf{26} \\ \hline
					
				\end{tabular}
			}
		\end{minipage}
		\centering 
	\end{small}
	\vspace{-2em}
\end{table*}

\subsection{Partition based Selection Method}\label{sec:partition_selection}

\begin{defn}\rm{(\textbf{Partition})}\label{partition}
	A partition of $U$ is a set of clusters $\{C_1,...,C_m\}$, such that $U = C_1 \cup C_2\cup...\cup C_m$, and $\forall i \neq j$, $C_i\bigcap C_j=\phi$.
	Without loss of generality, we assume that the clusters are sorted by their size, and $C_m$ is the largest cluster. 
\end{defn}

We follow a divide and conquer framework to combine partial solutions from the clusters. Let $S^*$ denote the billboard set returned by $\enumgreedy(\ur, \budget)$, $S[i][l]$ denote the billboard set  returned by $\enumgreedy(C_i, l)$, where $l<L$ is a budget for cluster $C_i$, as shown in Figure~\ref{fig:lem_1}.  Let $\limatrix[i][l]$ be the influence value of the billboard set $S[i][l]$, i.e.,  $\limatrix[i][l]=\ifl(S[i][l])$.  If $S[i][l]$ for $1\leq i\leq m$ have no overlap, we can assign a budget $l$ for each cluster and maximize the total influence based on Equation~\ref{PBIM}.

We note that the costs for billboards are integers in reality, e.g., the costs from a leading outdoor advertising company are all multiples of 100~\cite{adunitcost}. Thereby it allows us to design an efficient dynamic programming method to solve Equation~\ref{PBIM}. The pseudo code is presented in Algorithm~\ref{PBA}. 
It considers the clusters in $P$ one by one. Let $\gimatrix[i][l]$ denote the maximum influence value that can be attained with a budget not exceeding $l$ using up to the first $i$ clusters ($i \leq m$ and $l \leq \budget$).
Clearly, $\gimatrix[m][\budget]$ is the solution for Equation \ref{PBIM} since the union of the first $m$ clusters is $\ur$.
To obtain $\gimatrix[m][\budget]$, Algorithm~\ref{PBA} first initializes the matrices $\gimatrix$ and $\limatrix$ (line 3.3), 
and then constructs the global solution (line 3.7 to 3.17) with the following recursion:
\begin{equation}
\label{recursion}
\begin{array}{l}
\gimatrix[0][l] = 0 \\
\gimatrix[i][l]= \mathop {\max }\limits_{0 \le \svar{q} \le l} (\gimatrix[i - 1][l - q] + \limatrix[i][q] )
\end{array}
\end{equation}

Since the computation at the $i$th iteration only relies on the $(i-1)$th row of each matrix, we can use two $2 \times n$ matrices to replace $\gimatrix$ and $\limatrix$ for saving space.

	\begin{algorithm}[!t]
		\caption{\psel$(P, \budget)$}
		\label{PBA}		
		\begin{small}
		{\bf Input:} A $\theta$-partition $P$ of $\ur$, a budget $\budget$
		
		{\bf Output:} A billboard set $S$
		
		Initialize matrices $\gimatrix$ and $\limatrix$
		
		$m \gets |P|$
		
		\For {$i \gets 1$ to $m$}
		{
			
			\For {$l \gets 1$ to $\emph\budget$}
			{
					\tcc{$C_i$ is the $i$th cluster in $P$}
					
					Invoke $\textbf{\enumgreedy}(C_i,l)$ to compute $\limatrix[i][l]$ 	
							
					$q = \mathop {\arg\max }\limits_{0 \le \svar{q} \le l} (\gimatrix[i - 1][l - q] + \limatrix[i][q])$
				
					$\gimatrix[i][l] \gets \gimatrix[i-1][l-q] + \limatrix[i][q]$
			}			
		}			
		$S \leftarrow$ the corresponding selected set of $\gimatrix[m][\budget]$
		
		\Return {$S$}
	\end{small}
	\end{algorithm}

\begin{example}
	Given a partition of $\ur$ as $P=\{C_1,C_2,C_3\}$, where $C_1=\{1,2,3\}$, $C_2=\{4,5,6\}$ and $C_3=\{7,8,9\}$. For simplicity, we assume the cost of each billboard in $\ur$ is 1. For sake of illustration, we define two more notations: let $\lsmatrix[i][l]$ and $\gsmatrix[i][l]$ denote the sets of selected billboards corresponding to the influence value $\limatrix[i][l]$ and $\gimatrix[i][l]$ respectively. As a result we have four matrices as shown in Table~\ref{example-PBA}. Now we want to find an influential billboard set within $\budget=3$ by Algorithm \ref{PBA}.
	Initially, $\gimatrix[0][l]=0$, $0 < l \leq 3$. Clearly, for $l = 1,2...L$, $\gimatrix[1][l]$ is same as $\limatrix[1][l]$ and $\gsmatrix[1]$ is same as $\lsmatrix[1]$, as only one cluster is considered. When two clusters are considered:
		$\gimatrix[2][1]=max\{\gimatrix[1][1], \gimatrix[1][0]+\limatrix[2,1]\}=10$ and $\gsmatrix[1][1]=\{\lsmatrix[1][1]\}=\{1\}$; 
		$\gimatrix[2][2]=max\{\gimatrix[1][2], \gimatrix[1][0]+\limatrix[2,2], \gimatrix[1][1]+\limatrix[2,1]\}=18$ and $\gsmatrix[2][2]=\{\lsmatrix[2][2]\}=\{1,3\}$.
		$\gimatrix[2][3]=\gimatrix[1][2]+\limatrix[2,1]=26$ and $\gsmatrix[2][3]=\{\gsmatrix[1][2] \cup \lsmatrix[2][1]\}=\{1,3,6\}$.
This process is repeated until all the elements in $\gimatrix$ and $\gsmatrix$ are obtained. Finally, Algorithm~\ref{PBA} returns $\gsmatrix[3][3]=\{1,3,6\}$ as a solution, and its influence is $\gimatrix[3][3]=26$.
\end{example}

\begin{figure}[!t]
	\centering
	\includegraphics[width=0.60\linewidth]{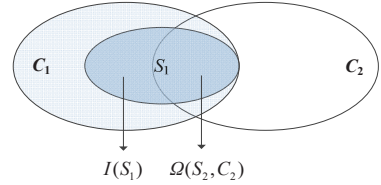}
		\vspace{-8pt}
	\caption{The relationship of $S_1$, $C_2$, and $\varOmega(S_1, C_2)$}
	\label{fig:lem_1}
\end{figure}

\noindent\textbf{Time Complexity Analysis}. Let $|C_i|$ be the cardinality of $C_i$. To obtain $\gimatrix[i][l]$, Algorithm~\ref{PBA} needs to invoke $\enumgreedy(C_i, l)$ to compute $\limatrix[i][l]$ and maximize $\gimatrix[i - 1][l - q] + \limatrix[i][q]$. When $\tau=2$, $\enumgreedy(C_i, l)$ takes $O(|\td|\cdot|C_i|^5)$ and there are $m\budget$ elements in $\gimatrix$. Therefore, the total time cost of Algorithm~\ref{PBA} is $\sum\nolimits_{i = 1}^{m} {|{C_i}|} ^5$ which is bounded by $O(m\budget \cdot |\td|\cdot|C_m|^5)$. It is more efficient than Algorithm~\ref{greedy_enumerate} ($O(|\td|\cdot|\ur|^{5})$), since $|C_m|$ is often significantly smaller than $|\ur|$ and $\budget$ is a constant. As shown in our experiment, \psel is faster than \enumgreedy by two orders of magnitude when $|\ur|$ is 2000.


\subsection{$\theta$-partition}\label{sec:partition}

A naive partition scheme will lead to poor quality due to large influence overlaps between clusters.  In order to reduce the influence overlap between the clusters, we introduce the concept of \textit{Overlap Ratio}. The basic idea is to control the maximum overlap ratio between any subset of a cluster and all the rest clusters.

\begin{defn} \rm{(\textbf{Overlap Ratio})}
\label{ratio}
	For two clusters $C_i$ and $C_j$, the ratio  of the overlap between $C_i$ and $C_j$ relative to $C_i$, denoted by $\ratio_{\svar{ij}}$, is defined as 
	\begin{equation} \label{ratio_eq}
	\ratio_{\svar{ij}} =\mathop {\arg \max }\limits_{\forall S_i \subseteq C_i} \{\varOmega(S_i|C_j)/\ifl (S_i)\}
	\end{equation}
	where $S_i$ is a subset of $C_i$, and $\varOmega(S_i|C_j)$ is the overlap between $S_i$ to $C_j$, i.e., $\ifl(S_i)+\ifl(C_j)-\ifl(S_i \cup C_j)$. The relationship of $S_i$, $C_j$ and $\varOmega(S_i|C_j)$ is illustrated in Figure~\ref{fig:lem_1}.
\end{defn}
Intuitively, the smaller $\ratio_{\svar{ij}}$ is, the lower influence overlap that $C_i$ and $C_j$ have. 

\textbf{Discussion on overlap ratio choices.}~Naturally, there are other ways to define the overlap ratio. Therefore, we describe two alternatives, and then discuss why the one defined in Definition~\ref{ratio} is generally a better choice. \\
\textit{Alternative 1.}~The influence overlap between the clusters can also be measured by the volume of the clusters' overlap directly, i.e.,  $\ratio_{\svar{ij}}=\varOmega(C_i,C_j)/\ifl(U)$. However, utilizing this measure to partition the billboards would incur a low performance for our partition based method and lazy probe method, especially when the budget $\budget$ is small. The reason is that, this measure does not reflect the overlap between two single billboards in different clusters, which may lead to the following situation: billboards $\bb_i$ and $\bb_j$ are partitioned into different clusters, while actually the trajectories influenced by $\bb_i$ can be fully covered by those trajectories influenced by $\bb_j$. Moreover, as our partition based method ignores the overlap between clusters, both $\bb_i$ and $\bb_j$ would be chosen as seeds while they actually have intense overlaps. Clearly, it is a grievous waste when the budget is limited.

\textit{Alternative 2.}~Another way is to measure the overlap ratio between billboards in one cluster and those that are not in this cluster, which can be described by the following equation:
\begin{equation} \label{ratio_eq_1}
\ratio_{\svar{i}} =\mathop {\arg \max }\limits_{{\bb_i} \in {C_i}} \frac{\ifl (\{\bb_i\}) + \ifl (\overline{C_i}) - \ifl (\overline{C_i} \cup \{\bb_i\})}{{\ifl (\{\bb_i\})}}
\end{equation}
where $\overline{C_i}=\ur \setminus C_i$.\\
If a partition $P$ satisfies $\ratio_{\svar{i}} \leq \theta$ ($\theta$ is a given threshold), for any $C_i \subseteq P$, \psel (\bbsel) can be approximated to be within a factor of $\theta(1-1/e)$. This statement holds because for any set $S \subseteq \ur$ and $S_i=S \cap C_i$, we have $\ifl(S) \geq \theta\sum_{S_i \subseteq S}{\ifl(S_i)}$ since the overlap between $S_i$ and $S\setminus S_i$ is at most $\theta \cdot \ifl(S_i)$. Moreover, the set $S'$ found by \psel (\bbsel) maximizes Equation \ref{recursion} and $S'_i$ is returned by a $(1-1/e)$-approximation algorithm (\enumgreedy), thus $\ifl(S')\geq \theta\sum_{S_i \subseteq S}{\ifl(S_i)} \geq \theta(1-1/e)\ifl(\opt)$.

Although this measure provides a good theoretical guarantee for our partition based method, it may cause that $\ur$ cannot be divided into a set of small yet balanced clusters due to its rigid constraint. As shown in Section~\ref{sec:partition_selection}, the time complexity of \psel depends on the size of the largest cluster in $P$, i.e.,  $|C_m|$. Therefore, if $|C_m|$ is close to $|\ur|$, the running time of \psel would be very high and even worse than our \enumgreedy baseline.

Given the overlap ratio, we present the concept of $\theta$-partition to trade-off between the cluster size and the overlap of clusters, 
where $\theta$ is a user-defined parameter to control the granularity of the partitions. 
\begin{defn} \rm{(\textbf{$\theta$-partition})}
	\label{theta_partition}
	Given a threshold $\theta$ $(0 \leq \theta \leq 1)$, we say a partition $P=\{C_1,..,C_m\}$ is a $\theta$-partition, if $\forall i,j\in[1,m]$ the overlap ratio $\ratio_{\svar{ij}}$ between any pair of clusters $\{C_i, C_j\}$ is less than $\theta$.   
\end{defn}

\begin{lem}
	\label{lem_for_guar}
	Let $P$ be a $\theta$-partition of $\ur$. Given any set $S \subseteq \ur$, and the billboards in $S$ belong to $k$ different clusters of $P$ in total. When $k \leq (1/\theta+1)$, we have $\ifl(S) \geq 1/2\sum\nolimits_{S_i \in S} {\ifl(S_i)}$, where $S_i=S \cap C_i$. 
\end{lem}
\begin{proof}
	To facilitate our proof, we assume $S=\{S_1,S_2...S_k\}$ and $\ifl(S_1) \geq,...,\geq \ifl(S_k)$. Let $\overline{\ifl(S)}$ denote the average influence among all $S_i \in S$, i.e., $\overline{\ifl(S)}=\frac{1}{k}\sum\nolimits_{j = 1}^k I(S_j)$. According to Definition \ref{theta_partition}, we observe that $\ifl(S_i \cup S_j) \geq \ifl(S_i) + (1-\theta)\ifl(S_j)$, as each subset of $S_j$ has at most $\theta$ percent of influence overlapping with the elements of $S_i$, or vice versa. Then for all subsets of $S$, we have:
	\begin{align*}		
	\ifl (S)  &  \geq  \ifl(S_1) + (1 - \theta )\ifl(S_2) +(1 - 2\theta )\ifl(S_3) ... + [1 - (k- 1)\theta]\ifl(S_k) \\
	&  =  \sum\nolimits_{i = 1}^k {\ifl(S_i)}  - \theta [\ifl(S_2) + 2\ifl(S_3)+... + (k - 1)\ifl(S_k)]\\
	&  = \sum\nolimits_{i = 1}^k {\ifl(S_i)}  - \theta [\sum\nolimits_{i = 2}^k {\ifl(S_i)} + \sum\nolimits_{i = 3}^k {\ifl(S_i)}+...+\ifl(S_k)]\\
	& \geq \sum\nolimits_{i = 1}^k {\ifl(S_i)}  - \theta [\overline {\ifl(S)} + 2\overline {\ifl(S)}+ ... + (k - 1)\overline {\ifl(S)}]\\
	&  =  \sum\nolimits_{i = 1}^k {\ifl(S_i)}  - \theta\frac{k(k - 1)}{2}\overline {\ifl(S)} 
	\end{align*}
	The second inequality above follows from the fact that $\overline {\ifl(S)} \geq \frac{1}{k-1}\sum\nolimits_{i = j}^k {\ifl(S_i)}$ for $j=2,3...k$, because we have assumed $\ifl(S_1) \geq,...,\geq \ifl(S_k)$.
	As $k \leq \frac{1}{\theta}+1$, we have $\theta \frac{{k(k - 1)}}{2}\overline {\ifl(S)} \leq \frac{{k}}{2}\overline {\ifl(S)}$ and
	\begin{align*}		
	\ifl (S)  & \geq \sum\nolimits_{i = 1}^k {\ifl(S_i)}  - \frac{k}{2}\overline {\ifl(S)} = 1/2\sum\nolimits_{S_i \in S} {\ifl(S_i)}
	\end{align*}
\end{proof}

Based on Lemma~\ref{lem_for_guar}, we proceed to derive the approximation ratio of Algorithm \ref{PBA} in Theorem \ref{PBA_approximation}.

\begin{thm}\label{PBA_approximation}
	Given a $\theta$-partition $P=\{C_1,..,C_m\}$, Algorithm~\ref{PBA} obtains a $\frac{1}{2}^{\lceil \log _{(1 + 1/\theta )}m\rceil}(1-1/e)$-approximation to the \problem problem. 
	
	\begin{proof}
		Let $S^*$ and {$S=S_1\cup S_2\cup\cdots\cup S_k$} $S$ be the solution returned by Algorithm \ref{greedy_enumerate} and Algorithm \ref{PBA} respectively, where $S_i=S \cap C_i$ and $i \leq k \leq m$, i.e., $S=S_1\cup S_2\cup\cdots\cup S_k$. 
		
		When $\theta=0$, we have $\ifl(S)=\sum\nolimits_{i = 1}^k {\ifl(S_i)}$. As $\sum\nolimits_{i = 1}^k {\ifl(S_i)}$ is the maximum value of Equation \ref{PBIM}, thus $\ifl(S)=\sum\nolimits_{i = 1}^k {\ifl(S_i)} \geq \ifl(S^*)$.
		Moreover, $\ifl(S^*) \geq (1-1/e)\ifl(\opt)$ since $S^*$ is returned by Algorithm \ref{greedy_enumerate}, thus $\ifl(S) \geq (1-1/e)\ifl(\opt)$ and the theorem holds.

		When $\theta > 0$, we have $\ifl(S) \leq \sum\nolimits_{i = 1}^k {\ifl(S_i)}$. In this case, let us consider an iterative process. 
		At iteration $0$, we denote $S^0$ as a set of billboard clusters in which cluster $S_j^0$ corresponds $S_j$. In each iteration $h$, we arbitrarily partition the clusters in $S^{h-1}$ and merge each partition to form new disjoint clusters for $S^h$. Each cluster $S^h_j$ in $S^h$ contains at most $(1+1/\theta)$ clusters from $S^{h-1}$. We note that the clusters in $S^h$ are always $\theta$-partitions since each cluster is recursively merged from $S^{h-1}$ and the clusters in $S^0$ are $\theta$-partitions. Thus, according to Lemma \ref{lem_for_guar}, we have the invariant $\ifl(S^h_j)\geq 1/2\sum\nolimits_{S^{h-1}_x \in S^{h}_j} \ifl(S^{h-1}_x)$. The iterative process can only repeat for $d$ times until no clusters can be merged.	
		Intuitively, $d$ should not exceed $\lceil\log _{(1 + 1/\theta )}m\rceil$ (as $k \leq m$) and thus $\ifl(S^d)\geq \frac{1}{2}^{\lceil\log _{(1 + 1/\theta )}m\rceil}\sum\nolimits_{i = 1}^k {\ifl(S_i)}$. Moreover, $S^d$ only has one billboard set $S_1^d$, then $\ifl(S^d)=\ifl(S_1^d)$ and $S_1^d$ is equal to $S$. Therefore, we have $\ifl(S)\geq \frac{1}{2}^{\lceil\log _{(1 + 1/\theta )}m\rceil}\sum\nolimits_{i = 1}^k {\ifl(S_i)}$. Moreover, as $\sum\nolimits_{i = 1}^k {\ifl(S_i)} \geq \ifl(S^*) \geq (1-1/e)\ifl(\opt)$, we conclude that $\ifl(S)\geq\frac{1}{2}^{\lceil\log _{(1 + 1/\theta )}m\rceil}(1-1/e)\ifl(\opt)$.	
	\end{proof}
\end{thm} 

Figure~\ref{fig:prove_thm42} presents a running example to explain the iteration process in the above proof. In this example, $S$ contains 9 clusters and $\theta=0.5$. At iteration 0, $S^0$ is initialized by $S$;  At iteration 1, each cluster of $S^1$ is generated by randomly merging $(1/\theta+1=3)$ clusters in $S^0$, i.e., $S_1^1=\{S_1^0 \cup S_2^0 \cup S_3^0\}$. According to Lemma \ref{lem_for_guar}, we have $\ifl(S_j^1) \geq 1/2\sum\nolimits_{S_x^0 \in S_j^1} {\ifl(S_x^0)}$, i.e., $\ifl(S_1^1) \geq \frac{1}{2}(\ifl(S_1^0)+\ifl(S_2^0)+\ifl(S_3^0))$. As $S^1$ only contains $(1/\theta+1=3)$ clusters, the second iteration merges all the clusters in $S^1$ into one cluster $S_1^2$. Since $\ifl(S_1^2) \geq 1/2\sum\nolimits_{S_x^1 \in S_j^2} {\ifl(S_x^1)}$, we have $\ifl(S_1^2) \geq 1/4\sum\nolimits_{j=1}^9 {\ifl(S_j^0)}$.

\begin{figure}
	\centering
	\includegraphics[width=0.85\linewidth]{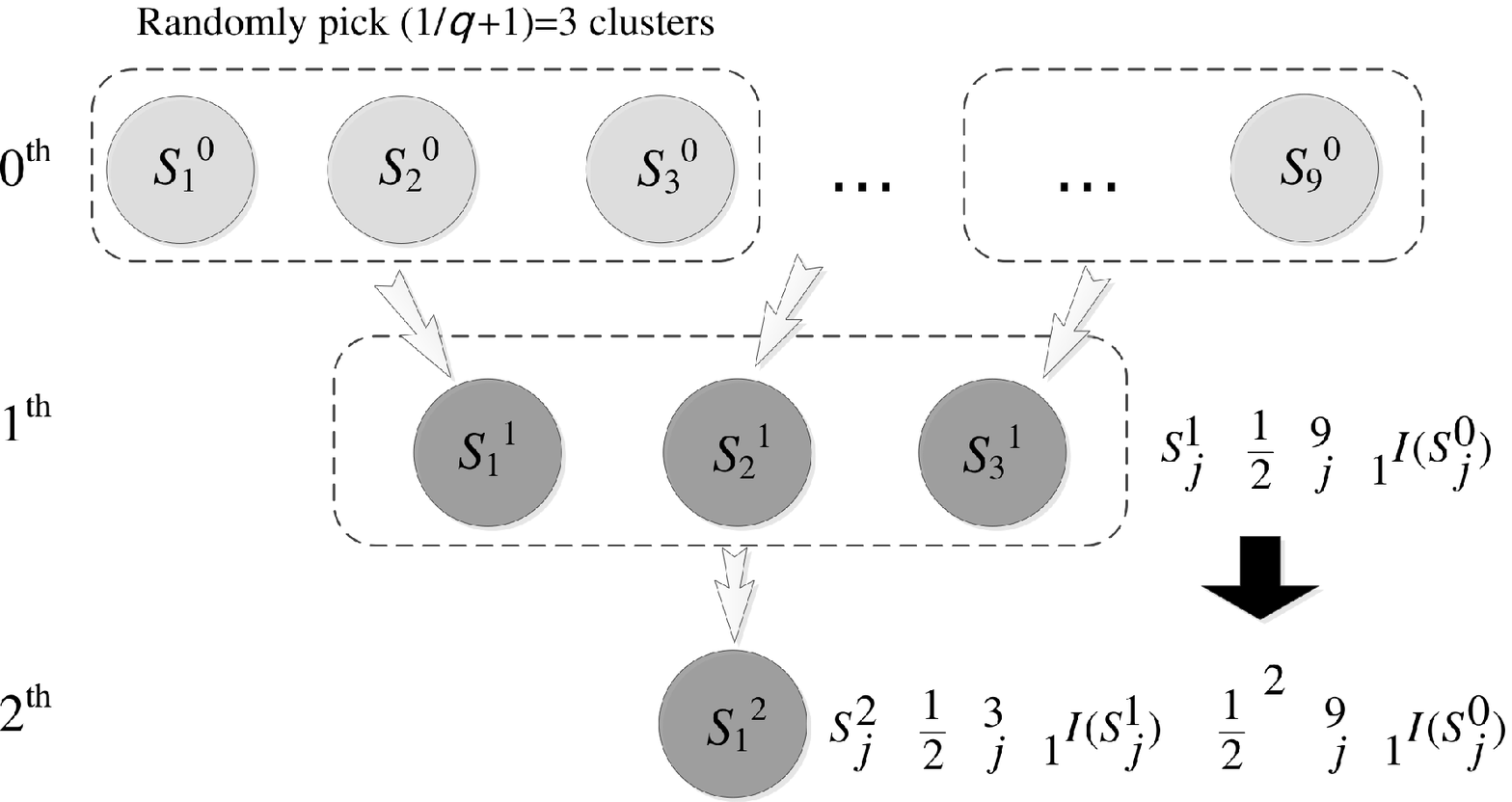}
	\caption{An example to explain the proof of Theorem \ref{PBA_approximation}}
	\label{fig:prove_thm42}
\end{figure}


\subsection{Finding a $\theta$-partition}\label{sec:findpartition}

It is worth noting that there may exist multiple $\theta$-partitions of $\ur$ (e.g., $\ur$ is a trivial $\theta$-partition). 
Recall Section~\ref{sec:partition_selection}, 
the time complexity of the partition based method (Algorithm~\ref{PBA}) is $O(m\budget \cdot |\mathcal{\td}|\cdot|C_m|^5)$, 
where $|C_m|$ is the size of the largest cluster in a partition $P$. Therefore, $|C_m|$ is an indicator of how good a $\theta$-partition is, and we want to minimize $|C_m|$. Unfortunately, finding a good $\theta$-partition is not trivial, since the it can be modeled as the balanced $k$-cut problem where each vertex in the graph is a billboard and each edge denotes two billboards with influence overlap, which is found to be NP-hard \cite{Wagner1993}.  Therefore, we use an approximate $\theta$-partition  by employing a hierarchical clustering algorithm~\cite{murtagh1983survey}.  It first initializes each billboard as its own cluster, then it iteratively merges these two clusters into one, if their overlap ratio (Equation \ref{ratio_eq}) is larger than $\theta$. That is, for each pair of clusters $C_i, C_j \subseteq U$, if $\ratio_{\svar{ij}}$ is larger than $\theta$, then $C_i$ and $C_j$ will be merged. By repeating this process, an approximate $\theta$-partition is obtained when no cluster in $U$ can be merged. 

Note that how to efficiently get a $\theta$ partition is not the key point of this paper and it can be processed offline; while our focus is how to find the influential billboards based on a $\theta$-partition.

\section{Lazy Probe}\label{sec:bound-selection}

We first propose our lazy probe algorithm to reduce the number of calls to $\enumgreedy(C_i,l)$ in Section~\ref{sec:bound-algo}, and then establish the theoretical equivalence on approximation ratio between \bbsel and \enumgreedy in Section~\ref{sec:bound-analysis}. 

\subsection{The Lazy Probe Algorithm}\label{sec:bound-algo}


\setlength{\algomargin}{1.2em} 	
\begin{algorithm}[!t]
		\caption{LazyProbe$(P, \budget)$}
		\label{bound_based}
		\begin{small}
			
		{\bf Input:} A $\theta$-partition $P$ of $U$, budget $\budget$
			
		{\bf Output:} A billboard set $S$
	
Initialize two matrices $\gimatrix$ and $\limatrix$ 


\For {$i=1$ to $m$}
{	
	
	\For{$l=1$ to $\emph\budget$}
	{
		$\gimatrix^\downarrow[i][l] \gets \gimatrix[i-1][l]$
		
		
		\For{$q=1$ to $l$}
		{
		
	        $\upperlimatrix[i][q] \gets \textbf{\estbod}(C_i,q)$ 
			
			
			\uIf{$\gimatrix^\downarrow[i][l] \leq \gimatrix[i-1][l-q] +\upperlimatrix[i][q]$}
			{
				
				\uIf{$\limatrix[i][q]$ has not been computed}{
				
					
					 Invoke $\textbf{\enumgreedy}(C_i,q)$ to compute $\limatrix[i][q]$ 		
				}
				
				\textbf{Update} $\gimatrix^\downarrow[i][l]$ \textbf{by} $\gimatrix[i-1][l-q]+\limatrix[i][q]$
				
			}
			\Else 
			{	
				contiune;	
			}	
		}
		$\gimatrix[i][l] \gets \gimatrix^\downarrow[i][l]$
		
	}		
}

$S \leftarrow$ the corresponding selected set of $\gimatrix[i][l]$

\Return {$S$}
		\end{small}
	\end{algorithm}

\begin{algorithm}[!t]
\SetAlgorithmName{Function}{Function}{Function}
\caption{EstimateBound($\ur, \budget$)} 
\label{mgreedy}
\begin{small}
{\bf Input:} A billboard set $\ur$, a budget $\budget$

{\bf Output:} An influence estimator $\upperlimatrix[i][q]$

		
		$S'=$ \gre$(U,L)$;
		
		$b_{k+1} $ is the next billboard with the largest unit marginal influence;
		
		 
		 			
		 
		
		
		$\upperlimatrix[i][q]=\ifl(S')+\frac{\marginifl(b_{k+1}|S')}{L-\cost(S')}$;

\Return $\upperlimatrix[i][q]$; 	
	\end{small}
\end{algorithm}

Recall that $\gimatrix[i][l]$ is the maximum influence value that can be attained with a budget not exceeding $l$ using up to the first $i$ clusters (in Section~\ref{sec:partition_selection}), and all clusters are processed in an order of their size (from the smallest to the largest by Definition~\ref{partition}). As mentioned in Section \ref{sec:partition_method}, 
$\gimatrix[i][l]= \mathop {\max }_{0 \le \svar{q} \le l} (\gimatrix[i - 1][l - q] + \limatrix[i][q] )$, we need to find a $q$ ($0\leq q \leq l$) to maximize this influence. Note that $\gimatrix[i - 1][l - q]$ can be easily gotten in the previous computation, but it is expensive to compute $\limatrix[i][q]$ by calling algorithm \enumgreedy. To address this issue, instead of computing the exact influence $\limatrix[i][q]$ in cluster $C_i$, we can estimate an upper bound of $\limatrix[i][q]$ for $0 \leq q \leq l$ (denoted by $\upperlimatrix[i][q]$), and then prune the $q$ that cannot get larger influence by bound comparison. 

Algorithm~\ref{bound_based} describes how our method works. 
Similar to \psel, we employ a dynamic programming approach to compute the selected billboard set and its influence value for each cluster $i$ and each cost $l$. 
However, the difference is that we first compute the lower bound $\gimatrix^\downarrow[i][l]$ of $\gimatrix[i][l]$. Obviously $\gimatrix^\downarrow[i][l]=\gimatrix[i-1][l]$  is a naive lower bound by setting $q=0$ (line 4.6). Initially, when $i=1$, for all $l \leq \budget$, we have $\gimatrix^\downarrow[i-1][l]=\gimatrix[0][l]=0$.
Then we compute an upper bound $\upperlimatrix[i][q]$ from $q=0$ to $q=l$ by calling function \texttt{EstimateBound}, which will be discussed later. 
Next if $\gimatrix^\downarrow[i][l]\geq\gimatrix[i - 1][l - q] + \upperlimatrix[i][q]$, we do not need to compute $\limatrix[i][q]$, because we cannot increase the influence using cluster $C_i$, and thus we can save the cost of calling \enumgreedy (lines 4.13-4.14). If $\gimatrix^\downarrow[i][l]<\gimatrix[i - 1][l - q] + \upperlimatrix[i][q]$, we need to compute $\limatrix[i][q]$, by calling $\enumgreedy(C_i,q)$, and  update $\gimatrix^\downarrow[i][l]=\gimatrix[i-1][l-q]+\limatrix[i][q]$  (lines 4.9-4.12).
Finally, we set $\gimatrix[i][l]$ as $\gimatrix^\downarrow[i][l]$ since we already know $\gimatrix^\downarrow[i][l]$ is good enough to obtain the solution with a guaranteed approximation ratio (line 4.16), and return the corresponding selected billboard set as $S$ (line 4.17).

\noindent {\bf Estimation of Upper Bound  $\upperlimatrix[i][q]$.} A key challenge to ensure the approximation ratio of \bbsel is to get a tight upper bound $\upperlimatrix[i][q]$. Unfortunately, we observe that it is hard to obtain a tight upper bound efficiently due to the overlap influence among billboards. Fortunately, we can get an approximate bound, with which our algorithm can still guarantee the $(1-1/e)$ approximation ratio (see Section~\ref{sec:bound-analysis}). 
To achieve this goal, we first utilize the basic greedy algorithm \gre to select the billboards $S'=\{b_1,b_2,\ldots, b_k\}$. Let $b_{k+1}$ be the next billboard with the maximal marginal influence. If we include $b_{k+1}$ in the selected billboards, then the cost will exceed $L$. If we do not include it, we will lost the cost of $L-\cost(S')$ where $\cost(S')=\sum_{1\leq i\leq k}\cost(b_i)$. Then we can utilize the unit influence of $b_{k+1}$ to remedy the lost cost, and thus we can get an upper bound $\upperlimatrix[i][q]=\ifl(S')+\frac{\ifl(S'\cup\{b_{k+1}\})-\ifl(S')}{L-\cost(S')}$. 
We later show that $\upperlimatrix[i][q]\geq (1-1/e) \limatrix[i][q]$.  
Moreover, the solution quality of Algorithm~\ref{bound_based} remains the same as Algorithm~\ref{PBA}. The details of the theoretical analysis will be presented in Section~\ref{sec:bound-analysis}.


\begin{figure*}[!t]
\vspace{-1em}
	\centering
	\includegraphics[width=0.75\linewidth]{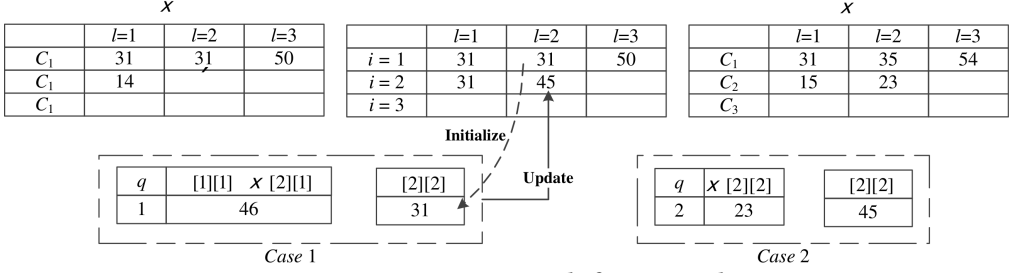}
		 \vspace{-10pt}
	\caption{A running example for \bbsel 
	}

	\label{fig:bbsel_example}
\end{figure*}

\begin{example}
	Figure~\ref{fig:bbsel_example} shows an example on how Algorithm~\ref{bound_based} works. There are three clusters $C_1$, $C_2$ and $C_3$ in a partition $P$, and the estimator matrix $\upperlimatrix$ is shown in upper right corner. When $C_i$ ($i=1,2,3$) is considered, it computes $\gimatrix[i][l]$, for each $l=1,..,\budget$, by the bound comparisons. Taking $\gimatrix[2][2]$ as an example, Algorithm \ref{bound_based} first initializes $\gimatrix^\downarrow[2][2]=\gimatrix[1][2]$ and then computes $\gimatrix[1][2-q]+\upperlimatrix[2][q]$ ($q=1,2$ ) for bound comparisons. For case 1 ($q=1$), as $\gimatrix^\downarrow[2][2] \leq \gimatrix^\uparrow[2][2]$, Algorithm~\ref{bound_based} needs to compute $\limatrix[2][1]$ by invoking \enumgreedy and update $\gimatrix^\downarrow[2][2]$ by $\gimatrix[1][1]+\limatrix[2][1]$. For case 2 ($q=2$), since $\gimatrix^\downarrow[2][2]\geq \gimatrix^\uparrow[2][2]$, $\gimatrix^\downarrow[2][2]$ does not need to be updated and finally $\gimatrix[2][2]=\gimatrix^\downarrow[2][2]=45$. 
\end{example}

\ping{The complexity of \bbsel is the same as \psel at the worst case, but the pruning strategy actually can work well and reduce the running time greatly (as evidenced in Section~\ref{sec:exp}).}

\subsection{Theoretical Analysis}\label{sec:bound-analysis}

In this section, we conduct theoretical analysis to establish the equivalence between \bbsel and \psel in terms of the approximation ratio.
We first show that if the bound $\upperlimatrix[i][l]$ in \bbsel is $(1-1/e)$ approximate to $\problem$ instance of billboards in cluster $i$ using budget $l$,
then the approximation ratio of $\bbsel$ and $\psel$ is the same (Theorem~\ref{theor_bound}). We then move on to show that $\upperlimatrix[i][l]$ is indeed $(1-1/e)$-approximate in Lemmas~\ref{gred_bound}-\ref{p_bound_greedy}.

\begin{thm}\label{theor_bound}
If $\upperlimatrix[i][l]$ obtained by Algorithm~\ref{mgreedy} achieves a $(1-1/e)$ approximation ratio to the \problem instance
for cluster $i$ with budget $l$, \bbsel ensures the same approximation ratio with \psel presented in Section~\ref{sec:partition_selection}.
\vspace{-.5em}
\begin{proof}
Let $C_i$ denote the the $i$th cluster considered by \bbsel and $U_i=\bigcup_{j=1}^i{C_j}$.
To prove the correctness of this theorem, we first prove $\gimatrix[i][l] \geq (1-1/e)\ifl(\opt_{U_i}^l)$ for all $i \leq m$ and $l \leq \budget$, where $\opt_{U_i}^l$ is the optimal solution of the \problem instance for billboard set $U_i$ with budget $l$. Clearly, if this assumption holds, we have $\sum\nolimits_{i = 1}^k {\ifl(S_i)} \geq (1-1/e)\ifl(\opt)$ ($\opt$ is the globally optimal solution). $S=\{S_1,...,S_k\}$ is the solution returned by \bbsel. $S_i=S \cap C_i$.

We prove it by mathematical induction. When $i=0$, this assumption holds immediately. When $i>0$, we assume that the assumption holds for the first $i$th recursion, and prove it still holds for the ($i+1$)th recursion. According to the definition, we have $\upperlimatrix[i+1][l] \geq (1-1/e)\ifl(\opt_{C_{i+1}}^l)$. Moreover, we have already assumed $\gimatrix[i][l] \geq (1-1/e)\ifl(\opt_{U_i}^l)$ ($l=1,...,\budget$), thus $\gimatrix[i+1][l]=\mathop { \max }\limits_{0 \le q \le l} \{\gimatrix[i][l - q] + \upperlimatrix[i+1][q]\} \geq (1-1/e)\mathop{ \max }\limits_{0 \le q \le l} \{\ifl(\opt_{U_i}^{l-q})+ \ifl(\opt_{C_{i+1}}^q)\}$. Moreover, as $\mathop{ \max }\limits_{0 \le q \le l} \{\ifl(\opt_{U_i}^{l-q} )+ \ifl(\opt_{C_{i+1}}^q)\} \geq \ifl(\opt_{U_{i+1}}^l)$, we have $\gimatrix[i+1][l] \geq (1-1/e)\ifl(\opt_{U_{i+1}}^l)$. The assumption gets proof.

Since $S$ comes from $k$ clusters and $k \leq m$, we can conclude that $\ifl(S) \leq \frac{1}{2}^{\lceil \log _{(1 + 1/\theta)}m\rceil}\sum\nolimits_{i = 1}^k {\ifl(S_i)}$. We omit to prove this conclusion here since the proof is similar to that of Theorem \ref{PBA_approximation}. Moreover, as $\sum\nolimits_{i = 1}^k {\ifl(S_i)} \geq (1-1/e)\ifl(\opt)$ (the assumption holds), thus $\ifl(S) \leq \frac{1}{2}^{\lceil \log _{(1 + 1/\theta)}m\rceil}(1-1/e)\ifl(\opt)$.
\end{proof}
\end{thm}

Theorem~\ref{theor_bound} requires that $\upperlimatrix[i][l]$ is $(1-1/e)$-approximate to the corresponding \problem instance in a small cluster. 
To show that $\upperlimatrix[i][l]$ returned by Algorithm~\ref{mgreedy} satisfies such requirement, we 
describe the following hypothetical scenario for running Algorithm~\ref{greedy} on the 
\problem instance for cluster $C_i$ and budget $l$. 
Let $\bb_{k^*+1}$ be a billboard in the optimal solution set, and it is the first billboard that violates the budget constraint in  Algorithm~\ref{greedy}. 
The following inequality holds \cite{khuller1999budgeted}.

\begin{lem}\label{gred_bound} \cite{khuller1999budgeted}
	\label{khuller}
After the $i$th iteration ($i=1,...,k^*+1$) of the hypothetical scenario running Algorithm~\ref{greedy}, the following holds:
	\begin{equation}
	I({S_i}) \ge [1 - \prod\nolimits_{j = 1}^i {(1 - \frac{\emph\cost(\emph\bb_j)}{\emph\budget})} ] \cdot 
	\ifl(\opt)
	\end{equation}
	Where $S_i$ be the billboard set that is selected by the first $i$ iterations of the hypothetical scenario.
\end{lem}

With lemma~\ref{khuller}, we analyze the solution quality of running the hypothetical scenario by using the $k^*+1$ billboards to deduce $\upperlimatrix[i][l]$. 

\begin{lem}\label{hypo_bound}
Let $\mathcal{M}_{k^*+1}$ denote the unit marginal influence of adding $\emph\bb_{k^*+1}$ in the hypothetical scenario, i.e., 
$\mathcal{M}_{k^*+1} = [\ifl(S_{k^*} \cup \{\emph\bb_{k^*+1}\}) - \ifl(S_{k^*})]/ w(\emph\bb_{k^*+1})$.
Then $\ifl{(S_{k^*})+[\emph\budget-\emph\cost(S_{k^*})]\cdot \mathcal{M}_{k^*+1}} \geq (1-1/e)\ifl(\emph\opt)$. 
\vspace{-.5em}
\begin{proof}
	 First, we observe that for $a_1,...,a_n \in \mathbb{R}^+$  such that $\sum\nolimits_{i=1}^n {a_i}=\alpha A$, the function
	 \begin{equation*}
	 	1-\prod\nolimits_{i=1}^n(1-\frac{a_i}{A})
	 \end{equation*}
	 achieves its minimum of $1-(1-\alpha/n)^n$ when $a_1=a_2=...=a_n=\alpha A/n$.
	 
	 Suppose $\bb'$ is a virtual billboard with cost $\budget-\cost(S_{k^*})$ and the unit marginal influence of $\bb'$ to $S_i$ is $\mathcal{M}_{k^*+1}$, for $i=1,...,k^*$. We modify the instance by adding $\bb'$ into $\ur$ and let $\ur \cup \{\bb'\}$ be denoted by $\ur'$. Then after the first $k^*$th iterations of Algorithm \ref{greedy} on this new instance, $\bb'$ must be selected at the $(k^*+1)$th iteration. As $\budget(S_{k^*})+\cost({\bb'})=\budget$, by applying Lemma \ref{khuller} and the observation to $I(S')$ $(S'= {S_k} \cup \{\bb'\})$, we get:
\vspace{-.6em}	 
	\begin{align*}
	\small		
	I((S') & \ge \left[ 1 -  \prod\nolimits_{j = 1}^{k^*+1} \left( {1 - \frac{w(\bb_j)}{L}} \right) \right] \cdot \ifl(\opt') \\
	& \ge \left( 1 - {(1 - \frac{1}{k^* + 1})^{k^* + 1}}\right)  \cdot \ifl(\opt')\\
	&  \ge (1 - \frac{1}{e}) \cdot \ifl(\opt')
	\end{align*} 
	
    Note that $\ifl(\opt')$ is surely larger than $\ifl(\opt)$, thus $\ifl(S_{k^*})+[\budget-\cost(S_{k^*})]\mathcal{M}_{k^*+1}=\ifl(S')\geq (1 - \frac{1}{e}) \cdot \ifl(\opt') \geq (1 - \frac{1}{e}) \cdot \ifl(\opt)$.
\end{proof}
\end{lem}

Finally, we show that the estimator obtained by Algorithm~\ref{mgreedy} is larger than the bound value obtained by 
the hypothetical scenario described in Lemma~\ref{hypo_bound}, which indicates that Algorithm~\ref{mgreedy} is $(1-1/e)$-approximate and it further implies Theorem~\ref{theor_bound} hold.

\begin{lem}	\label{p_bound_greedy}
Given an instance of \problem. Let $\limatrix[i][l]$ be an estimator returned by Algorithm~\ref{mgreedy}, we have $\limatrix[i][l] \geq (1-1/e)\ifl(\opt)$.
	\begin{proof}
	We observe that $\mathcal{M}_{k^*+1}$ cannot be larger than $\mathcal{M}_{k+1}$ and $\ifl(S_k)+(\budget-\cost(S_k))\mathcal{M}_{k+1} \geq \ifl(S_{k^*})+(\budget-\cost(S_{k^*}))\mathcal{M}_{k^*+1}$. Moreover, Lemma \ref{hypo_bound} shows $\ifl(S_{k^*})+(\budget-\cost(S_{k^*}))\mathcal{M}_{k^*+1} \geq (1-1/e)\ifl(\opt)$, so $\ifl(S_k)+(\budget-\cost(S_k))\mathcal{M}_{k+1} \geq (1-1/e)\ifl(\opt)$. As $\limatrix[i][l] = \ifl(S)+[\budget-\cost(S)]\mathcal{M}_{k+1}$ (Algorithm \ref{mgreedy} line 5.5), this lemma is proved.
	\end{proof}
\end{lem}

\section{Experiments}\label{sec:exp}


\subsection{Experimental Setup} 

\noindent\textbf{Datasets.}~We collect billboards and trajectories data for the two largest cities in US, i.e., NYC and LA. 

1) \textbf{Billboard} data is crawled from LAMAR\footnote{http://www.lamar.com/InventoryBrowser}, one of the largest outdoor advertising companies worldwide. 
%

2) \textbf{Trajectory} data is obtained from two types of real datasets: the TLC trip record dataset\footnote{http://www.nyc.gov/html/tlc/html/about/trip\_record\_data.shtml} for NYC and the Foursquare check-in dataset\footnote{https://sites.google.com/site/yangdingqi/home} for LA. 
For NYC, we collect TLC trip record containing green taxi trips from Jan 2013 to Sep 2016. Each individual trip record includes the pick-up and drop-off locations, time and trip distances. We use Google maps API\footnote{https://developers.google.com/maps/} to generate the trajectories, and if the distance of the recommended route by Google is close to the trip distance and travel time in the original record (within 5\% error rate), we use this route as an approximation of this trip's real trajectory. As a result, we obtain 4 million trajectories for trip records as our trajectory database. 
For LA, as there is no public taxi record, we collect the Foursquare checkin data in LA, and generate the trajectories using Google Map API by randomly selecting the pick-up and drop-out locations from the checkins. 

The statistics of those datasets are shown in Table~\ref{tab:datasets}, the distribution of trajectories' distance is shown in Figure~\ref{fig:exp-disdistr}, and a snapshot of the billboards' locations in NYC is shown in Figure~\ref{fig:lamar}. We can find that over 80\% trips finish in 5 kilometers.


\noindent\textbf{Algorithms.}~As mentioned in Section~\ref{sec:intro} and Section~\ref{sec:relatedwork}, this is the first work that studies the \problem problem, there exists no previous work for direct comparison.
In particular, we compare five methods:  \topk which picks billboards by a descending order of the volume of trajectories meeting those billboards within the budget $\budget$; a basic greedy \gre (Algorithm~\ref{greedy}); a Greedy Enumeration method \enumgreedy (Algorithm~\ref{greedy_enumerate}); our partition based method \psel (Algorithm~\ref{PBA}); our lazy probe method \bbsel (Algorithm~\ref{bound_based}).  Note that \enumgreedy is too slow to converge in 170 minutes even for a small dataset (because the complexity of \enumgreedy is proportional to $|\ur|^5$). Thus in our default setting we do not include \enumgreedy. Instead we add one experiment on a smaller dataset of NYC to evaluate it in Section \ref{exp:comley}. 

\noindent\textbf{Performance Measurement.}~
We evaluate the performance of all methods by the runtime and the  influence value of the selected billboards. Each experiment is repeated 10 times, and the average performance is reported.

\noindent\textbf{Billboard costs.}~
Unfortunately, all advertising companies do not provide the exact leasing cost; instead, they provide a range of costs for a suburb. E.g., the costs of billboards in New Jersey-Long Island by LAMAR range from 2,500 to 14,000 for 4 weeks~\cite{adunitcost}. So we generate the cost of a billboard $b$ by designing a function proportional to the number of trajectories influenced by $b$: $\cost(\bb)= \lfloor\beta\cdot\ifl(\bb)/100\rfloor \times 1000$, where $\beta$ is a factor chosen from 0.8 to 1.2 randomly to simulate various cost/benefit ratios. Here we compute the cost w.r.t. $|\td|=200k$ trajectories.

\begin{table}[!t]
	\caption{Statistics of Datasets.}
	\centering
	\vspace{-1em}
	\begin{small}
		\label{tab:datasets}
		\begin{tabular}[c]{|c|c|c|c|c|c|}
			\hline
			&$|\td|$			   & $|U|$		             & AvgDistance         & AvgTravelTime   & AvgPoint\# \\ \hline
			NYC    		 &4m     				&1500       			 &2.9km   					 &569s   &159  \\ \hline
			LA    		 &200k      			&2500      				 &2.7km    					 &511s    &138  \\ \hline
		\end{tabular}
	\end{small}
\end{table}
\begin{table}[!t]\vspace{.5em}
	\caption{Parameter setting.}	\vspace{-1em}
	\centering
	\begin{small}
		\label{exp-param}
		\begin{tabular}[r]{|p{2cm}<{\centering}|p{5.5cm}<{\centering}|}
			\hline
			\multicolumn{1}{|c|}{\textbf{Parameter}}                                                      & \multicolumn{1}{c|}{\textbf{Values}} \\ \hline
			$\budget$                                                               &100k, \textbf{150k}, ... 300k\\ \hline
			$|\td|$ (NYC)                                                          &  40k, ...,\textbf{120k}..., 4m    \\ \hline
			$|\td|$ (LA)                                                         &  40k, 80k,\textbf{120k}, 160k, 200k    \\ \hline
			\begin{tabular}[c]{@{}l@{}}$|U|$ (NYC) \end{tabular}					 & 0.5k, 1k, 1.46k, (\textbf{2k}...10k by replication) \\ \hline
			\begin{tabular}[c]{@{}l@{}}$|U|$ (LA)\end{tabular}					 &1k, \textbf{2k}, 3k, (4k... 10k by replication) \\ \hline
			$\theta$                                                       & 0, 0.1, \textbf{0.2}, 0.3, 0.4 \\ \hline
			$\lambda$                                                       & 25m, \textbf{50m}, 75m, 100m \\ \hline
		\end{tabular}
	\end{small}
\end{table}

\noindent\textbf{Choice of influence probability} $\pr(\bb_i, \tr_j)$. In Section~\ref{sec:pf}, we define two choices to compute $\pr(\bb_i, \tr_j)$. By default, we use the first one. 
The experimental result of the second can be found in Section~\ref{sec:2nd_probability}.

\noindent
\textbf{Parameters.}~
Table~\ref{exp-param} shows the settings of all parameters, such as the distance threshold $\lambda$ to determine the influence relationship between a trajectory and a billboard, the threshold for $\theta$-partition, the budget $\budget$ and the number of trajectories $|\td|$. The default one is highlighted in bold; we vary one parameter while the rest parameters are kept default in all experiments unless specified otherwise. Since the total number of real-world billboards in LAMAR is limited (see Table~\ref{tab:datasets}), the $|\ur|$ larger than the limit are replicated by random selecting locations in the two cities.

\noindent
\textbf{Setup.}~All codes are implemented in Java, and experiments are conducted on a server with 2.3 GHz Intel Xeon 24 Core CPU and 256GB memory running Debian/4.0 OS.

\begin{figure}[!t]\vspace{-1em}
	\centering
	\subfloat[Trajectory Distance]{\includegraphics[clip,width=0.24\textwidth]{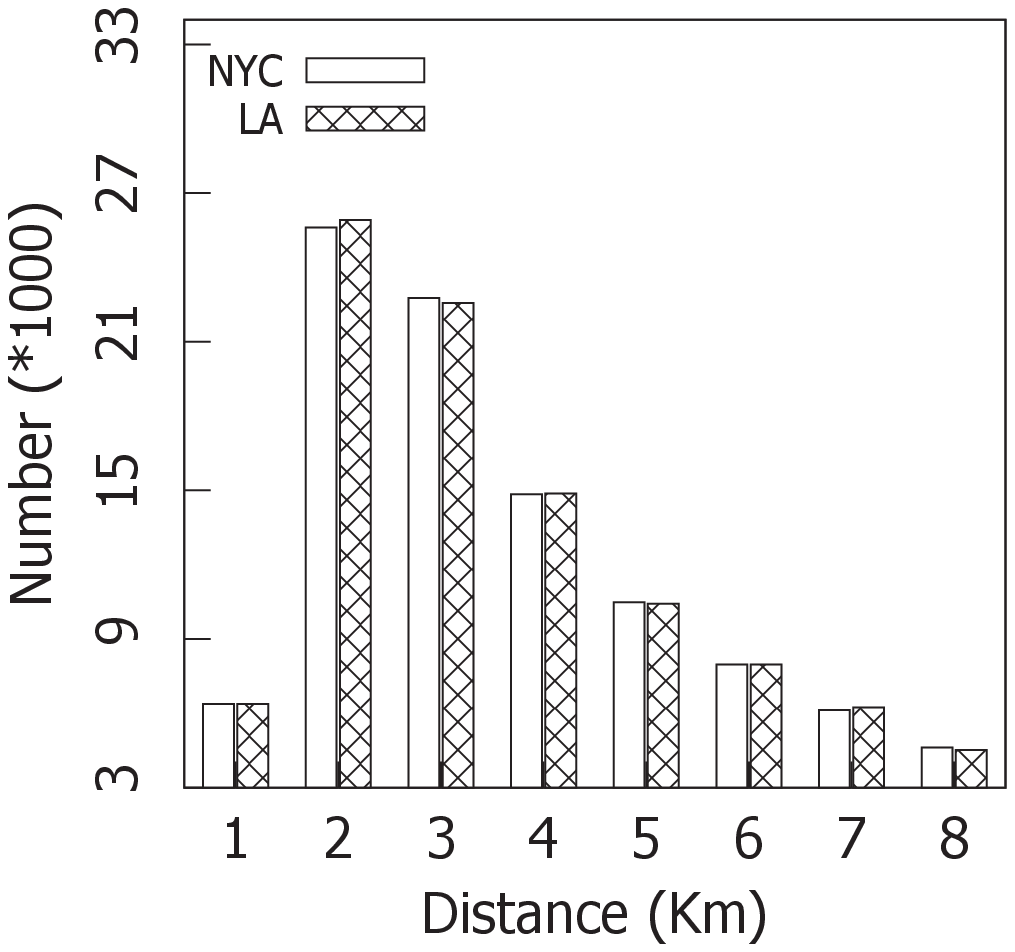}\label{fig:exp-disdistr}}
	\subfloat[Billboard Location ]{\includegraphics[clip,width=0.24\textwidth]{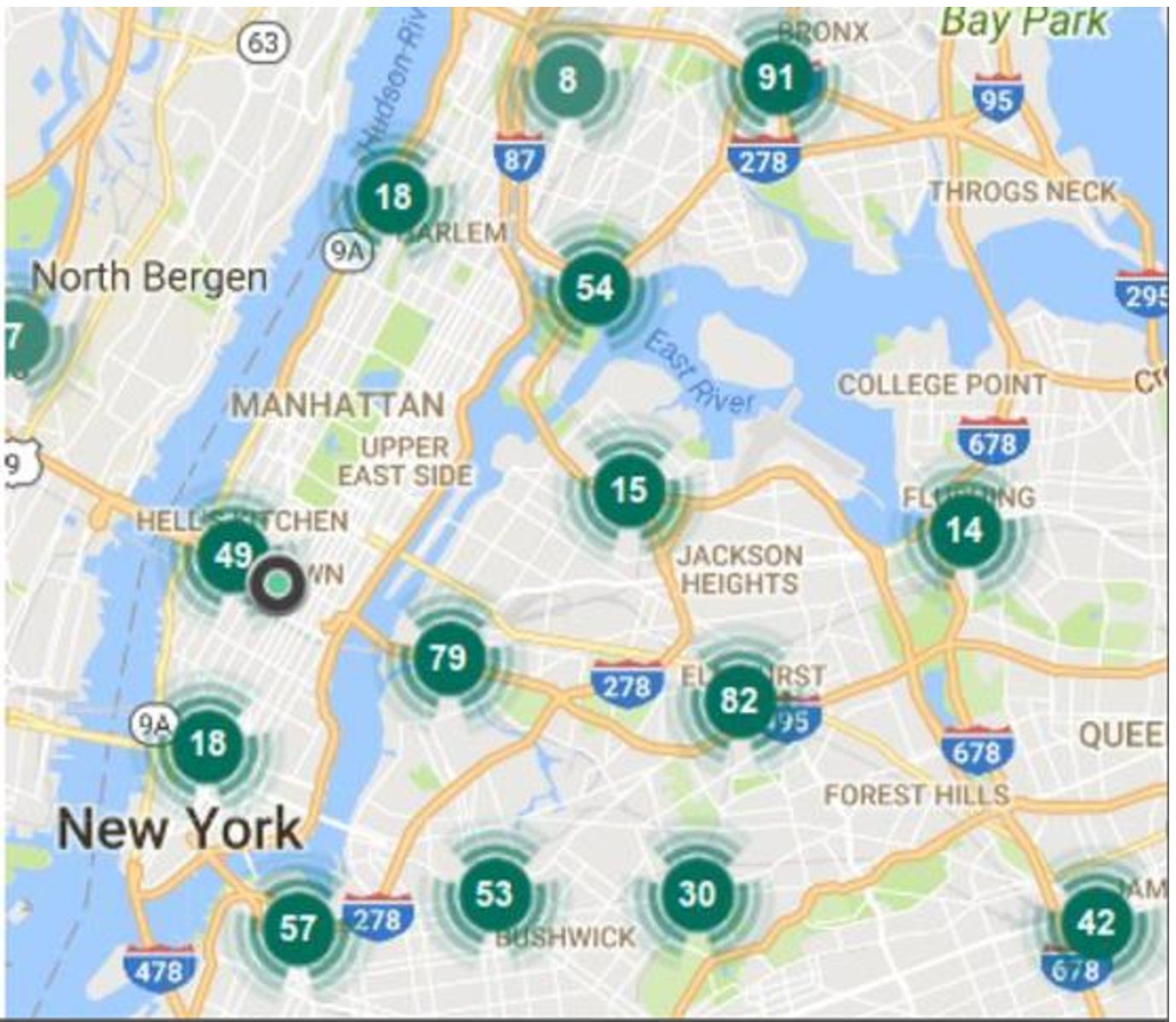}\label{fig:lamar}}
	\vspace{-12pt}
	\caption{Distribution of Datasets in NYC}
\end{figure}

\begin{figure*}[!t]
	\centering
	\vspace{-.75em}
	\hspace{-1.2em}
	\subfloat[Influence (NYC)]{\includegraphics[clip,width=0.195\textwidth]{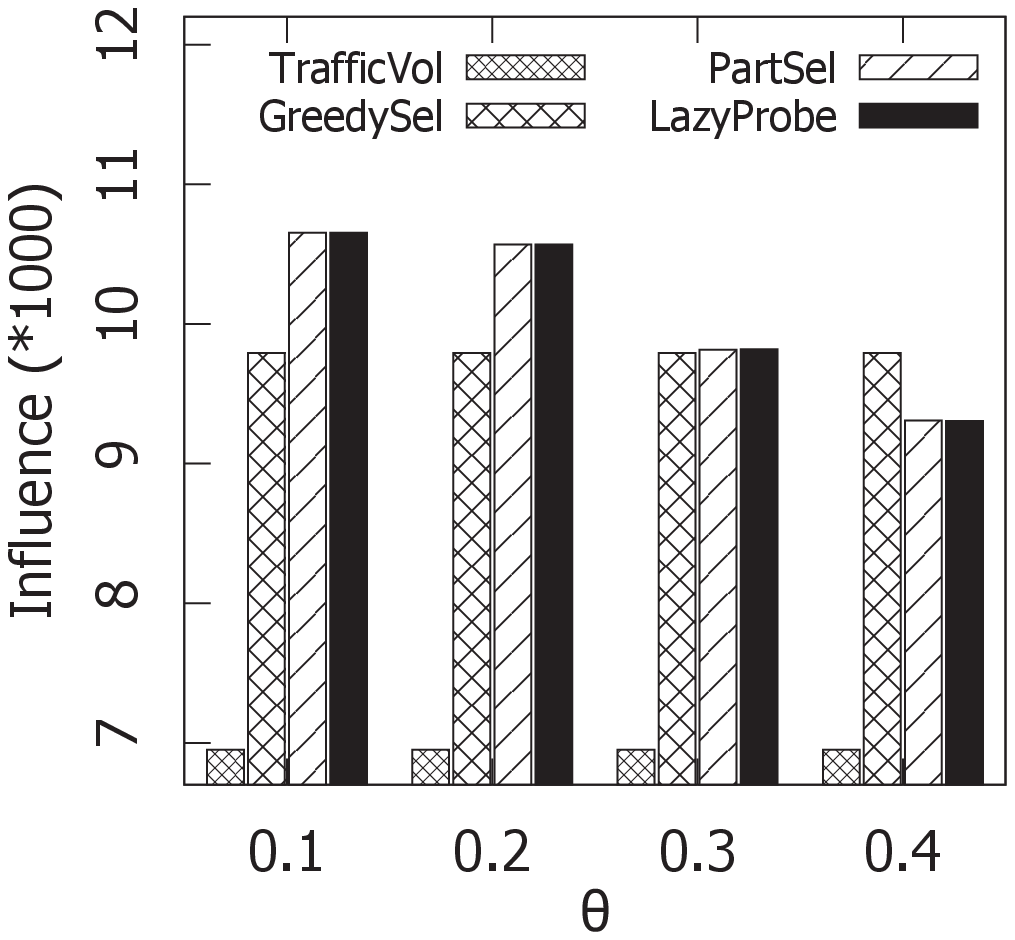}\label{fig:exp-theta-ifl-NYC}}
	\subfloat[Efficiency (NYC)]{\includegraphics[clip,width=0.205\textwidth]{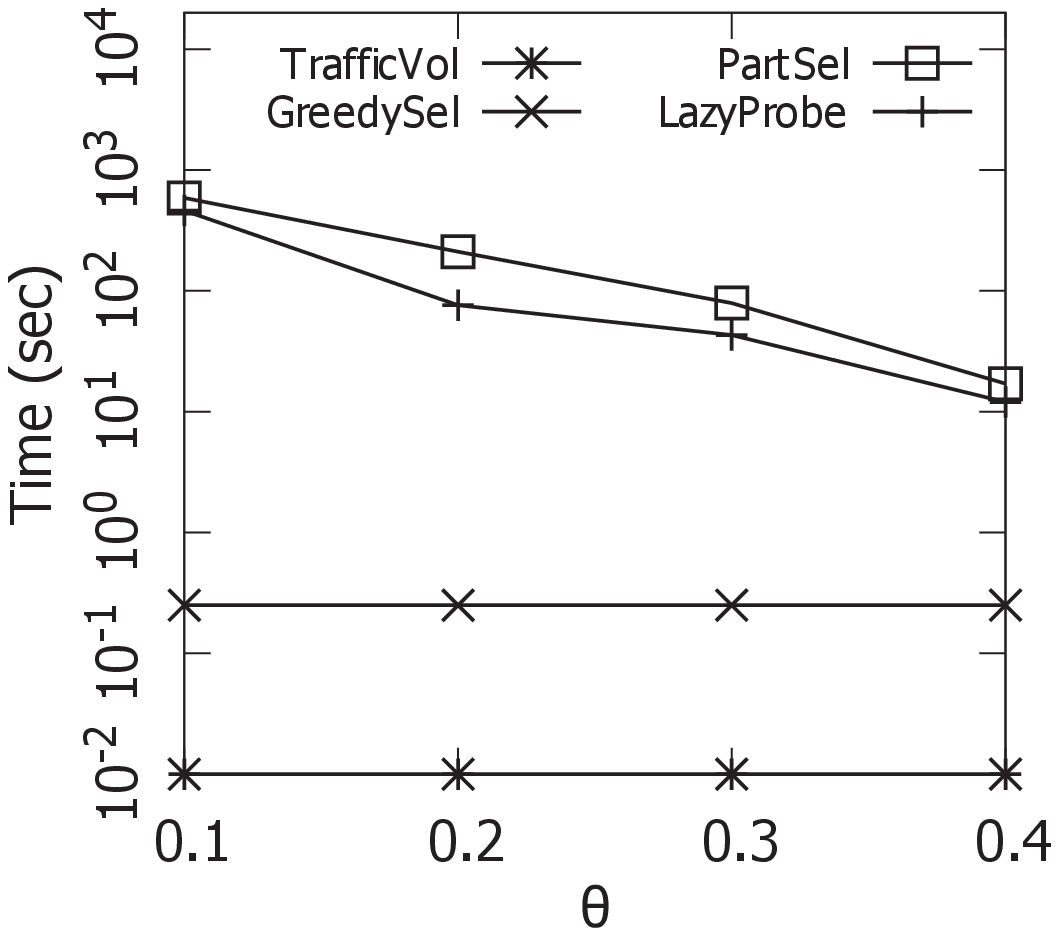}\label{fig:exp-theta-eff-NYC}}
	\hspace{-.6em}
	\subfloat[Influence (LA)]{\includegraphics[clip,width=0.195\textwidth]{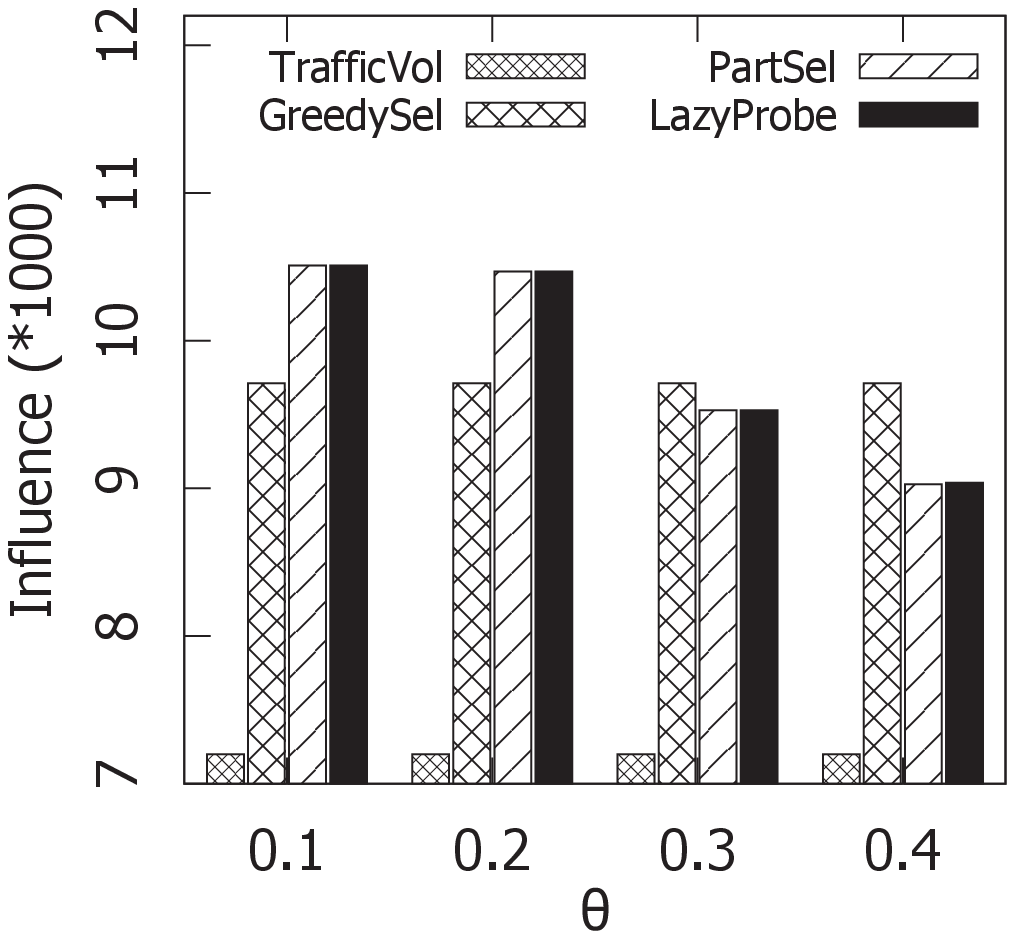}\label{fig:exp-theta-ifl-LA}}
	\subfloat[Efficiency (LA)]{\includegraphics[clip,width=0.205\textwidth]{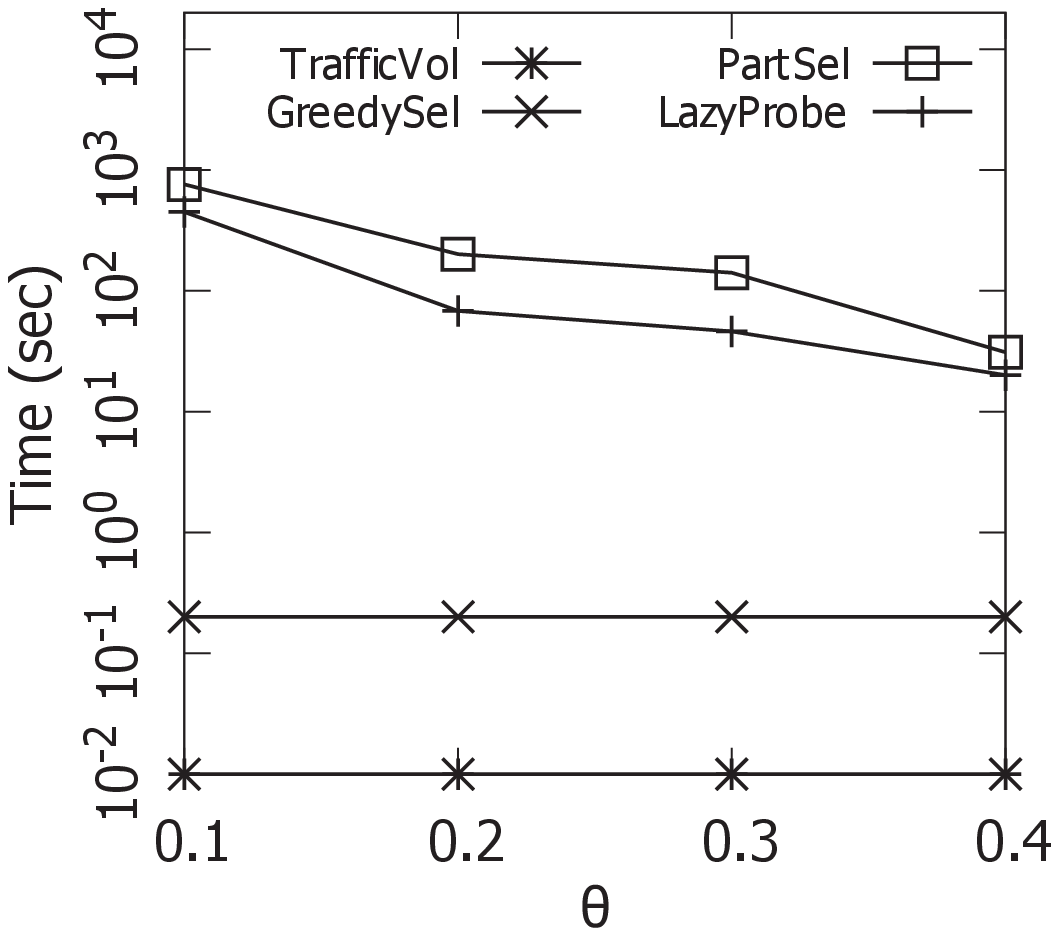}\label{fig:exp-theta-eff-LA}}
	\hspace{-.6em}
	\subfloat[Number of clusters]{\includegraphics[clip,width=0.195\textwidth]{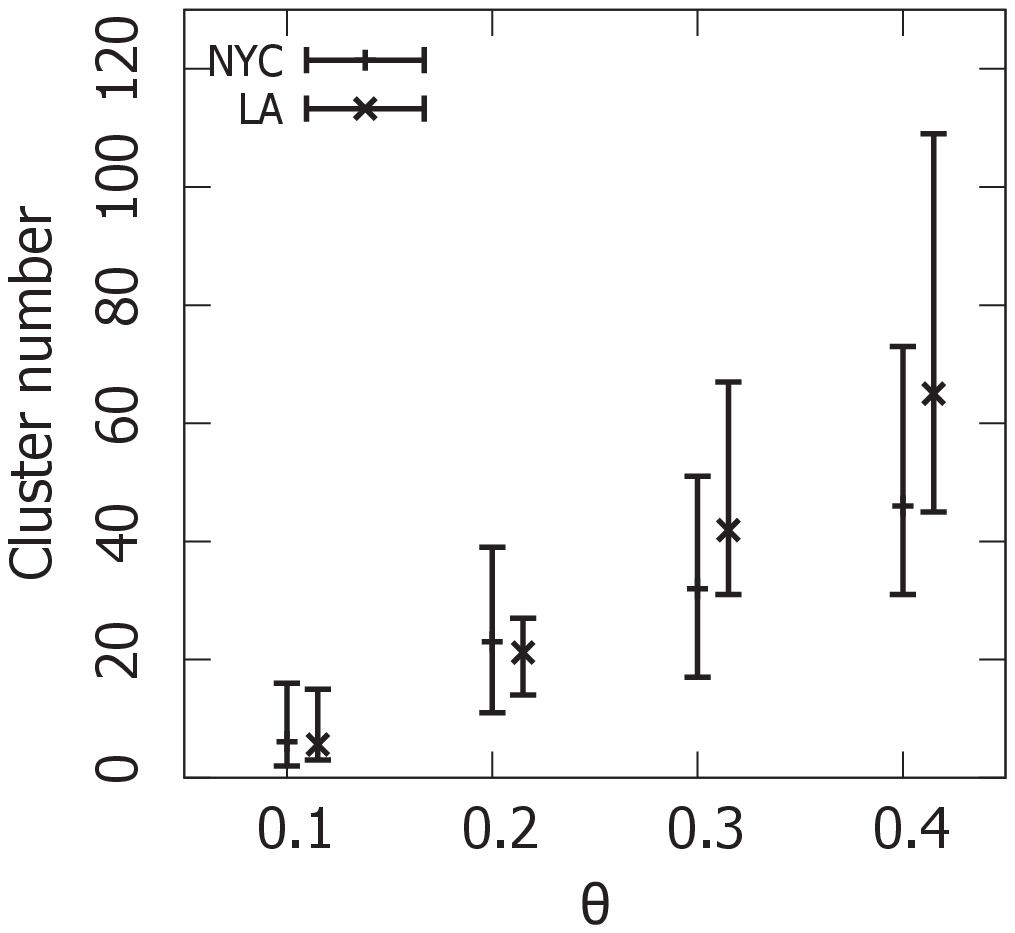}\label{fig:theta_cluster}}
	\vspace{-12pt}
	\caption{Effect of varying $\theta$ on NYC and LA}
	\label{fig:varying_theta}
\end{figure*}

\vspace{-.2em}
\subsection{Experiments}\label{exp-billsel}
\subsubsection{Choice of $\theta$-partition}

\begin{table}[!t] \vspace{1.5em}
	\centering
	\caption{The $|C_m|/|\ur|$ ratio w.r.t. varying $\theta$}
	\label{max_cluster}
	\vspace{-.3em}
	\begin{tabular}{|c|c|c|c|c|}
		\hline
		& 0.1                       & 0.2                       & 0.3                       & 0.4                       \\ \hline
		\multicolumn{1}{|l|}{NYC}  & \multicolumn{1}{l|}{12.6\%} & \multicolumn{1}{l|}{7.1\%} & \multicolumn{1}{l|}{6.4\%} & \multicolumn{1}{l|}{5.8\%} \\ \hline
		
		\multicolumn{1}{|l|}{LA} & \multicolumn{1}{l|}{13.5\%} & \multicolumn{1}{l|}{7.8\%} & \multicolumn{1}{l|}{5.9\%} & \multicolumn{1}{l|}{5.1\%} \\ \hline
	\end{tabular}
\end{table}

\begin{figure*}[!t]
	\vspace{-1em}
	\centering
	\subfloat[Influence (NYC)]{\includegraphics[clip,width=0.245\textwidth]{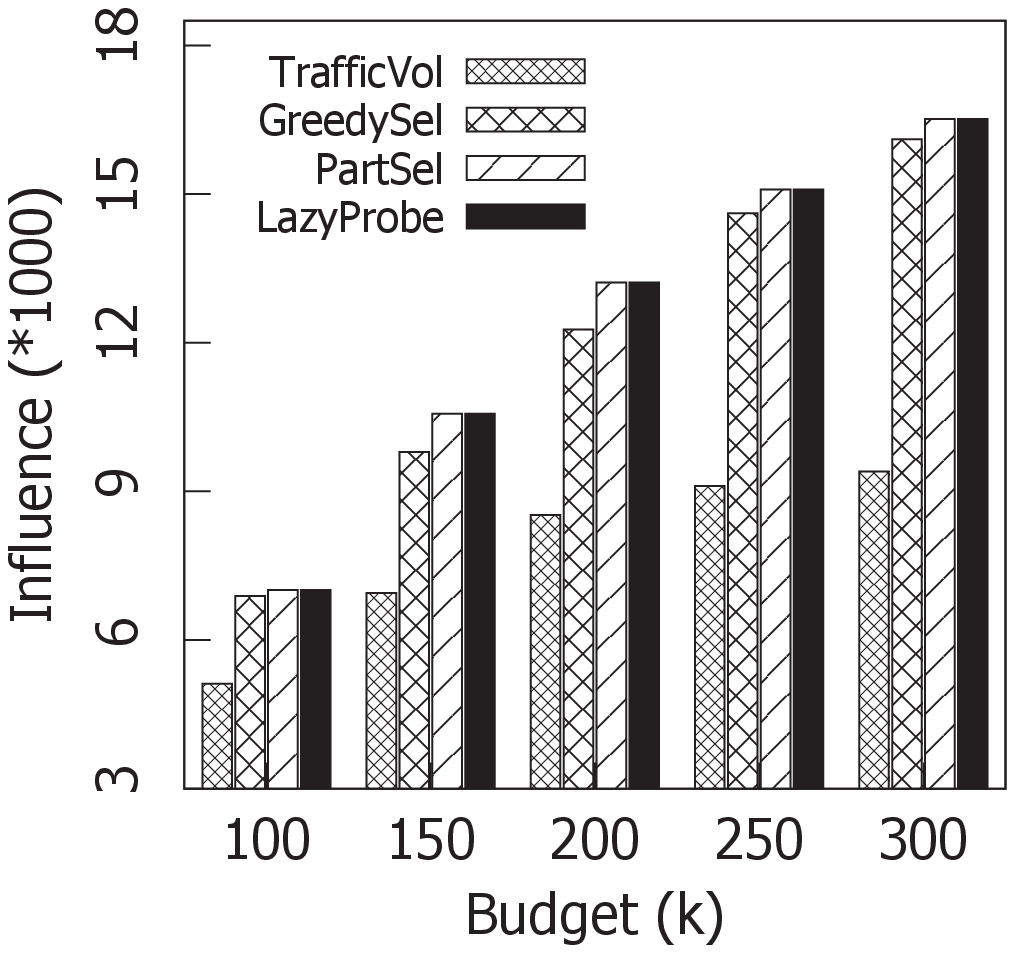}\label{fig:exp-L-ifl-NYC}}
	\subfloat[Efficiency (NYC)]{\includegraphics[clip,width=0.25\textwidth]{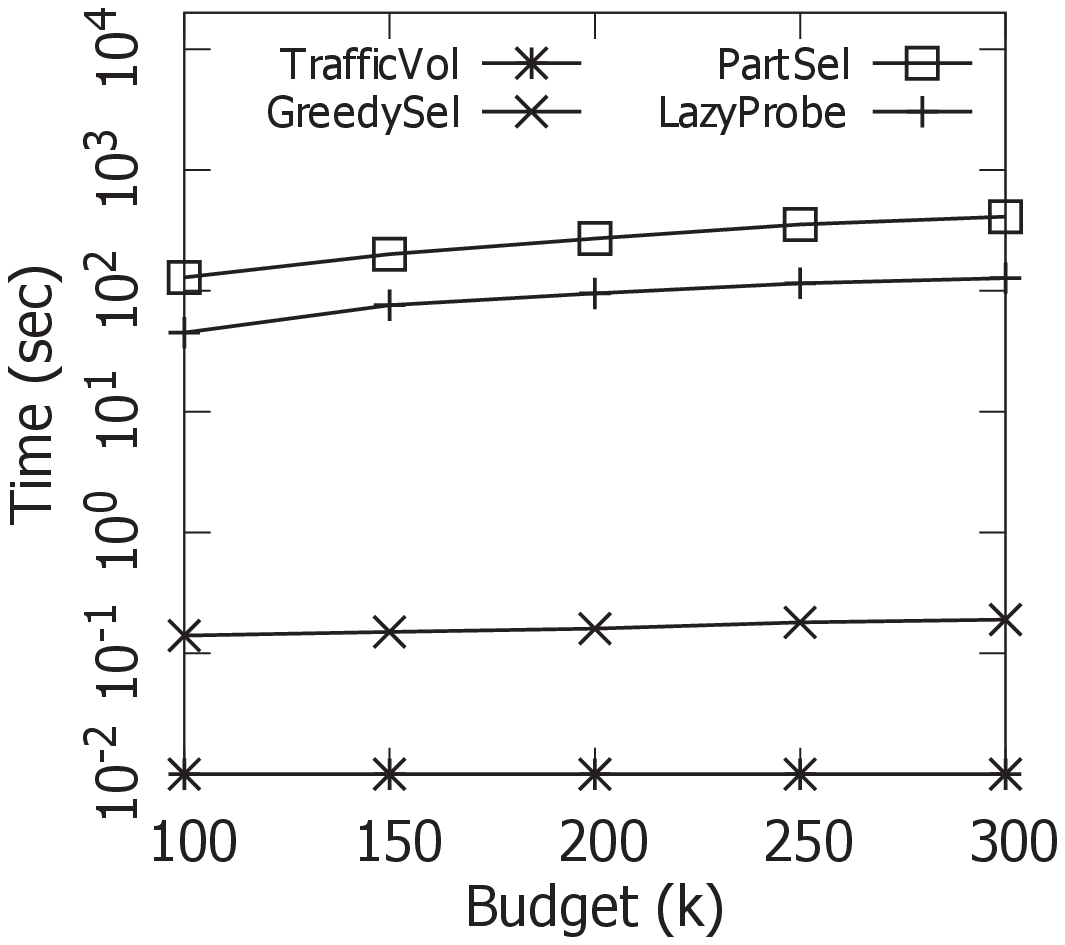}\label{fig:exp-L-eff-NYC}}
	\subfloat[Influence (LA)]{\includegraphics[clip,width=0.245\textwidth]{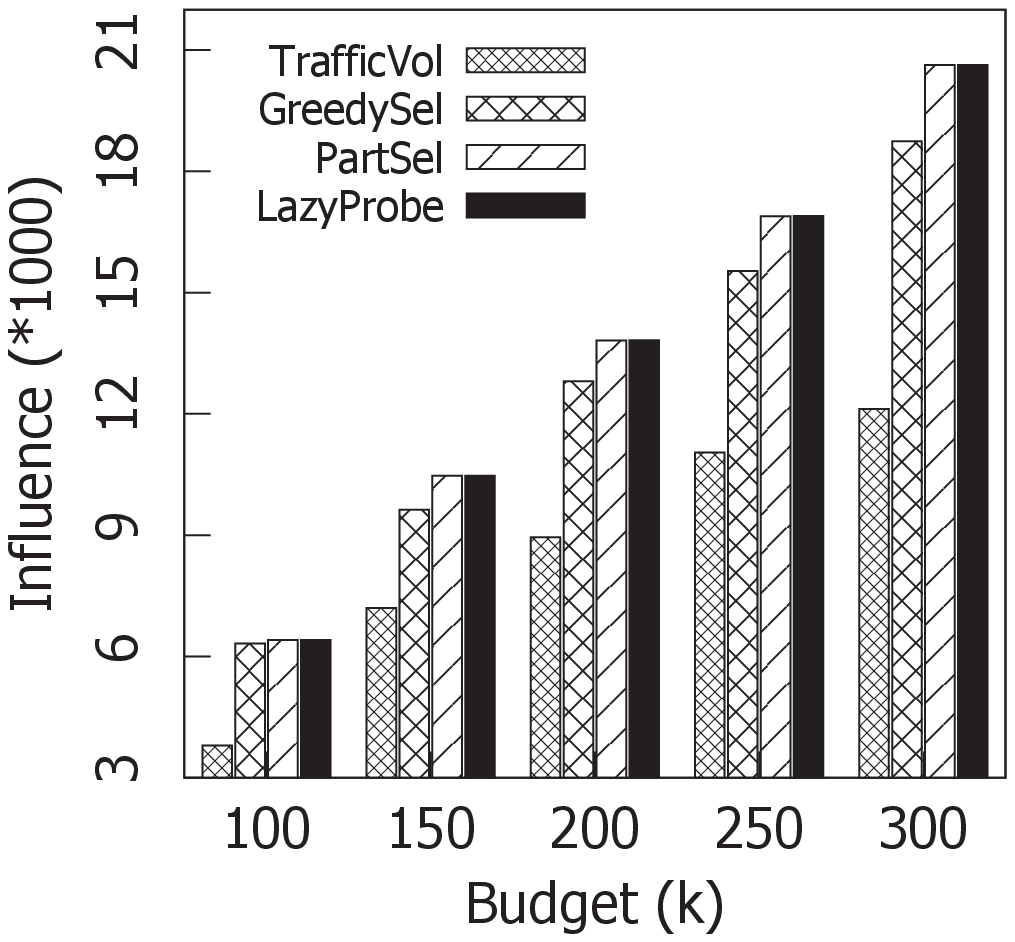}\label{fig:exp-L-ifl-LA}}
	\subfloat[Efficiency (LA)]{\includegraphics[clip,width=0.25\textwidth]{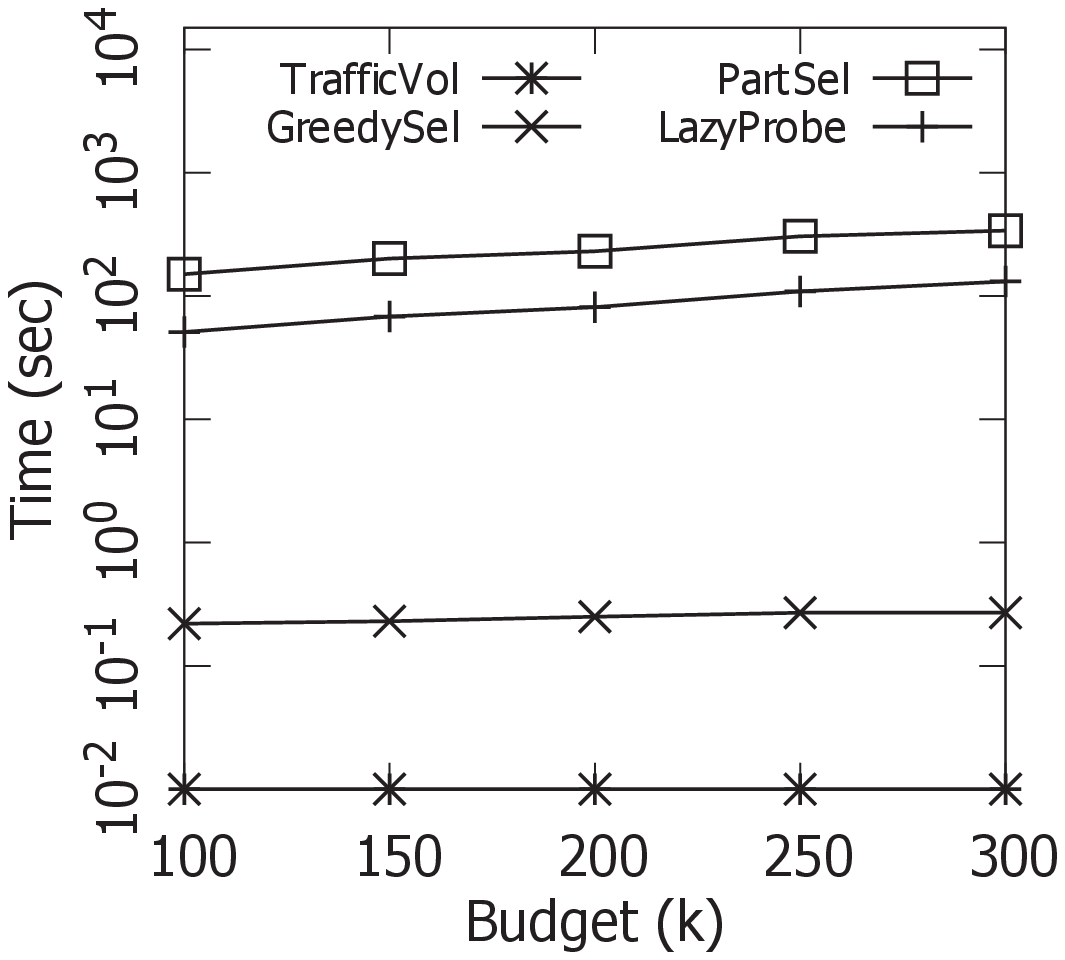}\label{fig:exp-L-eff-LA}}
	\vspace{-12pt}
	\caption{Effect of varying budget $\budget$}
	\label{fig:varying_l_nyc}
	
\end{figure*}
\begin{figure*}[!t]	\vspace{-1em}
	\centering
	\subfloat[Influence (NYC)]{\includegraphics[clip,width=0.245\textwidth]{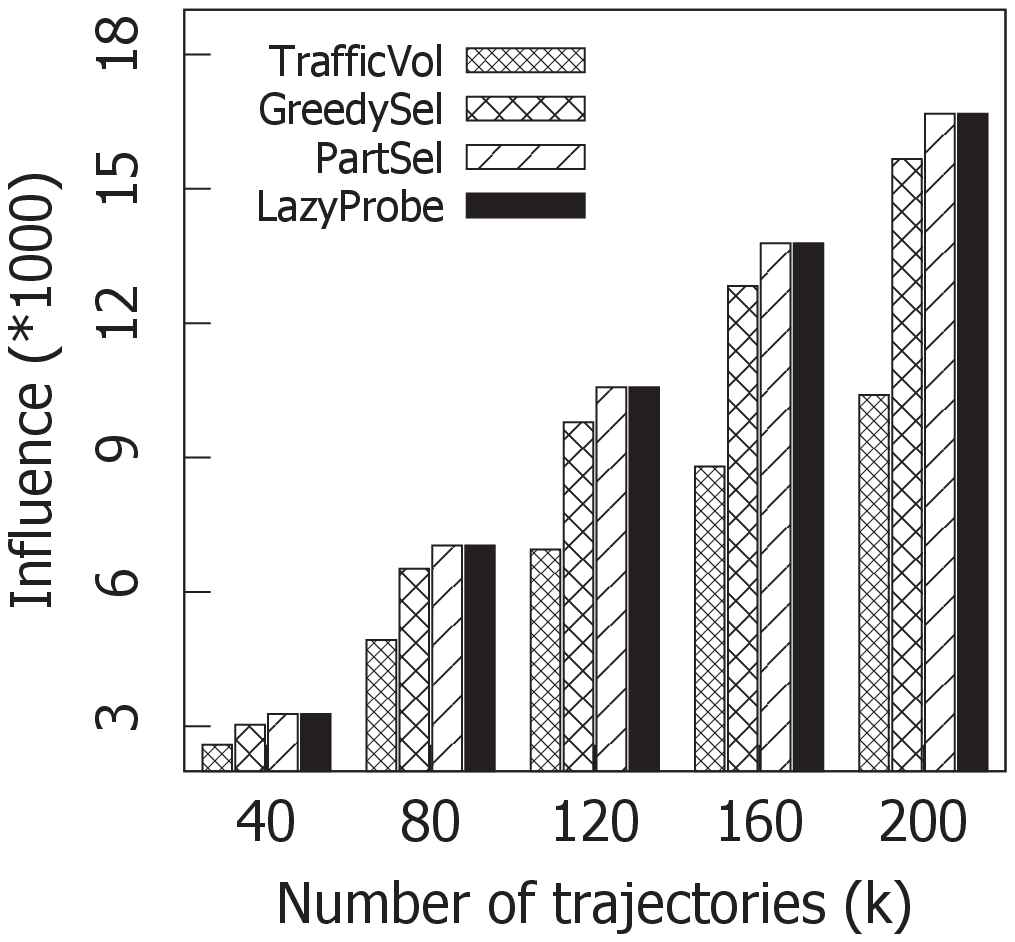}\label{fig:exp-Tr-ifl-NYC}}
	\subfloat[Efficiency (NYC)]{\includegraphics[clip,width=0.25\textwidth]{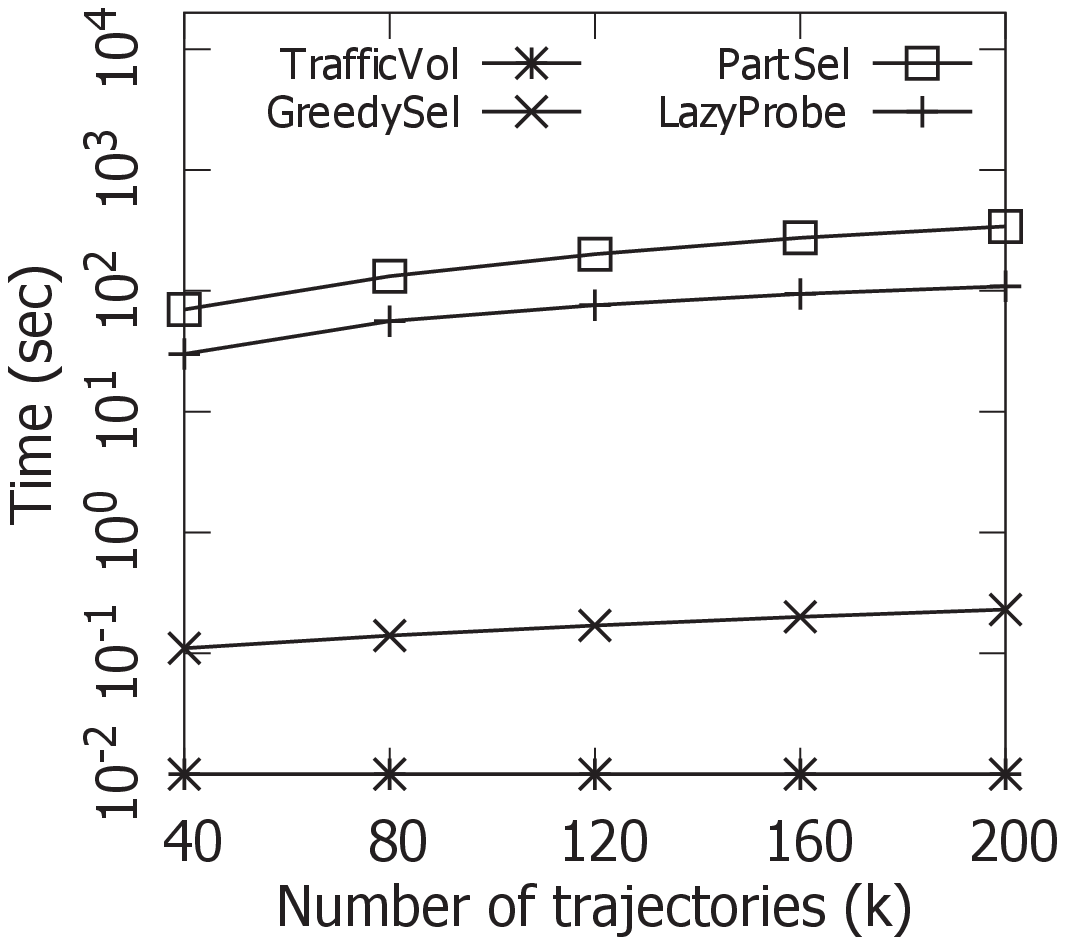}\label{fig:exp-Tr-eff-NYC}}
	\subfloat[Influence (LA)]{\includegraphics[clip,width=0.245\textwidth]{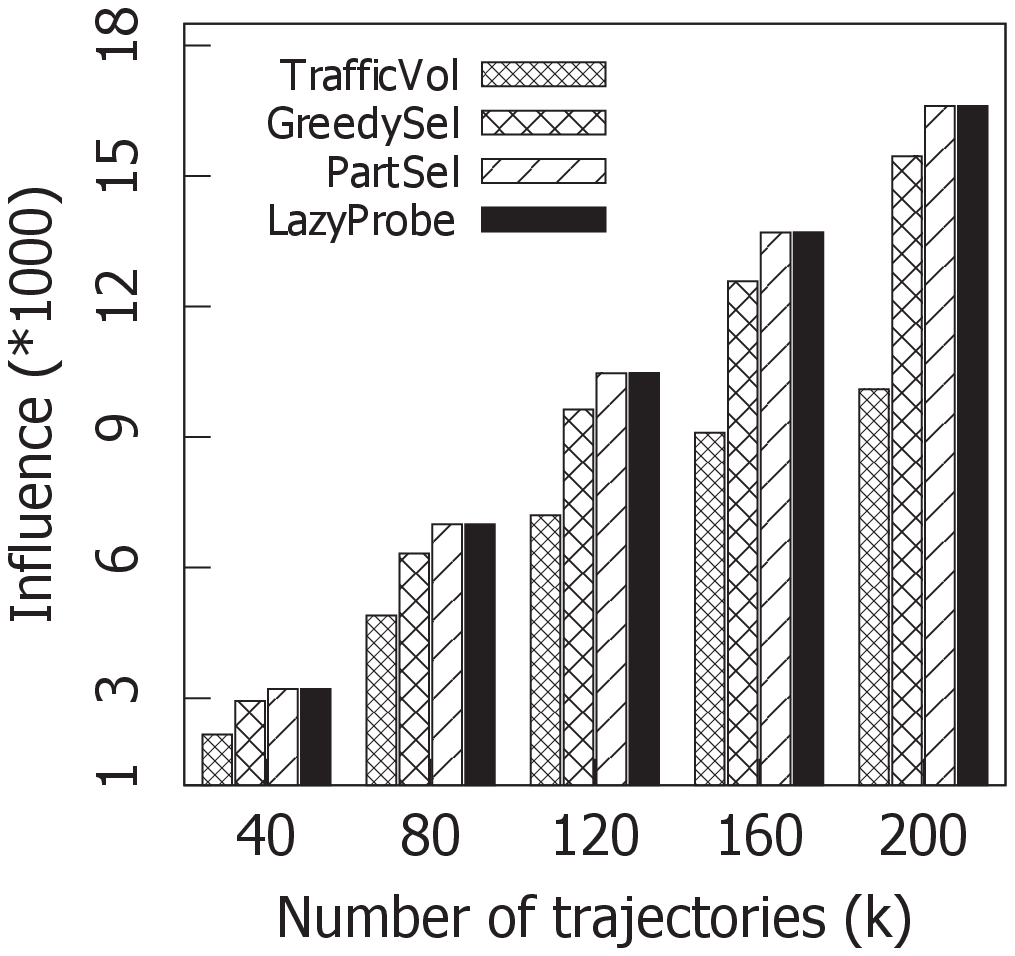}\label{fig:exp-Tr-ifl-LA}}
	\subfloat[Efficiency (LA)]{\includegraphics[clip,width=0.25\textwidth]{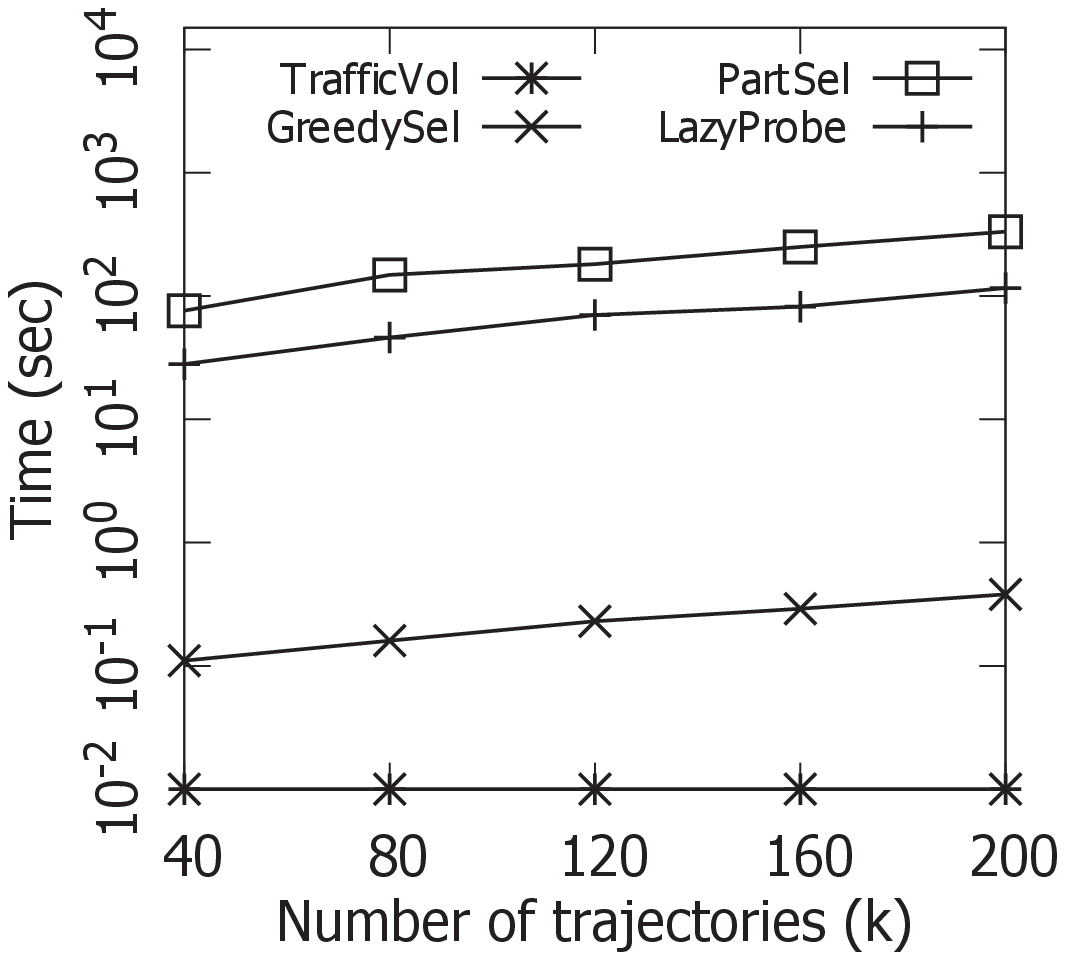}\label{fig:exp-Tr-eff-LA}}
	\vspace{-12pt}
	\caption{Effect of varying trajectory number $|\td|$}
	\label{fig:varying_t_nyc}
		
\end{figure*}

Since $\theta$-partition is an input of \psel method and $\theta$ indicates the degree of overlap among clusters generated in the partition phase of \psel (and \bbsel), we would like to find a generally good choice of $\theta$ that strikes a balance between the efficiency and effectiveness of \psel and \bbsel.

We vary $\theta$ from 0 to 0.4, and record the number of clusters as input of \psel and \bbsel methods, the percentage of the largest cluster size over $\ur$ (i.e., $\frac{|C_m|}{|\ur|}$), the runtime and the influence value of \psel and \bbsel. The results on both datasets are shown in Figure \ref{fig:varying_theta} and Table \ref{max_cluster}. Note that \enumgreedy is too slow, we do not include it here.
By linking those results, we have four main observations.
(1) With the increase of $\theta$, the influence quality decreases and the efficiency is improved, because for a larger $\theta$, the tolerated influence overlap is larger and  there are many more clusters with larger overlaps. 
(2) When $\theta$ is 0.1 and 0.2, \psel and \bbsel achieve the best  influence (Figures \ref{fig:exp-theta-ifl-NYC} and \ref{fig:exp-theta-ifl-LA}), while the efficiency of 0.2 is not much worse than that of $\theta$=0.3 (Figures \ref{fig:exp-theta-eff-NYC} and \ref{fig:exp-theta-eff-LA}). The reason is that, it results in an appropriate number of clusters (e.g., 23 clusters for NYC dataset at $\theta$=0.2 in Figure \ref{fig:theta_cluster}), and the largest cluster covers 7.1\% of all billboards, as evidenced by the value of $\frac{|C_m|}{|\ur|}$ in Table \ref{max_cluster}. 
(3) In one extreme case that $\theta$=0.4, although the generated clusters are dispersed and small, it results in high overlaps among clusters, so the influence value drops and becomes worse than \gre, and meanwhile the efficiency of \psel (\bbsel) only improves by around 12 (6) times as compared to that of $\theta$=0.2 on the NYC (LA) dataset. The reason is that \psel and \bbsel find influential billboards within a cluster and do not consider the influence overlap to billboards in other cluster; thus when $\theta$ is large, high overlaps between the clusters lead to a low precision of \psel and \bbsel. 
(4) All other methods beat the \topk baseline by 45\% in term of the influence value of selected billboards. 
%

The result on LA is very similar to that of NYC, so we omit the description here. Therefore, we choose 0.2 as the default value of $\theta$ in the rest of the experiments.

\vspace{-.5em}
\subsection{Effectiveness Study}\label{exp-billsel}
\vspace{-.25em}

We study how the influence is affected by varying the budget $\budget$, trajectory number $|\td|$, distance threshold $\lambda$ and overlap ratio respectively. Last we study the approximation ratio of all algorithms.

\subsubsection{Varying the budget \emph\budget}
The influence of all algorithms on NYC and LA by varying the $\budget$ is shown in Figure \ref{fig:varying_l_nyc}, and we have the following observations on both datasets. (1) \topk has the worst performance. \psel and \bbsel achieve the same influence. The improvement of \psel and \bbsel over \topk exceeds 99\%.  (2) With the growth of $L$, the advantage of \psel and \bbsel over \gre are increasing, from 1.8\% to 6.5\% when $\budget$ varies from 100k to 300k on LA dataset. This is because when the influence overlaps between clusters cannot be avoided, then how to maximize the benefit/cost ratio in clusters is critical to enhance the performance, which is exactly achieved by \psel and \bbsel, since they exploit the locality feature within clusters.  
   
\subsubsection{Varying the trajectory number $|\td|$}
Figure \ref{fig:varying_t_nyc} shows the result by varying $|\td|$. We find: (1) the influence of all methods increase because more trajectories can be influenced; (2) the influence by \psel and \bbsel is consistently better than that of \gre and \topk, because the trajectory locality is an important factor that should be considered to increase the influence. 


\begin{figure}[!t]
	\centering
	\subfloat[NYC]{\includegraphics[clip,width=0.24\textwidth]{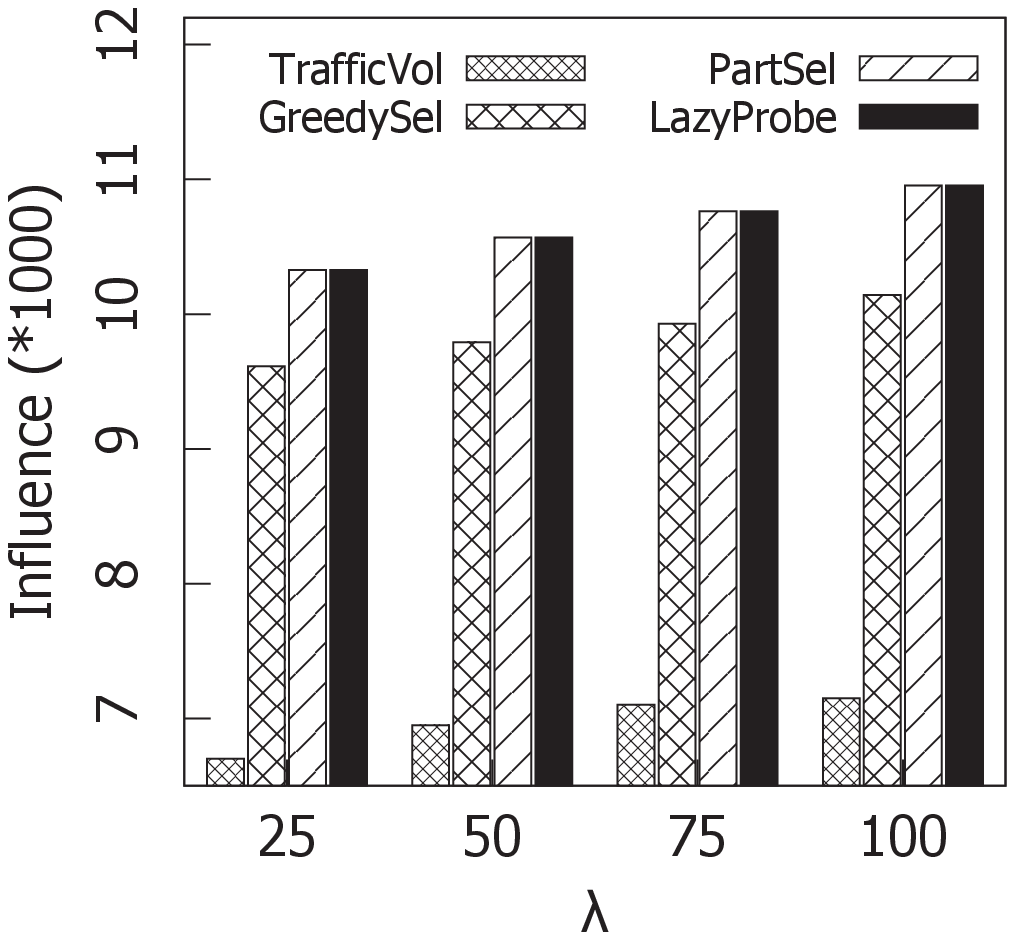}\label{fig:exp-lambda-ifl-NYC}}
	\subfloat[LA]{\includegraphics[clip,width=0.24\textwidth]{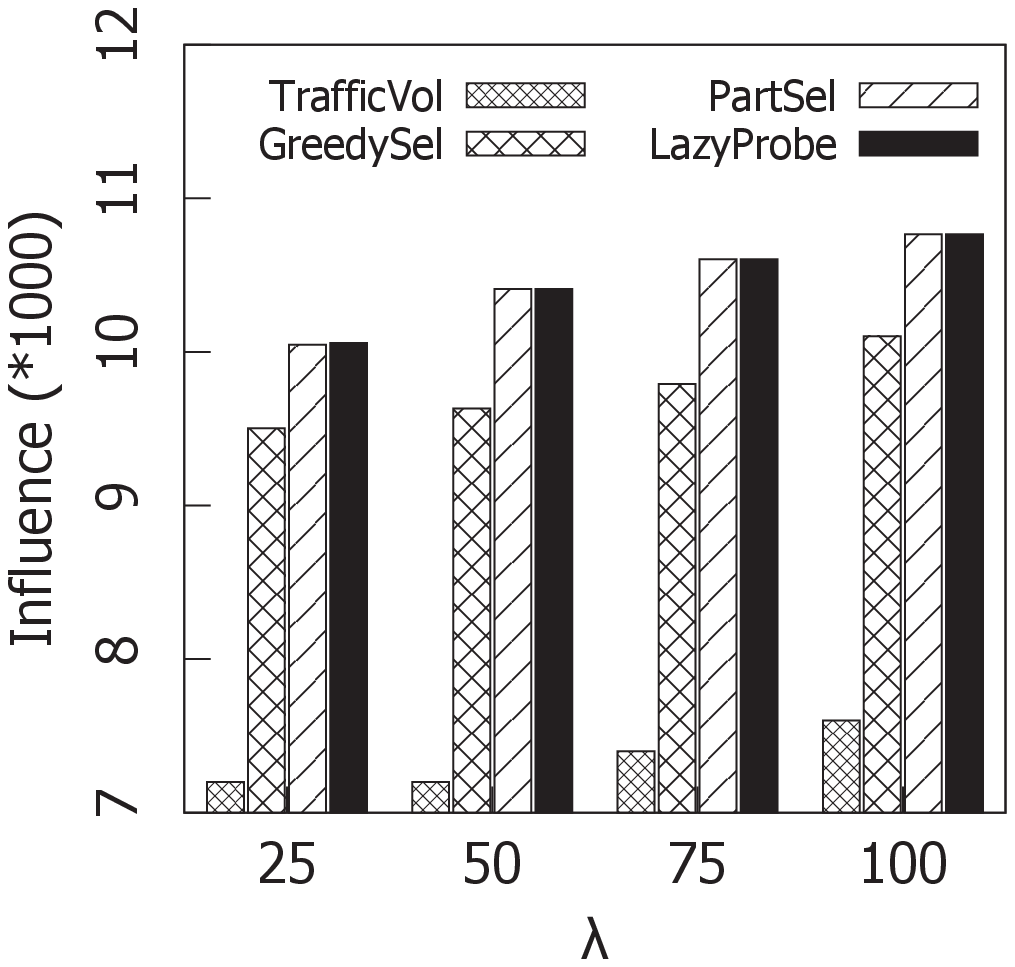}\label{fig:exp-lambda-ifl-LA}}
	\vspace{-12pt}
	\caption{Effect of varying $\lambda$\label{fig:vary:lambda}}
\end{figure}

\subsubsection{Varying $\lambda$}
Figure~\ref{fig:vary:lambda} shows the influence result by varying the threshold $\lambda$, which determines the influence relationship between billboards and trajectories (in Definition~\ref{def:meet}). We make two observations. (1) With the increase of $\lambda$, the performance of all algorithms becomes better, because a single billboard can influence more trajectories. (2) \psel and \bbsel have the best performance and outperform the \gre baseline by at least 8\%. This is because the enumerations can easily find influential billboards when the influence overlap becomes larger.

\begin{table}[!t]	\vspace{1.5em}
	\centering
	\caption{Additional test on NYC}
	\vspace{-1em}
	\label{fig:exp-approxi-NYC}
	\begin{small}
		\begin{tabular}{|c|c|c|c|c|c|c|}
			\hline
			& \multicolumn{2}{c|}{\budget=100k} & \multicolumn{2}{c|}{\budget=200k} & \multicolumn{2}{c|}{\budget=300k} \\ \hline
			\ann        & 6805 & 0.00\% & 11777 & 0.00\% & 15773 & 0.00\% \\ \hline
			\topk      & 5111 & -24.89\% & 8520& -27.66\% & 9400& -40.40\% \\ \hline	
			\gre        & 6890 & 1.25\% & 12267 & 5.56\% & 16108 & 2.12\% \\ \hline
			\enumgreedy & 7080 & 4.04\% & 13161 & 11.75\% & 16570 & 5.05\% \\ \hline
			\psel       & 7013 & 3.06\% & 13215 & 12.21\% & 16512 & 4.69\% \\ \hline
			\bbsel     & 7013 & 3.06\% & 13215 & 12.21\% & 16512 & 4.69\% \\ \hline
		\end{tabular}
	\end{small}
\end{table}

\subsubsection{Additional Discussion}

We also compared our solution with a meta heuristic algorithm, Simulated Annealing (Annealing), to verify the practical effectiveness. Although \ann is costly and provides no theoretical bound for our problem, it has been proved to be a very powerful way for most optimization problems and always can find a near optimal solution~\cite{talbi2009metaheuristics}. 
Since \ann is a random search algorithm and its performance is not stable, we run it 50 times for each instance and select the best solution as our baseline. Table~\ref{fig:exp-approxi-NYC} reports both the influence value and its relative improvement percentage w.r.t. \ann for three different choices of budget \budget.
We observe: (1) \psel and \bbsel have a very close performance to \enumgreedy in average. This is because when the overlap between clusters are small, each billboard selected by of \psel and \bbsel is less likely to overlap with the billboards in other clusters, and thus the performance of \psel and \bbsel would note lose a large accuracy. As discussed later in Section~\ref{sec:efficiency}, \enumgreedy is very slow to work in practice. (2) \psel and \bbsel improve the influence by 6.6\% in average as compared to \ann. (3) \topk which simply uses the traffic volume to select billboards has the worst approximation.

\subsection{Efficiency Study}\label{sec:efficiency}
\subsubsection{Varying the budget \emph\budget}
Figure \ref{fig:exp-L-eff-NYC} and Figure \ref{fig:exp-L-eff-LA} present the efficiency result when budget $L$ varies from 100k to 300k on NYC and LA. As \enumgreedy is too slow to converge in $10^4$ seconds (because the complexity of \enumgreedy is proportional to $|\ur|^5$), we omit it in the figures.
We have three main observations. (1) \bbsel consistently beats \psel by almost 3 times. (2) The gap between \gre and \psel becomes more significant w.r.t. the increase of $\budget$. It is because \psel has to invoke \enumgreedy $\budget$ times for each cluster to construct the local solution for each cluster, so the runtime of \psel grows quickly with $\budget$ being increased. (3) \topk is the fastest one with no surprise, because it simply adopts a benefit-based selection.

\subsubsection{Varying the trajectory number $\td$}
Figure \ref{fig:exp-Tr-eff-NYC} and Figure \ref{fig:exp-Tr-eff-LA} show the runtime of all algorithms on NYC and LA datasets.
We observe that \psel and \bbsel scale linearly w.r.t. $\td$ which is consistent with our time complexity analysis; moreover, \bbsel is around 4 times faster than  \psel.

\begin{figure}[!t]
	\centering
	\subfloat[Varying $|\td|$]{\includegraphics[clip,width=0.25\textwidth]{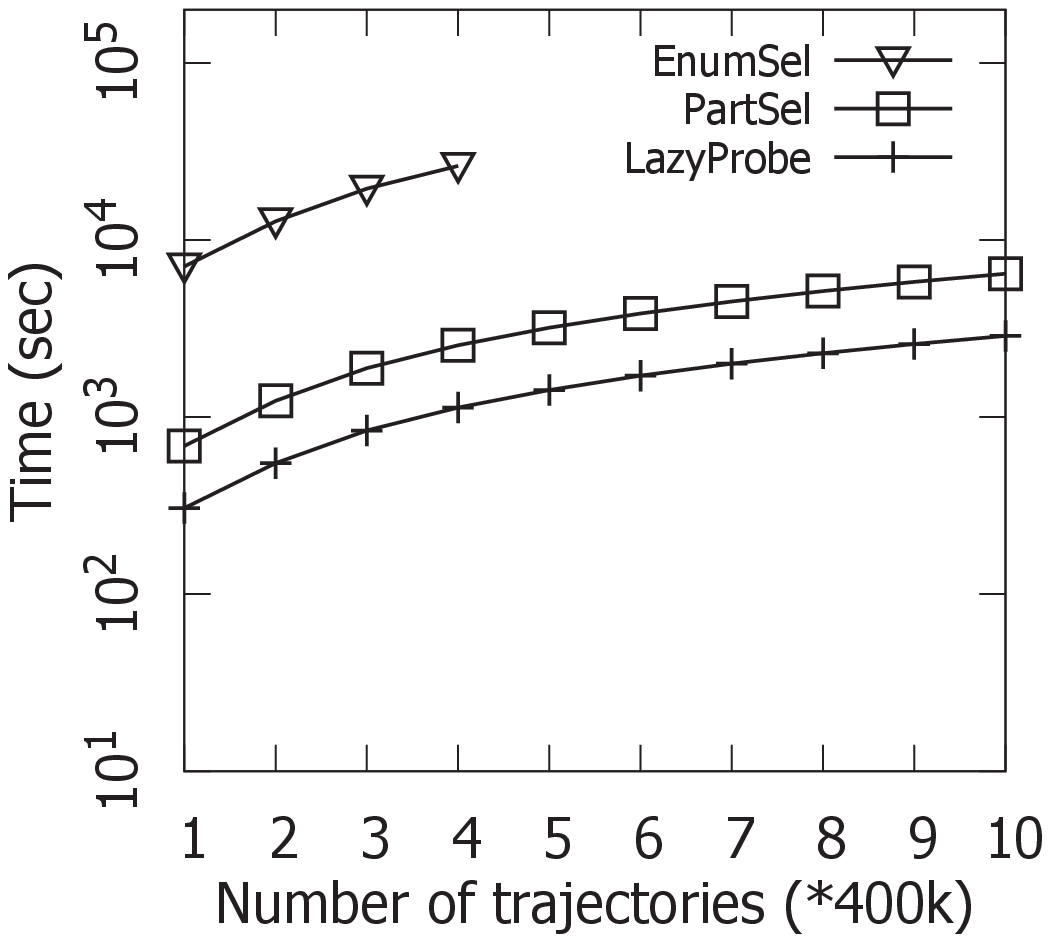}\label{fig:scale_tnum}}
	\subfloat[Varying $|\ur|$]{\includegraphics[clip,width=0.25\textwidth]{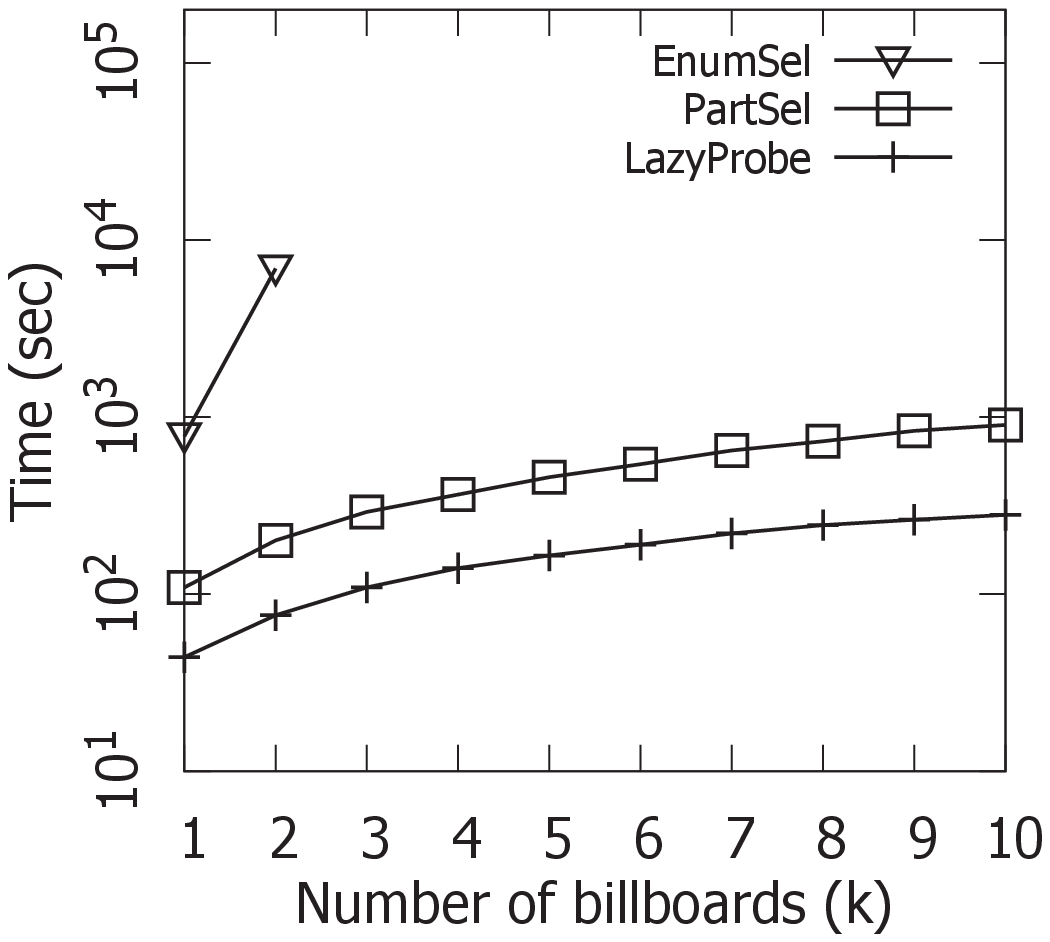}\label{fig:exp-bil-eff-NYC}}
	\vspace{-9pt}
	\caption{Scalability test of our methods on NYC dataset}
	\label{fig:trajectory_scale}
\end{figure}

\begin{figure}[!t]
	\centering
	\subfloat[Influence (Varying $\budget$)]{\includegraphics[clip,width=0.25\textwidth]{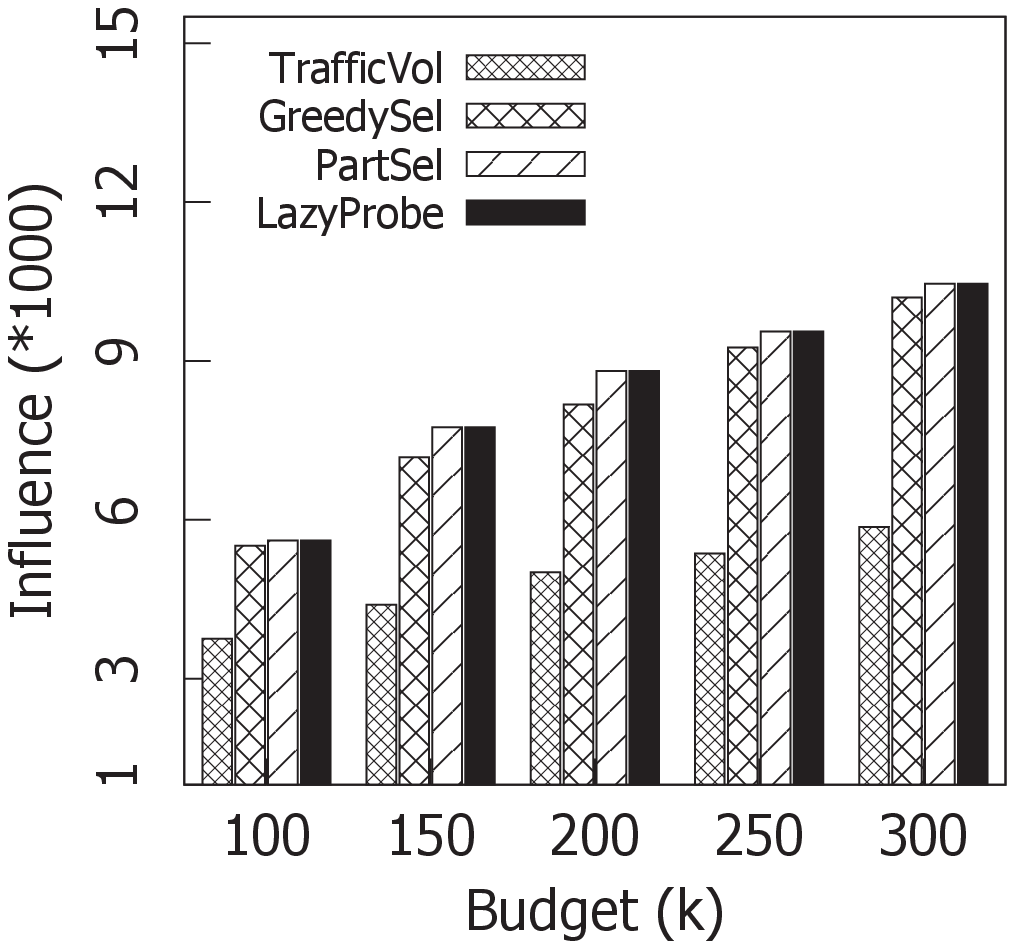}\label{fig:NYC_p_budget_inf}}
	\subfloat[Influence (Varying $|\td|$)]{\includegraphics[clip,width=0.25\textwidth]{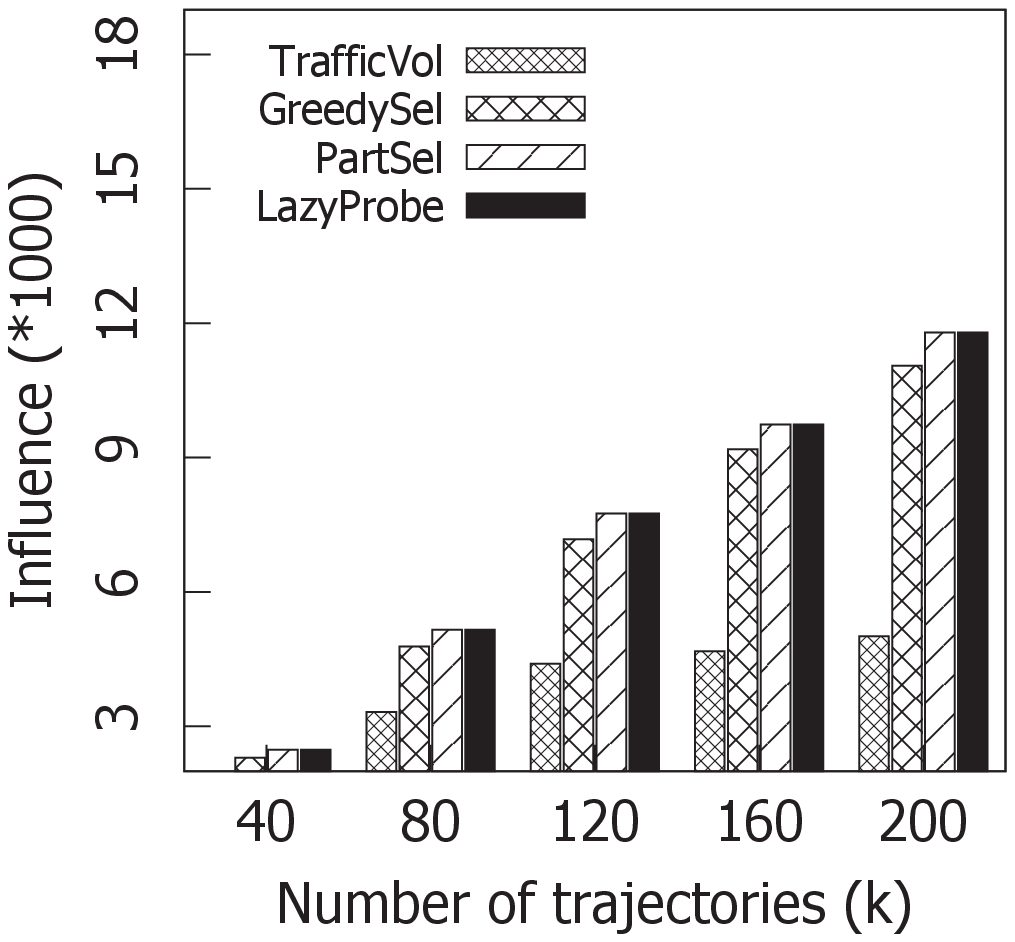}\label{fig:NYC_p_trajectory_inf}}
	\vspace{-9pt}
	\caption{Test on other choice of influence probability}
	\label{fig:vary_p}
\end{figure}

\subsection{Scalability Study}

In this experiment we evaluate the scalability of our methods, \enumgreedy, \psel and \bbsel, by varying $|\td|$ (from 400k to 4m) and $|\ur|$ (from 1k to 10k). Since the effectiveness of \gre is not satisfying (as evidenced in our effectiveness study), we do not compare the efficiency of \gre. The results are shown in Figure~\ref{fig:scale_tnum} and Figure~\ref{fig:exp-bil-eff-NYC}. We can see that \bbsel scales very well and outperforms \psel by 4-6 times. This is because even if the number of billboards is large, \bbsel does not need to compute all local solutions for each cluster with different budgets, while it still can prune a large number of insignificant computations. Since \enumgreedy takes more than 10,000 seconds when the billboard number $|U|$ is larger than 2k, its result is omitted in the Figure for readability reason. It also shows that \enumgreedy has serious issues in efficiency making it impractical in real-world scenarios, while \psel and \bbsel scale well and can meet the efficiency requirement.

\vspace{.2em}
\noindent {\bf Summary.} (1) Our methods \enumgreedy, \psel and \bbsel achieve much higher influence value than existing techniques (\gre, \topk, and \ann).  (2)  \psel and \bbsel achieve similar influence with \enumgreedy, but \enumgreedy is too slow and not acceptable in practice while \bbsel and \psel are much faster than \enumgreedy and can meet the efficiency requirement on large datasets.

\subsection{Complementary study} \label{exp:comley}
As reported in Section~\ref{sec:exp}, $\enumgreedy$ could not terminate in a reasonable time for most experiments' default settings due to its dramatically high computation cost $O(|\td| \cdot {|\ur|}^5)$, we generate a small subset of the NYC dataset to ensure that it can complete in reasonable time, and compare its performance with other approaches proposed in this paper. In particular, we have the default setting of $|\ur|$=1000 and $|\td|$=120k.

Figure~\ref{fig:exp-s-L-ifl-NYC} and Figure~\ref{fig:exp-s-L-ifl-LA} show the effectiveness of all algorithms when varying the budget $\budget$ and the number of trajectories respectively. From  Figure~\ref{fig:exp-s-L-ifl-NYC} we make two observations:
(1) When $\budget$ is small, the influence of \enumgreedy is better than that of \psel and \bbsel. It is because when only a small number of billboards can be afforded, the enumerations can easily find the optimal set since the possible world of feasible sets is small, whereas \psel and \bbsel are mainly obstructed by reduplicating the influence overlaps between clusters. (2) With the growth of $L$, the advantage of \enumgreedy gradually drops, while \psel and \bbsel achieve better influence; and when the budget reaches 200k, they have almost the same influence as \enumgreedy. Similar observations are made in Figure~\ref{fig:exp-s-L-ifl-LA}. 

The efficiency results are presented in Figure~\ref{fig:exp-s-L-eff-NYC} and Figure~\ref{fig:exp-s-L-eff-LA} w.r.t. a varying budget and trajectory number. The efficiency of the \topk baseline is not recorded because it is too trivial to get any approximately optimal solution. We find: (1) \bbsel and \psel consistently beat \enumgreedy by almost two and one order of magnitude respectively. (2) \enumgreedy has the worst performance among all algorithms.

\begin{figure*}[!tb]
	\centering
	\subfloat[Influence (Varying $\budget$)]{\includegraphics[clip,width=0.245\textwidth]{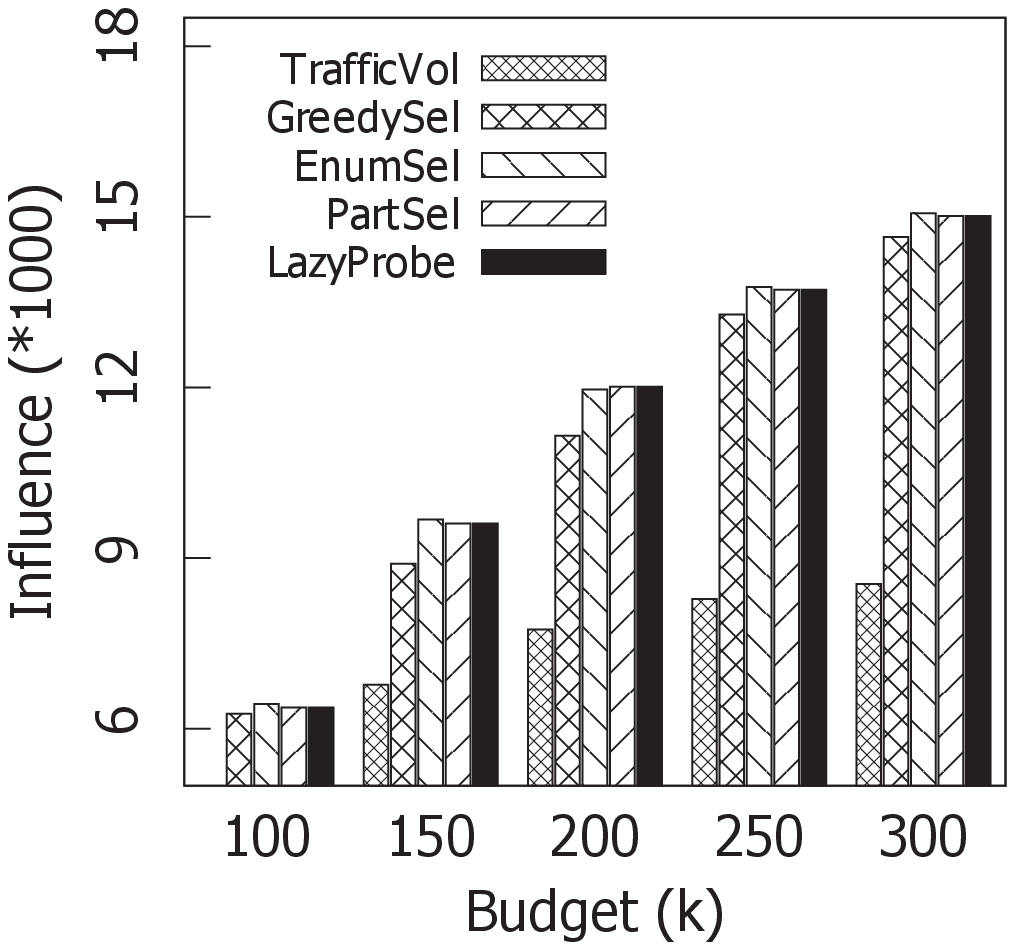}\label{fig:exp-s-L-ifl-NYC}}
	\subfloat[Efficiency (Varying $\budget$)]{\includegraphics[clip,width=0.254\textwidth]{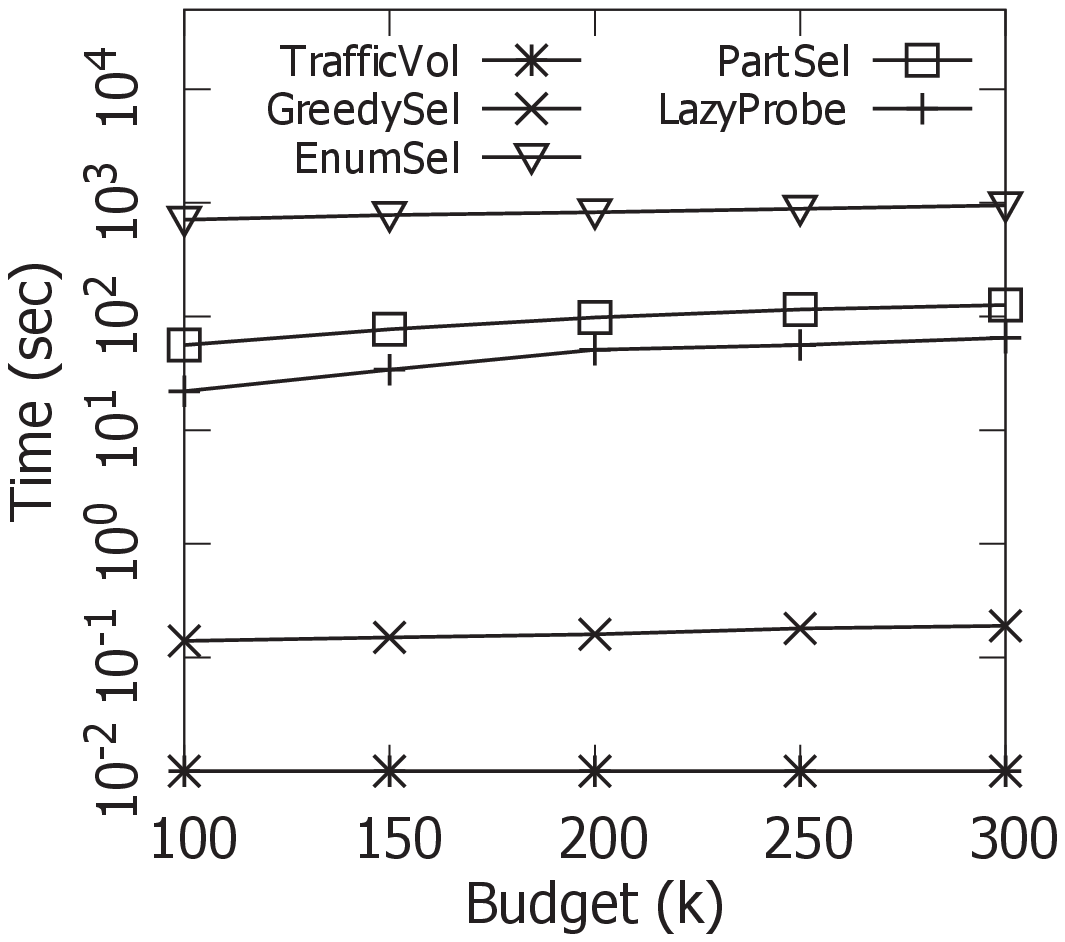}\label{fig:exp-s-L-eff-NYC}}
	\subfloat[Influence (Varying $|\td|$)]{\includegraphics[clip,width=0.245\textwidth]{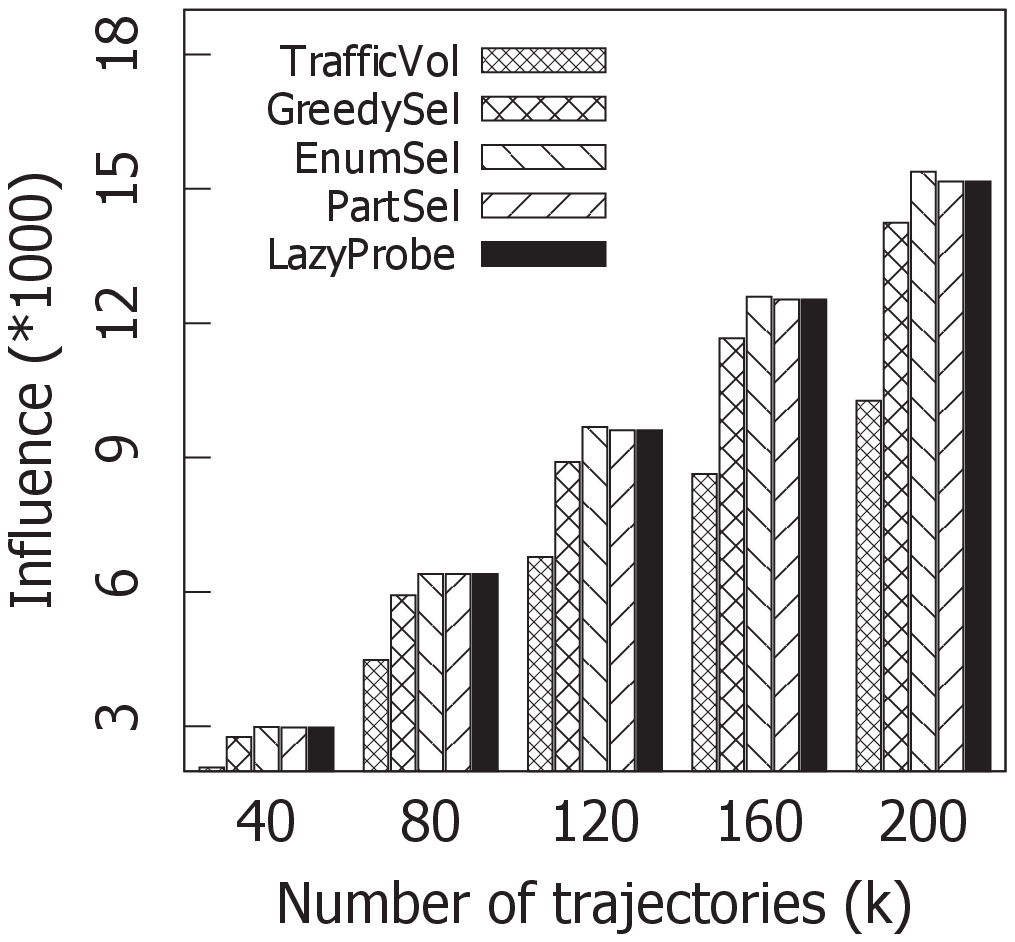}\label{fig:exp-s-L-ifl-LA}}
	\subfloat[Efficiency (Varying $|\td|$)]{\includegraphics[clip,width=0.254\textwidth]{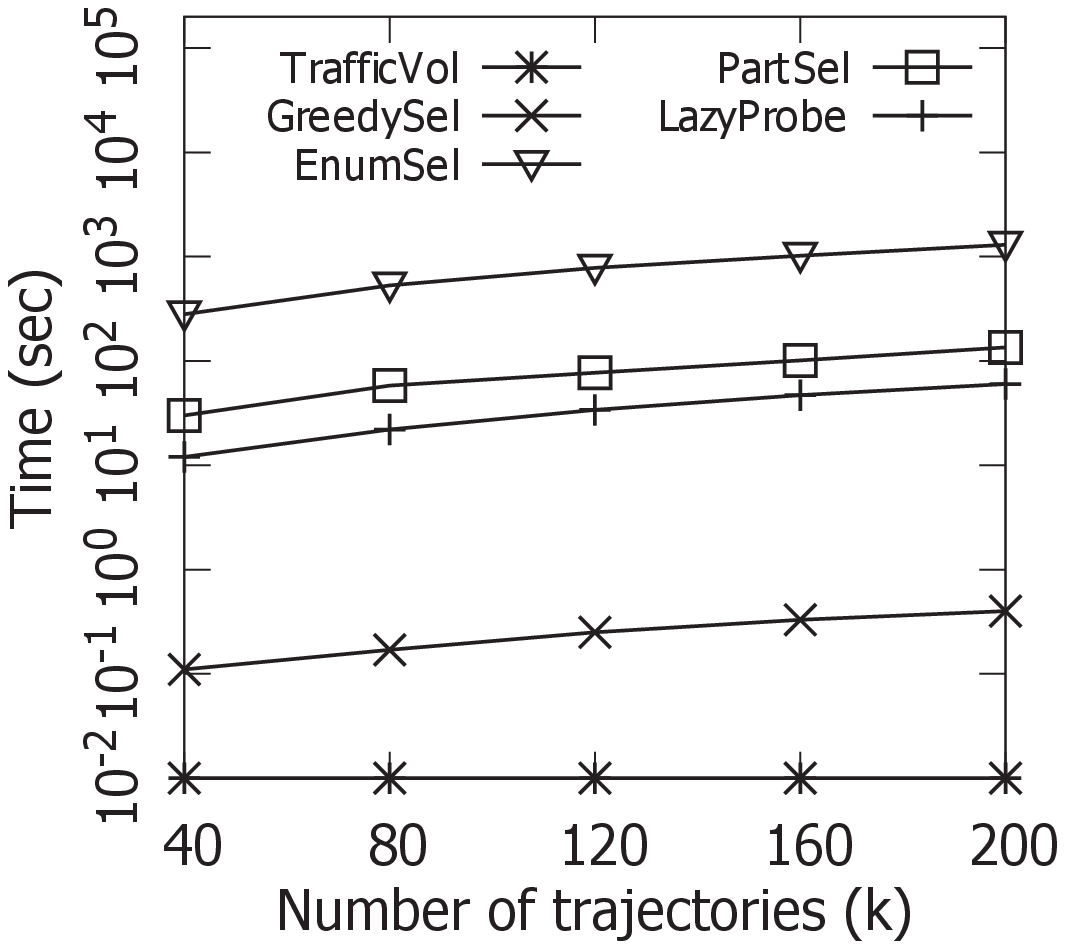}\label{fig:exp-s-L-eff-LA}}
	\vspace{-9pt}
	\caption{Testing \enumgreedy on a small NYC dataset ($|\ur|$=1000, $|\td|$=120k)}
	\label{fig:nyc_small_dataset}
\end{figure*}

\begin{figure*}[!tbh]
	\centering
	\subfloat[Influence (Alternative 1)]{\includegraphics[clip,width=0.245\textwidth]{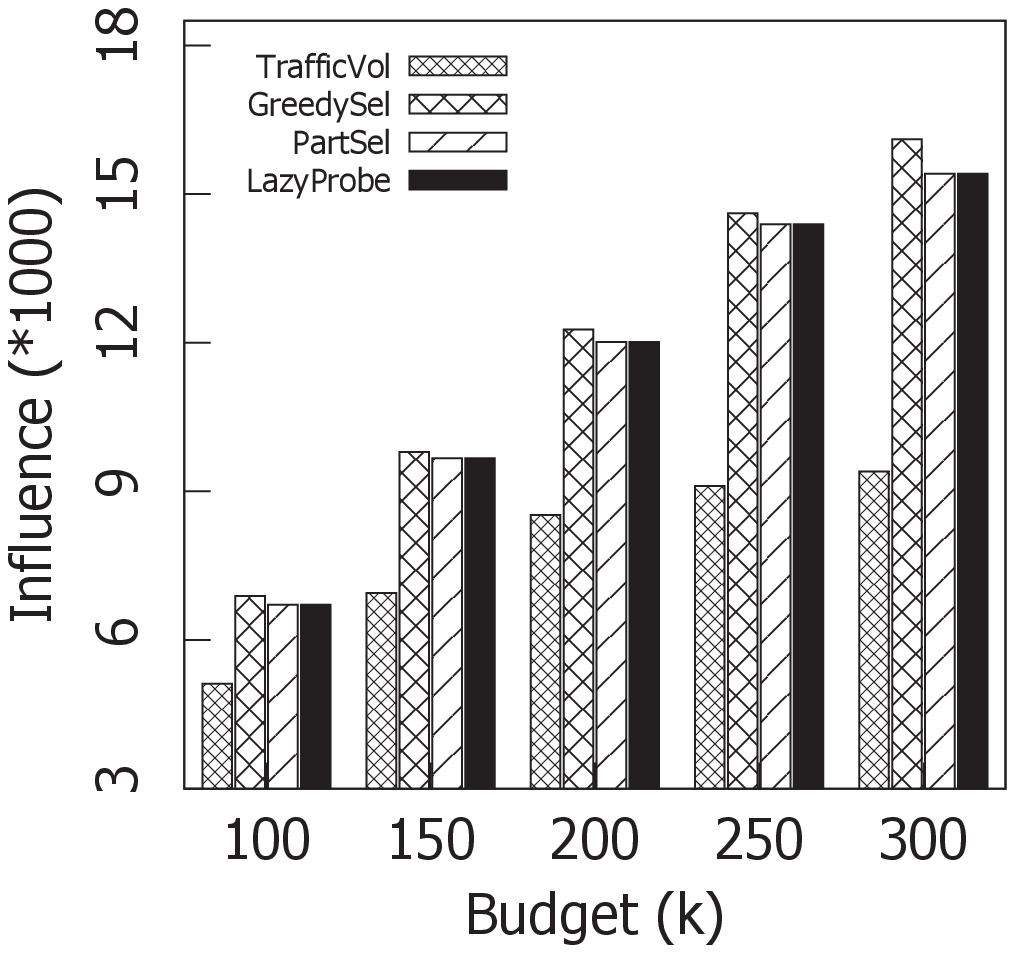}\label{fig:exp-theta2-ifl-LA}}
	\subfloat[Efficiency (Alternative 1)]{\includegraphics[clip,width=0.254\textwidth]{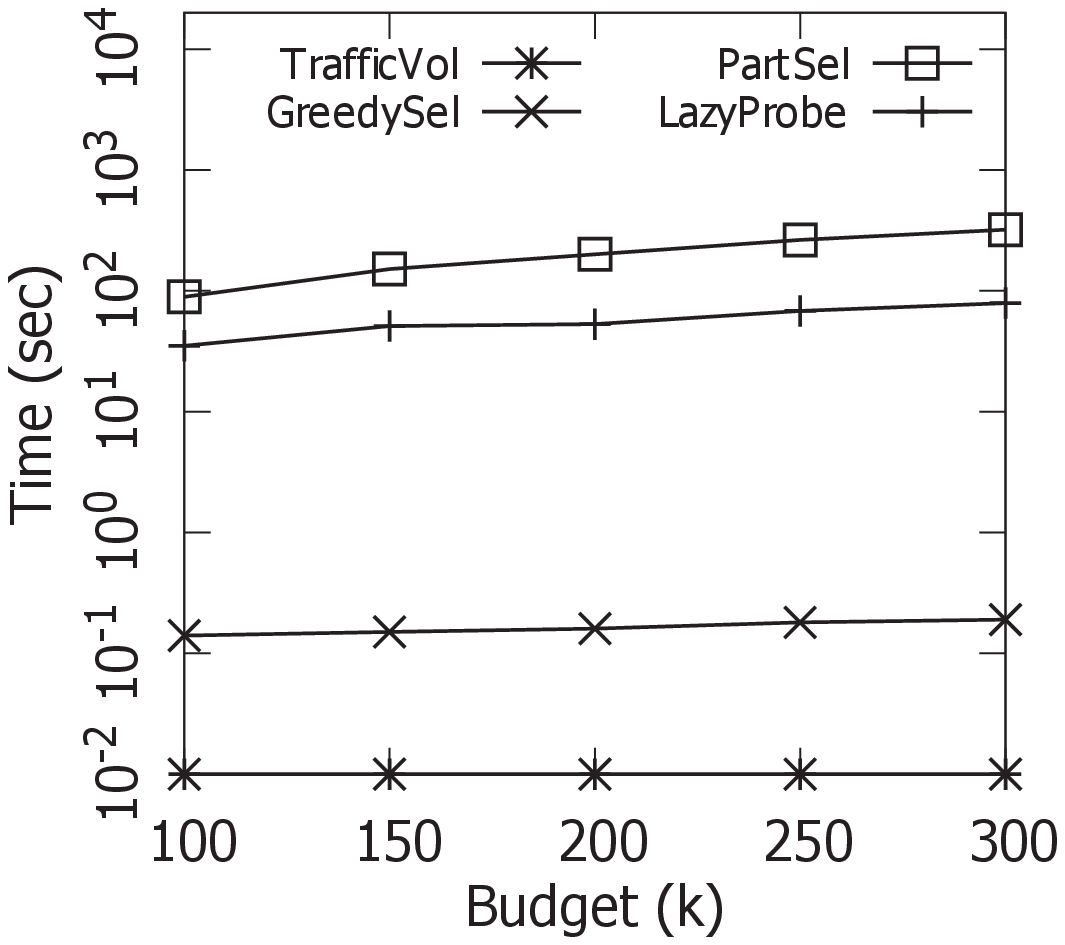}\label{fig:exp-theta2-eff-LA}}
	\subfloat[Influence (Alternative 2)]{\includegraphics[clip,width=0.245\textwidth]{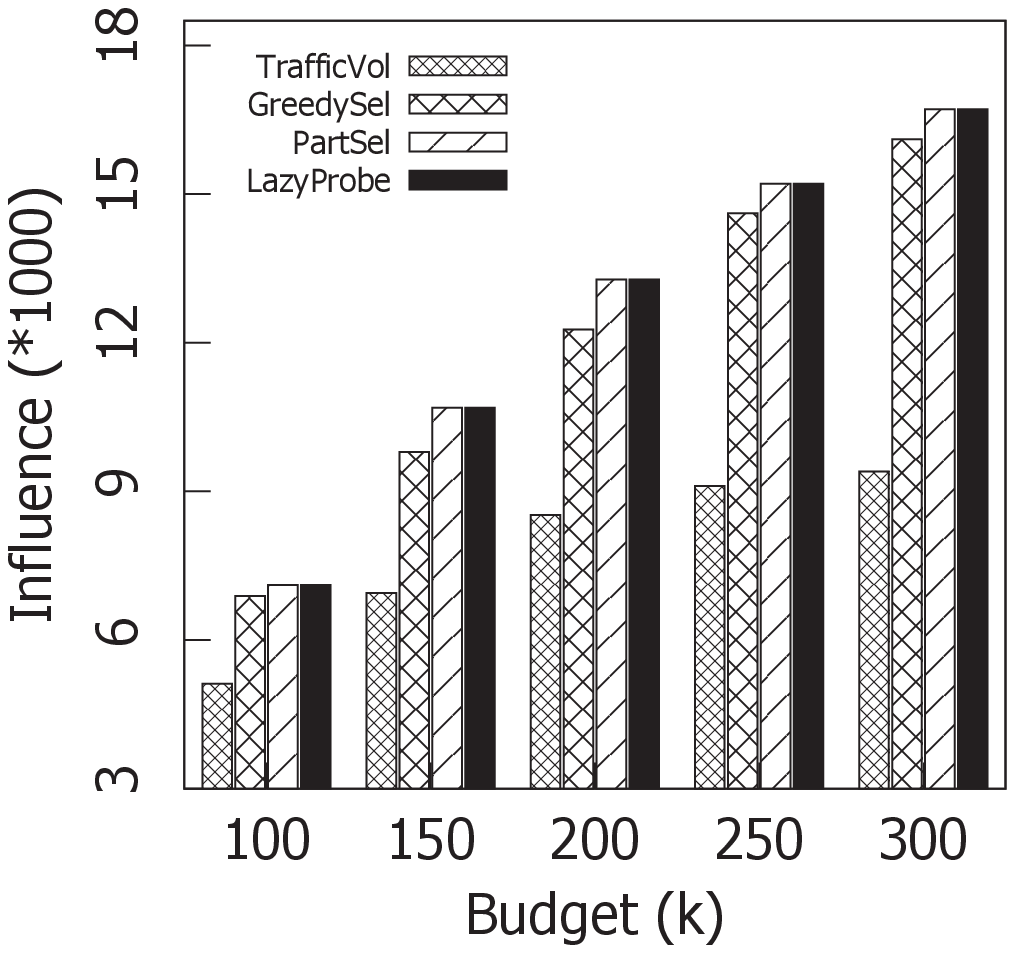}\label{fig:exp-theta1-ifl-NYC}}
	\subfloat[Efficiency (Alternative 2)]{\includegraphics[clip,width=0.254\textwidth]{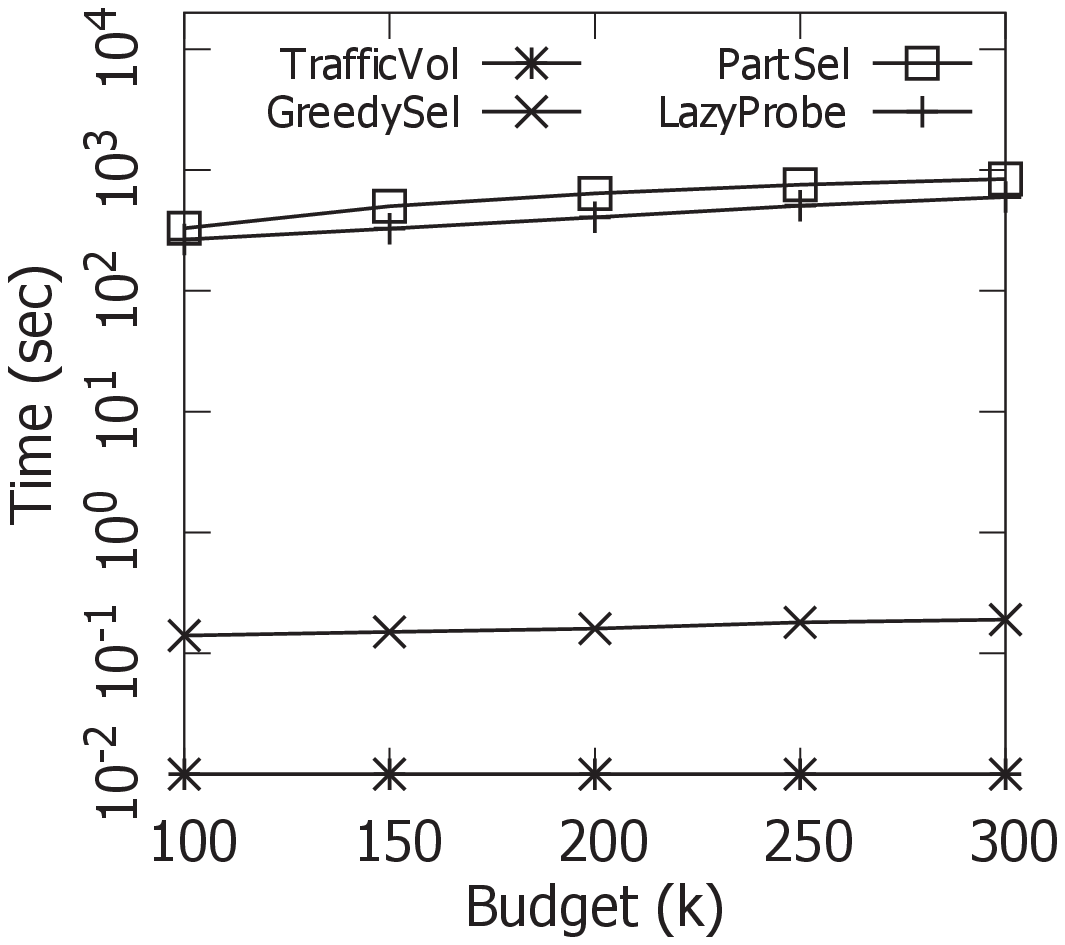}\label{fig:exp-theta1-eff-NYC}}
	\vspace{-9pt}
	\caption{The impact of two alternative of the overlap ratio $\ratio_{\svar{ij}}$ ($|\ur|$=1000, $|\td|$=120k)}
	\label{fig:theta_choice}
\end{figure*}

\subsubsection{Test on alternative choices of overlap ratio $\ratio_{\svar{ij}}$}\label{sec:2nd_probability}
Here we study how other two alternatives of the overlap ratio described in Section~\ref{sec:partition} affects the effectiveness and efficiency of all algorithms, and the result ($\theta=0.1$) w.r.t. the varying budget is presented in Figure~\ref{fig:theta_choice}. 

By comparing Alternative 1 with our choice, we find: (1) from Figure~\ref{fig:exp-theta2-ifl-LA} vs. Figure~\ref{fig:exp-L-ifl-NYC}, the \gre baseline consistently beats \psel and \bbsel under alternative 1, while it is the other way around under our choice. The reason is that this choice only restricts the overlap between clusters rather than billboards. Consequently, some billboards in different cluster should still have a relative high overlap. (2) from Figure~\ref{fig:exp-theta2-eff-LA} vs. Figure~\ref{fig:exp-L-eff-NYC}, the efficiency of all algorithms are almost the same for both choices. 

By comparing Alternative 2 with our choice, we make two observations. (1) From Figure~\ref{fig:exp-theta1-ifl-NYC} vs. Figure~\ref{fig:exp-L-ifl-NYC}, \psel and \bbsel consistently beat \gre under both cases. (2) From Figure~\ref{fig:exp-theta1-eff-NYC} vs. Figure~\ref{fig:exp-L-eff-NYC}, the efficiency of all algorithms under our choice is faster than those under alternative 2 by almost one order of magnitude. We interpret the results as, the partition condition of alternative 2 is too strict that $\ur$ cannot be divided into a set of small yet balanced clusters, thus a larger $|C_m|$ is incurred to increase the runtime.

\subsubsection{Experiment on an alternative choice of influence probability}

Recall Section~\ref{sec:pf} that the influence of a billboard $\bb_i$ to a trajectory $\tr_j$, $\pr(\bb_i, \tr_j)$, is defined. Here we conduct more experiments to test the impact of an alternative choice for the influence probability measurement as described in Section~\ref{sec:exp}. The alternative choice is: $\pr(\bb_i, \tr_j)=\bb_i.panelsize/(2*maxPanelSize)$ where $maxPanelSize$ is the size of the largest billboard in $\ur$, and we further normalize by 2 to avoid a too large probability, say 1.

The influence result of all algorithms on the NYC dataset is shown in Figure~\ref{fig:vary_p}. Recall our corresponding experiment of adopting choice 1 in Figure~\ref{fig:exp-L-ifl-NYC} and Figure~\ref{fig:exp-Tr-ifl-NYC}, we have the same observations: \psel and \bbsel outperform all the rest algorithms in influence. 
To summarize, our solutions are orthogonal to the choice of these metrics.


\section{Conclusion}

We studied the problem of trajectory-driven influential billboard placement: given a set of billboards $\ur$, a database of trajectories $\td$ and a budget $\budget$, the goal is to find a set of billboards within $\budget$ so that the placed ads can influence the largest number of trajectories. We showed that the problem is NP-hard, and first proposed a greedy method with enumeration technique. Then we exploited the locality property of the billboard influence and proposed a partition-based framework \psel to reduce the computation cost. 
Furthermore, we proposed a lazy probe method \bbsel to further prune billboards with low benefit/cost ratio, which significantly reduces the practical cost of \psel while achieving the same approximation ratio as \psel. Lastly we conducted experiments on real datasets to verify the efficiency, effectiveness and scalability of our method.

\begin{acks}
	Zhiyong Peng was supported by the Ministry of Science and Technology of China (2016YFB1000700), and National Key Research \& Development Program of China (No. 2018YFB1003400). Zhifeng Bao was supported by ARC (DP170102726, DP180102050), NSFC (61728204, 91646204), and was a recipient of Google Faculty Award. Guoliang Li was supported by the 973 Program of China (2015CB358700), NSFC (61632016, 61472198, 61521002, 61661166012) and TAL education. 
\end{acks}

\bibliographystyle{abbrvnat}
\bibliography{draft}



\end{document}